\documentclass[12pt]{article}
\usepackage[sort]{natbib}
\usepackage[margin=1in]{geometry}
\usepackage{times}
\usepackage[italic, noexclam, nopunctuation, noplus, nominus, noplusnominus, noequal, noparenthesis, nospecials, defaultnormal]{mathastext}
\usepackage{amsmath,amssymb,amsfonts, mathtools}
\usepackage{amsthm}
\usepackage{setspace}
\usepackage[labelfont=bf]{caption}
\usepackage{caption}
\usepackage{subcaption}
\usepackage{hyperref}
\usepackage{multirow}
\usepackage{longtable}
\usepackage[table]{xcolor}
\definecolor{maroon}{cmyk}{0,0.87,0.68,0.32}
\usepackage{hhline}
\usepackage{csquotes}
\usepackage{fancyhdr}
\usepackage{color}

\usepackage{natbib}

\newcommand{\orth}{\ensuremath{\perp\!\!\!\perp}}%
\newcommand{\indep}{\orth}%

\newcommand{\E}{\mathbb{E}}
\newcommand{\hE}{\widehat{\mathbb{E}}_n}

\newcommand{\R}{\mathbb{R}}

\newtheorem{definition}{Definition}
\numberwithin{definition}{section}

\numberwithin{axiom}{section}
\newtheorem{theorem}{Theorem}
\numberwithin{theorem}{section}
\newtheorem{assumption}{Assumption}
\numberwithin{assumption}{section}
\newtheorem{proposition}{Proposition}
\numberwithin{proposition}{section}

\numberwithin{corollary}{section}

\numberwithin{lemma}{section}
\newtheorem{example}{Example}
\numberwithin{example}{section}
\newtheorem{remark}{Remark}
\numberwithin{remark}{section}

\numberwithin{equation}{section}

\begin{document}

\author{Ganesh Karapakula\footnote{Website: \href{https://c-metrics.com}{https://c-metrics.com}. Email: \href{mailto:vgk22@cam.ac.uk}{vgk22@cam.ac.uk} or \href{mailto:gkarapakula@gmail.com}{gkarapakula@gmail.com}. I gratefully acknowledge financial support from the Janeway Institute, Faculty of Economics, The University of Cambridge.}}
\title{\LARGE \textbf{Stable Probability Weighting} \\ \vspace{2mm} \large Large-Sample and Finite-Sample Estimation and Inference Methods for Heterogeneous Causal Effects of Multivalued Treatments Under Limited Overlap}
\date{\normalsize 12 January 2023}

\maketitle

\vspace{-8mm}
\begin{abstract}
\noindent 
In this paper, I try to tame ``Basu's elephants'' (data with extreme selection on observables). I propose new practical large-sample and finite-sample methods for estimating and inferring heterogeneous causal effects (under unconfoundedness) in the empirically relevant context of limited overlap. I develop a general principle called ``Stable Probability Weighting'' (SPW) that can be used as an alternative to the widely used Inverse Probability Weighting (IPW) technique, which relies on strong overlap. I show that IPW (or its augmented version), when valid, is a special case of the more general SPW (or its doubly robust version), which adjusts for the extremeness of the conditional probabilities of the treatment states. The SPW principle can be implemented using several existing large-sample parametric, semiparametric, and nonparametric procedures for conditional moment models. In addition, I provide new finite-sample results that apply when unconfoundedness is plausible within fine strata. Since IPW estimation relies on the problematic reciprocal of the estimated propensity score, I develop a ``Finite-Sample Stable Probability Weighting'' (FPW) set-estimator that is unbiased in a sense. I also propose new finite-sample inference methods for testing a general class of weak null hypotheses. The associated computationally convenient methods, which can be used to construct valid confidence sets and to bound the finite-sample confidence distribution, are of independent interest. My large-sample and finite-sample frameworks extend to the setting of multivalued treatments.
\end{abstract}

% SHORTER ABSTRACT:
% In this paper, I try to tame ``Basu’s elephants'' (data with extreme selection on observables). For the empirically relevant setting of limited overlap, I propose new methods for estimating and inferring heterogeneous causal effects of multi-valued treatments under unconfoundedness (non/semi)parametrically. I develop a general principle called ``Stable Probability Weighting'' (SPW) that can be used as an alternative to the widely used Inverse Probability Weighting (IPW) technique, which relies on strong overlap. I show that IPW (or its augmented version), when valid, is a special case of the more general SPW (or its doubly robust version), which adjusts for the extremeness of the conditional probabilities of treatment states. I also provide an unbiased ``Finite-Sample Stable Probability Weighting'' (FPW) set-estimator of average effects within fine strata. Moreover, I propose computationally convenient finite-sample inference methods, which are of independent interest, for bounding confidence distributions and testing a general class of weak null hypotheses.

\onehalfspacing

\section{Introduction}
\label{section:introduction}

Inverse Probability Weighting (IPW), or its augmented doubly robust version, is now a widely used technique for estimating causal effects and for dealing with missing data. The technique is based on the influential work of \cite{horvitz1952generalization}, \cite{robins1994estimation}, \cite{robins1995semiparametric}, \cite{hahn1998role}, \cite{hirano2003efficient}, and several others. However, even in its early days, the method was not without its critics. \cite{basu1971essay} illustrates some practical issues with it using an amusing tale, in which a circus statistician gets sacked by the circus owner after proposing an unsound IPW estimator of the total mass of circus elephants (given by the mass of a randomly chosen elephant  multiplied by the reciprocal of an associated extreme selection probability). Unfortunately, ``Basu's elephants'' are not confined to fictional circuses; they manifest themselves in many policy-relevant empirical settings, and extreme weights based on IPW can result in questionable and unstable estimates of causal effects.\footnote{See, e.g., \cite{kang2007demystifying, robins2007comment, tsiatis2007comment, kang2007rejoinder, frolich2004finite, khan2010irregular, ma2020robust, heiler2021valid, sasaki2022estimation, ma2022testing, d2021overlap, crump2006moving, li2019addressing, petersen2012diagnosing}, and the references therein.} This problem with IPW shows up in many important datasets across various fields, including the health sciences\footnote{See, e.g., \cite{crump2009dealing, scharfstein1999adjusting, heiler2021valid}.} and the social sciences.\footnote{See, e.g., \cite{ma2020robust, huber2013performance, busso2014new}.} In fact, the pioneers of the modern form of (augmented) IPW have themselves repeatedly issued warnings regarding highly variable, extreme inverse probability weights.\footnote{See, e.g., \cite{robins1995semiparametric, scharfstein1999adjusting, robins2000inference, robins2007comment}.}

In this paper, I provide new large-sample and finite-sample methods, which serve as alternatives to the (augmented) IPW framework, for estimating and inferring heterogeneous causal effects under unconfoundedness, i.e., the standard assumption of strongly ignorable treatment assignment \citep{rosenbaum1983central}. Even though I primarily work with binary treatment variables to avoid cumbersome notation, I show that my framework readily extends to the setting of multivalued treatments with generalized propensity scores \citep{imbens2000role}. While studying the existing literature on limited overlap (extreme selection on observables), I have come across only papers that deal with the average treatment effect (ATE) parameter in this context (to various extents). To the best of my knowledge, this paper is the first to propose estimation and inference methods for higher-dimensional parameters that characterize heterogeneity in the (possibly multivalued) treatment effects when limited overlap is empirically important and non-negligible.\footnote{As far as I am aware, the existing literature (from a frequentist or machine learning or policy learning perspective) on heterogeneous treatment effect estimation and inference relies on strong overlap. See \cite{wang2020debiased, zhao2019selective, zhao2022selective, athey2019generalized, powers2018some, nie2021quasi, kunzel2019metalearners, wager2018estimation, athey2016recursive, kennedy2022towards, kennedy2022minimax, semenova2021debiased, imai2013estimating,belloni2017program,singh2020kernel,chernozhukov2018generic,knaus2021machine, oprescu2019orthogonal,semenova2022estimation,crump2008nonparametric,sant2021nonparametric, athey2021policy, kitagawa2018should, kitagawa2022stochastic, sun2021empirical,abrevaya2015estimating,fan2022estimation,lee2017doubly}. There is some discussion on limited overlap in the Bayesian literature on heterogeneous treatment effect estimation \citep[see, e.g.,][]{li2022bayesian}. \cite{hill2013assessing} propose some heuristic strategies based on Bayesian additive regression trees to deal with poor overlap. \cite{hahn2020bayesian} propose finding the regions with limited overlap and then trimming them or to use the spline-based extrapolation procedures suggested by \cite{nethery2019estimating}. The existing literature on causal estimation and inference in the case of multi-valued treatments also relies on the assumption of strong overlap. See, \cite{ai2021unified, ai2022estimation, bugni2019inference, cattaneo2010efficient, kennedy2017non,farrell2015robust,su2019non,colangelo2020double,yang2016propensity, zhang2022towards}. A very recent exception is the work of \cite{dias2021inference}, who allow the proportion of treated units to diminish asymptotically. They provide some interesting results in a specialized setting with discrete covariates and some shape restrictions, especially a ``conditional rank invariance restriction [that] is analogous to the control function framework for quantile regression.'' They also provide valid inferential methods for a narrow class of null hypotheses. Since \cite{dias2021inference} allow for dynamic treatment effects, their methods are very useful in some panel data settings, but their model setup is specialized and different from the standard version of unconfoundedness that I use.}

In the setting of \cite{rosenbaum1983central}, if the assumption of strong overlap (also known as strict overlap, i.e., propensity scores being bounded away from zero and one) does not hold, then inverse weight estimation is problematic; specifically, \cite{khan2010irregular} show that IPW may not lead to an asymptotically normal regular $\sqrt{n}$-consistent estimator of even just the ATE, which is a simpler parameter than the heterogeneous treatment effect parameters. Depending on how the propensity scores are distributed, the IPW estimator may be non-Gaussian (e.g., asymmetric L\'evy stable) asymptotically \citep{ma2020robust}. Although strong overlap is often invoked in the literature \citep[see, e.g.,][]{hirano2003efficient, chernozhukov2022automatic} or is implicit in other assumptions \citep{d2021overlap}, it rules out many simple, plausible models (e.g., some probit models) for the treatment assignment variable and also contradicts many empirically observed propensity score distribution shapes \citep{heiler2021valid, ma2022testing,lei2021distribution}.

In the presence of limited overlap, some researchers suggest focusing attention instead on alternative estimands and parameters that can be estimated efficiently using heavy trimming or winsorization of the propensity scores \citep{crump2009dealing, zhou2020propensity} or using other weighting procedures.\footnote{See \cite{li2018balancing, graham2012inverse, zubizarreta2015stable, wang2020minimal, athey2018approximate, hainmueller2012entropy, zhao2019covariate, imai2014covariate, imai2015robust, robins2000marginal, ai2021unified, yang2018asymptotic,wong2018kernel,ning2020robust,hirshberg2021augmented,ben2022using,wang2020debiased,matsouaka2022overlap,khan2021adaptive, chen2021robust,li2019propensity}. Some methods, such as those of \cite{chen2008semiparametric} and \cite{hirshberg2021augmented}, make weaker assumptions than strong overlap but still restrict limited overlap. For example, \cite{hirshberg2021augmented} require the first moment of the inverse propensity score to be bounded, but this testable assumption may lack empirical support \citep{ma2022testing}.} \cite{lee2021bounding} propose partial identification methods. \cite{ma2020robust} develop bias-correction and robust inference procedures for trimmed IPW estimators.\footnote{\cite{chaudhuri2014heavy, yang2018asymptotic, sasaki2022estimation, khan2022uniform,khan2022heteroscedasticity} also propose related trimming approaches for ATE estimation and inference under limited overlap.} \cite{heiler2021valid} propose an empirically appealing method (a modified $m$-out-of-$n$ bootstrap procedure that does not rely on trimming) for robust inference on the (augmented) IPW estimators of the ATE in the presence of limited overlap. However, \cite{heiler2021valid} restrict ``the occurrence of extreme inverse probability weighted potential outcomes or conditional mean errors.'' Although the focus of my paper is on the heterogeneity of causal effects, my proposed methods also offer an alternative to the aforementioned existing large-sample methods for inference on the ATE.

In Section \ref{section:setup}, I formalize the observational setup and the parameters of interest that are used throughout the paper. I do not impose the empirically restrictive assumption of strict overlap, but the setting is otherwise quite standard in the literature on causal inference under unconfoundedness. In Section \ref{section:alternatives}, I review issues with the popular (augmented) inverse probability weighting techniques under weak overlap. I illustrate these issues by considering a very simple setting where the propensity scores are known and a linear model describes heterogeneity in average causal effects. This particular setting is designed to help us focus on the core issues related to limited overlap, paving the way for a more thorough consideration of the problem from semi/non-parametric perspectives in later sections. Even in the simple setting of Section \ref{section:alternatives}, IPW approaches can have undesirable statistical properties when overlap is not strong. However, I show that there exist simple alternatives: ``Non-Inverse Probability Weighting'' (NPW) estimators, which are consistent and asymptotically normal while also being robust to limited overlap. I define a general estimator class that nests IPW and NPW estimators and show that they can be expressed as weighted versions of one another. I also discuss NPW-based estimation of the ATE.

My estimators in Section \ref{section:alternatives} are based on a general principle called ``Stable Probability Weighting'' (SPW), which I develop more fully in the next section, for learning about the conditional average treatment effect (CATE) function that characterizes the heterogeneity in the average causal effects. In Section \ref{section:spw}, I present the relevant conditional moment restrictions that enable estimation and inference procedures that are robust to limited overlap and provide several examples. (If strict overlap holds, my general framework nests IPW as a special case.) I also show how to augment the basic conditional moment model in order to achieve double robustness (i.e., consistent estimation despite misspecification of one of the nuisance parameters) and compare my approach with some moment-based methods that are widely used in the existing literature. I then generalize the SPW approach for analyzing multivalued treatments and distributional or quantile treatment effects. I also discuss SPW-based estimation and inference from parametric, semiparametric, nonparametric, machine learning, and policy learning perspectives. Thanks to the currently rich econometric literature on conditional moment models, I do not have to reinvent the wheel for large-sample estimation and inference. Next, I proceed to develop finite-sample methods in Section \ref{section:fpw}.

Even under strict overlap, the standard IPW estimators may have poor finite-sample behavior. But under limited overlap, the problem is much worse, and those estimators typically have very undesirable small-sample properties \citep{armstrong2021finite, rothe2017robust, hong2020inference}. \cite{hong2020inference} provide restrictions on limited overlap that allow the use of standard methods for valid asymptotic inference on the finite-population version of the ATE when propensity scores degenerate to zero asymptotically. In a setting with finite strata, \cite{rothe2017robust} assumes normally distributed potential outcomes and proposes valid inference for the empirical ATE (EATE), which is the sample average of the expected treatment effect conditional on the covariates, by turning the problem into a general version of the Behrens--Fisher problem.\footnote{An alternative to the approach of \cite{rothe2017robust} is that of \cite{dias2021inference}, who also use discrete covariates. Rather than assuming normally distributed potential outcomes like \cite{rothe2017robust} does, \cite{dias2021inference} impose some shape restrictions, especially a ``conditional rank invariance restriction [that] is analogous to the control function framework for quantile regression.'' My finite-sample methods in this paper do not impose either set of assumptions.} \cite{armstrong2021finite} also develop some finite-sample inference methods for the EATE under the assumption of normally distributed errors with known variances for potential outcomes whose conditional means lie within a known convex function class. 

Both \cite{rothe2017robust} and \cite{armstrong2021finite} acknowledge the restrictiveness of the normality assumption for finite-sample inference.\footnote{\cite{rothe2017robust} says, ``We work with normality since without some restriction of this type, it would seem impossible to obtain meaningful theoretical statements about the distribution of (studentized) average outcomes in covariate-treatment cells with very few observations,'' in addition to saying that the normality assumption ``is clearly restrictive; but without imposing some additional structure it would seem impossible to conduct valid inference in the presence of small groups.''} They also use fixed designs not only for covariates but also treatment status. While it is standard in the literature on finite-sample inference to condition on covariates for testing purposes,\footnote{See, e.g., \cite{lehmann1993fisher, lehmann2005testing,zhang2022randomization}.} it is not standard to also condition on the treatment status. Conditioning on the treatment status may be even more problematic in the limited overlap setting, because such an approach may understate actual uncertainty in the resulting estimates.\footnote{In the limited overlap scenario, very few treated (or untreated) observations are used to estimate the conditional means of the relevant potential outcomes. However, the number of those few treated (or untreated) observations is itself a random variable. Conditioning on the treatment statuses of the observations ignores this source of randomness, potentially leading to underreported uncertainty in the estimates of the conditional means.}

The finite-sample framework of my paper goes beyond the EATE parameter, which is the focus of \cite{rothe2017robust}, \cite{armstrong2021finite}, and \cite{hong2020inference}.\footnote{To be more precise, \cite{rothe2017robust} and \cite{armstrong2021finite} focus on the EATE, and \cite{hong2020inference} focus on the finite-population ATE, which is related to but conceptually different from EATE. Another important note is that \cite{rothe2017robust} and \cite{armstrong2021finite} use different terms for the same EATE parameter. \cite{rothe2017robust} calls it the ``sample average treatment effect'' (SATE), but this term is used for a different parameter in much of the literature. \cite{armstrong2021finite} use the term ``conditional average treatment effect'' (CATE) to refer to the EATE parameter, but their usage is very nonstandard in the literature. They themselves say, ``We note that the terminology varies in the literature. Some papers call this object the sample average treatment effect (SATE); other papers use the terms CATE and SATE for different objects entirely.'' To avoid all this confusion, I ensure that the names and abbreviations of some common parameters used in this paper, as defined in Section \ref{section:setup}, are the same as those used in the vast majority of the causal inference literature.} I propose finite-sample estimation and inference methods for each term comprising the EATE, which is the sample average of heterogeneous treatment effect means, rather than just the aggregate EATE object. Unlike \cite{armstrong2021finite}, I do not impose assumptions on the smoothness or the shape of the conditional means of the outcomes (for finite-sample results), and so my focus is on finite-sample unbiased estimation and valid finite-sample inference rather than ``optimal'' procedures. To allow for fine partitions of the covariate space (e.g., using unsupervised feature learning), I use the strata setup of \cite{rothe2017robust}, \cite{hong2020inference}, and \cite{dias2021inference} but with two generalizations. Unlike their setups, I allow the number of strata to potentially grow with the overall sample size, while also allowing overlap to be arbitrarily low within each stratum. I do not impose the restrictive normality assumptions used by \cite{rothe2017robust} and \cite{armstrong2021finite}. In addition, unlike these papers, I do not condition on the treatment status variables, which are the source of the fundamental missing (counterfactual) data problem. I instead use a design-based finite-sample approach\footnote{See, e.g., \cite{lehmann2005testing, wu2021randomization, abadie2020sampling, imbens2021causal, ding2016randomization, young2019channeling, fisher1925statistical, fisher1935design,zhang2022randomization,xie2013confidence}.} that I believe is simpler, broader, and more transparent.

Even though the heterogeneous average treatment effects are point-identified in the above setting, it is typically not straightforward to obtain unbiased point-estimates of the objects of interest. This difficulty is a result of some practical statistical issues with the reciprocal of the estimated propensity score. Thus, in Section \ref{section:fpw}, I develop a ``Finite-Sample Stable Probability Weighting'' (FPW) set-estimator for analyzing heterogeneous average effects of multivalued treatments. I define a notion of finite-sample unbiasedness for set-estimators (a generalization of the usual concept for point-estimators), and I show that the property holds for the FPW set-estimator. An interesting feature of the FPW set-estimator is that it actually reduces to a point-estimator for many practical purposes and thus serves as a simpler alternative to other available set-estimators that are wider \citep[see, e.g.,][]{lee2021bounding} in the (high-dimensional) strata setting.

After discussing finite-sample-unbiased estimation of average effects within strata, I propose new finite-sample inference methods for testing a general class of ``weak'' null hypotheses, which specify a hypothesized value for average effects. Except for the case of binary outcomes,\footnote{See \cite{rigdon2015randomization} and \cite{li2016exact}. See \cite{caughey2021randomization} for another exception.} currently there only exist asymptotically robust tests\footnote{See \cite{wu2021randomization,chung2013exact,chung2016multivariate}.} of general weak null hypotheses that have finite-sample validity under some ``sharp'' null hypotheses, which restrictively specify hypothesized values of counterfactual outcomes for all the observations. I thus tackle the question of how to conduct finite-sample tests of some general weak null hypotheses. I also use partial identification for this purpose. The finite-sample inference methods that I propose are computationally convenient, and they can be used to construct valid confidence sets and to bound the finite-sample confidence distribution.\footnote{Since sharp null hypotheses are a subset of weak null hypotheses, my proposals also offer a computationally fast way to compute the finite-sample confidence distributions when sharp null hypotheses are of interest.} These contributions (to the literature on finite-sample inference) are of independent interest. An appealing aspect of my finite-sample methods is that they are valid under weak assumptions and are yet simple conceptually and computationally.

\section{Standard Observational Setup and Parameters of Interest}
\label{section:setup}

In this section, I formalize the observational setup and parameters of interest that are used throughout the paper. The parameters and the associated empirical setting that I consider are quite standard in the literature on causal inference under unconfoundedness \citep[see, e.g.,][]{imbens2000role,hirano2004propensity, rosenbaum1983central,imbens2015causal}. However, I do not impose the restrictive assumption of strict overlap (boundedness of the reciprocals of propensity scores), which is crucial for many of the results in the existing literature. Thus, the following setup is more general than the standard one that is usually used \citep[see, e.g.,][]{hirano2003efficient, kennedy2022towards, kennedy2022minimax, ai2021unified, kennedy2017non,semenova2021debiased}.
\pagebreak

Let $\mathbb{W} \subset \R$ be a bounded set of possible treatment states, and let $Y^*_w$ be the random variable representing the potential outcome under the treatment state $w \in \mathbb{W}$. However, the fundamental problem of causal inference is that the random element $(Y^*_w)_{w\in \mathbb{W}}$ is not completely observable; only one component of it is not latent. The observable component is $Y = Y^*_W$, where $W$ is itself a random variable on $\mathbb{W}$ so that $W$ represents the treatment status. When $\mathbb{W} = \{0,1\}$ and $W$ is the associated binary random variable, $Y = W\,Y^*_1+(1-W)\,Y^*_0$. More generally, when $\mathbb{W}$ is discrete, $Y = \sum_{w \in \mathbb{W}} \mathbb{I}\{W = w\}\,Y^*_w$. When $\mathbb{W}$ is an open interval, $Y = \int_{\mathbb{W}}\, Y^*_{\tilde{w}}\,\delta(W - \tilde{w})\,d\tilde{w} = Y^*_W$, where $\delta(\cdot)$ is the Dirac delta function. The random elements $(Y^*_w)_{w \in \mathbb{W}}$ and $W$ may all depend on characteristics or covariates $X$, where $X$ is a random vector on a subset $\mathbb{X}$ of some Euclidean space. The only observable random elements in this setting are $(Y, X, W)$, but a lot can be learned from them. The relationship between $W$ and $X$ is characterized by an object called the propensity score, which plays a crucial role in causal analysis \citep{rosenbaum1983central}. Building on \cite{imbens2000role} and \cite{hirano2004propensity}, the propensity score can be generally defined as follows.

\begin{definition}[Propensity Score]
    \label{definition:ps}
    For all $x \in \mathbb{X}$, let $\tilde{p}(\cdot\,|\,x)$ be the Radon--Nikodym derivative of the probability measure induced by $W$ conditional on $X = x$ with respect to the relevant measure. Then, the (generalized) propensity score is given by $p(w, x) = \tilde{p}(w\,|\,x)$ for all $(w,x) \in \mathbb{W} \times \mathbb{X}$.
\end{definition}

\begin{remark}
    When $\mathbb{W}$ is discrete, $p(x; w) = \mathbb{P}\{W = w \,|\,X = x\} = \E[\,\mathbb{I}\{W = w\}\,|\,X = x]$ for all $(w,x) \in \mathbb{W} \times \mathbb{X}$, according to the above definition. However, if $W \,|\,X = x$ (i.e., $W$ conditional on $X = x$) is continuously distributed on an open interval $\mathbb{W} \subset \R$ for all $x \in \mathbb{X}$, then $p(w, x)$ is a conditional probability density function, which may be written using the Dirac delta $\delta(\cdot)$ notation as $p(w, x) = \E[\,\delta(W - w) \,|\,X = x]$, since $\int_{\mathbb{W}}\, \delta(\tilde{w} - w) \, p(\tilde{w}, x) \, d\tilde{w} = p(w, x)$ $\,\,\forall\,(w,x) \in \mathbb{W} \times \mathbb{X}$.
\end{remark}

Note that Definition \ref{definition:ps} is applicable even if $W$ is not a purely discrete or a purely continuous random variable. For example, if $W$ is a mixed continuous and discrete random variable, then $\tilde{p}(\cdot \,|\,x)$ in the above definition is the Radon--Nikodym derivative of the probability measure (induced by $W \,|\,X = x$) with respect to a combined measure, which equals the Lebesgue measure plus another measure on $\R$ for which the measure of any Borel set is equal to the number of integers on that Borel set. Definition \ref{definition:ps} can be applied even in a setting where $W$ is a random vector.

It is relatively less difficult to learn about the marginal distributions of the potential outcomes (under some assumptions) than about the joint distribution of the potential outcomes. Thus, it is useful to focus on parameters such as the ``Conditional Average Response'' (CAR) function and the `Conditional Average Contrast'' (CAC), which are defined as follows.

\begin{definition}[CAR: Conditional Average Response]
    \label{definition:car} CAR function $\mu: \mathbb{W} \times \mathbb{X} \to \R$ is given by
    $$\mu(w,x) \equiv \mu_w(x) = \mathbb{E}[Y^*_w \mid X = x].$$
\end{definition}

\begin{definition}[CAC: Conditional Average Contrast]
    \label{definition:cac} CAC function $\theta[\widetilde{\mathbb{W}}, \kappa]: \mathbb{X} \to \R$ is given by
    $$\textstyle \theta[\widetilde{\mathbb{W}}, \kappa](x) = \sum_{w \,\in\, \widetilde{\mathbb{W}}} \,\,\kappa_w\,\mu_w(x)$$
    when $\widetilde{\mathbb{W}}$ is a finite subset of $\mathbb{W}$ and $\kappa \equiv (\kappa_w)_{w\,\in\,\mathbb{W}}$ is a $|\widetilde{\mathbb{W}}|$-dimensional vector of real constants. In addition, the Average Contrast $\bar{\theta}(\widetilde{\mathbb{W}}, \kappa)$ is a real number given by $\,\,\bar{\theta}(\widetilde{\mathbb{W}}, \kappa) = \mathbb{E}[\theta[\widetilde{\mathbb{W}}, \kappa](X)]$.
\end{definition}

\begin{remark}
    When $\widetilde{\mathbb{W}} = \mathbb{W} = \{0,1\}$ and $\kappa_w = 2\,w - 1$, CAC equals $(1)\,\mu(1,x) + (-1)\,\mu(0,x)$. See Definition \ref{definition:cate}. When $\mathbb{W}$ is a bounded open interval on $\mathbb{R}$, an extension of the notion of CAC is given by $\,\,\widetilde{\theta}[\kappa](x) = \int_\mathbb{W}\,\kappa(w)\,\mu(w,x)\,dw\,\,$ for all $x \in \mathbb{X}$ when $\kappa: \mathbb{W} \to \R$ is a bounded function.
\end{remark} 

In addition to CAR and CAC, one may also be interested in parameters that depend on the conditional marginal distributions of the potential outcomes. For this purpose, it is useful to define the ``Conditional Response Distribution'' (CRD) as follows.

\begin{definition}[CRD: Conditional Response Distribution]
    \label{definition:crd} CRD function $\varrho: \mathbb{R} \times \mathbb{W} \times \mathbb{X} \to [0, 1]$ is a conditional cumulative distribution function given by
    $$\varrho(u; w,x) \equiv \varrho_w(u; x) = \mathbb{P}\{Y^*_w \leq u \mid X = x \} = \mathbb{E}[\,\mathbb{I}\{Y^*_w \leq u \}\mid X = x] = \E[Y^{*,u}_w \mid X = x],$$
    which is equivalent to the CAR of the transformed potential outcome $Y^{*,u}_w = \mathbb{I}\{Y^*_w \leq u\}$.
\end{definition}

It is difficult to identify the above parameters (CAR, CAC, and CRD) without further assumptions. To make parameter identification feasible, a large strand of literature on causal inference---as reviewed by \cite{imbens2004nonparametric}, \cite{imbens2015causal}, and \cite{wager2020causal}---assumes unconfoundedness (also known as selection on observables, conditional exogeneity, or strongly ignorable treatment assignment), population overlap (i.e., at least weak overlap), and SUTVA (stable unit treatment value assumption). These notions are formalized in Assumptions \ref{assumption:unconfoundedness}--\ref{assumption:pop_overlap} below.

\begin{assumption}[Unconfoundedness]
    \label{assumption:unconfoundedness}
    For all $w \in \mathbb{W}$, $\,\,Y^*_w \indep W \, | \, X$.
\end{assumption}

\begin{assumption}[Stable Unit Treatment Value Assumption]
    \label{assumption:sutva}
    $Y$ satisfies $\,\mathbb{P}\{ Y = Y^*_W \} = 1$.
\end{assumption}

\begin{assumption}[Population Overlap]
    \label{assumption:pop_overlap}
    For all $(w, x) \in \mathbb{W} \times \mathbb{X}$, $\,\, 0 < p(w, x) < \overline{p}$ for some $\overline{p} > 0$.
\end{assumption}

The above assumptions are sufficient to identify CAR as follows for all $(w, x)\in \mathbb{W} \times \mathbb{X}$: $$\mu_w(x) = \mu(w, x) = \E[Y^*_w \mid X = x] = \E[Y^*_W \mid W = w, X = x]  = \E[Y \mid W = w,\, X = x].$$
Note that the conditioning event in the above expression is $\{W = w,\, X = x\}$. Assumption \ref{assumption:pop_overlap}, i.e., $p(w, x) > 0$, ensures that $\{W = w,\, X = x\}$ is not a null event. It is also possible to identify CAR using the conditioning event $\{X = x\}$. For example, when $\mathbb{W}$ is discrete, CAR can be identified as follows for all $(w, x)\in \mathbb{W} \times \mathbb{X}$:
$$\frac{\E[\,\mathbb{I}\{W = w\}\,Y \mid X = x]}{\E[\,\mathbb{I}\{W = w\}\mid X = x]} = \frac{\E[\,\mathbb{I}\{W = w\}\,Y^*_w \mid X = x]}{p(w, x)} = \frac{p(w, x)\,\E[Y^*_w \mid X = x]}{p(w, x)} = \mu_w(x).$$
When $\mathbb{W}$ is an interval, if $K_h(W - w) \equiv h^{-1}\,K\big(\frac{W - w}{h}\big)$, where $K(\cdot)$ is a kernel and $h > 0$, then $$\mathrm{lim}_{\,h \,\downarrow\, 0}\frac{\E[\,K_h(W-w)\,Y \mid X = x]}{\E\,[K_h(W-w) \mid X = x]} = \frac{\int_\mathbb{W}\,\delta(\tilde{w} - w)\,\mu(\tilde{w}, x)\,p(\tilde{w}, x)\,d\tilde{w}}{\E[\delta(W-w) \mid X = x]} = \frac{\mu(w, x)\,p(w, x)}{p(w, x)} = \mu_w(x).$$
Thus, the positivity of the propensity score plays a key role in the identification of CAR. However, note that Assumption \ref{assumption:pop_overlap} is not the same as the more restrictive condition of strict (or strong) overlap defined below. Although a vast number of results in the existing literature depend on strict overlap, it is a testable assumption that may lack empirical support \citep{ma2022testing,lei2021distribution}. In this paper, I only assume population overlap (Assumption \ref{assumption:pop_overlap}), which allows for arbitrarily low propensity scores, i.e., ``weak\,/\,limited overlap,'' rather than the following condition.

\begin{definition}[Strict\,/\,Strong Overlap]
    \label{definition:strict_overlap}
    There exist $(\underline{p}^*, \overline{p}^*) \in \R_{>\,0}^2$ and a function $\underline{p}: \mathbb{W} \to \R_{>\,0}$ such that $0 < \underline{p}^* < \underline{p}(w) < p(w, x) < \overline{p}^*$ for all $(w, x) \in \mathbb{W} \times \mathbb{X}$. Thus, strict\,/\,strong overlap holds when the (generalized) propensity score function is uniformly bounded away from zero.
\end{definition}

Since CAR is identified under Assumptions \ref{assumption:unconfoundedness}--\ref{assumption:pop_overlap}, CAC and CRD are also identified because they both can be expressed using different forms of CAR. Although these three assumptions are adequate for identifying the parameters of interest, estimating them is still challenging without more restrictions (on the outcomes, covariates, and observations), such as the following assumptions used by, e.g., \cite{kitagawa2018should,kennedy2022towards,kennedy2022minimax}.

\begin{assumption}[Bounded Outcomes]
    \label{assumption:bounded_outcomes}
    For all $w \in \mathbb{W}$, $\,\mathbb{P}\{ |Y^*_w| \leq C/2\} = 1$ for some constant $C$.
\end{assumption}

\begin{assumption}[Bounded Covariates with Bounded Distribution]
    \label{assumption:bounded_covariates}
    The support $\mathbb{X}$ of $X$ is a Cartesian product of compact intervals. In addition, the Radon--Nikodym derivative of the probability measure induced by the random element $X$ is bounded from above and bounded away from zero.
\end{assumption}

\begin{assumption}[Independent and Identically Distributed Observations]
    \label{assumption:iid}
    For each $i \in \{1,\dots,n\}$, the vector $V_i \equiv (Y_i, X_i, W_i)$ has the same distribution as $V \equiv (Y, X, W)$, and $V_i \indep V_j$ for all $j \neq i$.
\end{assumption}

Several authors \citep[see, e.g.,][]{kitagawa2018should,kennedy2022towards,kennedy2022minimax} use Assumption \ref{assumption:bounded_outcomes}, which restricts the support of the potential outcomes, because it is convenient analytically; it also rules out cases that are practically unimportant but theoretically pathological. It implies that CAR and CAC are also bounded in magnitude. Assumption \ref{assumption:bounded_outcomes} may seem restrictive, but it is indeed satisfied in practice for most economic or health outcomes of interest and also, of course, for binary outcomes. Nevertheless, Assumption \ref{assumption:bounded_outcomes} is not strictly needed; it is possible to do the analysis in this paper by dropping it and using a weaker version of it. For example, one could instead assume that the conditional marginal distributions of the potential outcomes (conditional on covariates) satisfy standardized uniform integrability \citep{romano2012uniform}. Alternatively, one could use an even weaker assumption that $\E[(Y^*_w)^2 \,|\,X = x] < \infty$ for all $(w, x) \in \mathbb{W} \times \mathbb{X}$. Assumptions \ref{assumption:bounded_covariates} and \ref{assumption:iid} are also quite standard in the literature \citep[see, e.g.,][]{hirano2003efficient}.

The majority of the literature on causal inference under unconfoundedness focuses on the empirically important setting where observations can be classified as ``treated'' and ``untreated'' (or ``control'') units. In this setting, $\mathbb{W} = \{0, 1\}$, and $W$ is a binary indicator of treatment status; the event $\{W = 0\}$ refers to the control state, and the event $\{W = 1\}$ refers to the treatment state. In this case, the parameters CAR, CAC, and the related parameters all have special names as follows.

\begin{definition}[CTM, CCM, UTM, UCM: Conditional\,/\,Unconditional Treatment\,/\,Control Mean]
    \label{definition:ctm} When $\mathbb{W} = \{0, 1\}$ so that $W$ is binary, CTM is $\mu_1(x) = \mu(1,x)$ and CCM is $\mu_0(x)=\mu(0,x)$ for all $x \in \mathbb{X}$. In addition, UTM is $\overline{\mu}_1 = \E[\,\mu_1(X)]$ and UCM is $\overline{\mu}_0 = \E[\,\mu_0(X)]$.
\end{definition}

\begin{definition}[CATE: Conditional Average Treatment Effect]
    \label{definition:cate} When $\mathbb{W} = \{0, 1\}$ so that the treatment status is binary, the CATE function $\tau: \mathbb{X} \to \R$ is given by
    $$\tau(x) = \E[Y^*_1 - Y^*_0 \mid X = x] = \E[Y^*_1 \mid X = x] - \E[Y^*_0 \mid X = x] = \mu(1,x) - \mu(0,x) = \mu_1(x) - \mu_0(x).$$
\end{definition}

\begin{definition}[PATE: Population Average Treatment Effect]
    \label{definition:pate} When $\mathbb{W} = \{0, 1\}$, the PATE is
    $$\tau^* = \E[\tau(X)] = \E[\E[Y^*_1 - Y^*_0 \mid X]] = \E[Y^*_1 - Y^*_0] = \E[Y^*_1] - \E[Y^*_0].$$
\end{definition}

\begin{definition}[EATE: Empirical Average Treatment Effect]
    \label{definition:eate} When $\mathbb{W} = \{0, 1\}$ so that the treatment status is binary, the EATE for the sample is given by $$\tau^\circ = \widehat{\E}_n[\tau(X_i)] \equiv \frac{1}{n}\sum_{i = 1}^n \tau(X_i) = \frac{1}{n}\sum_{i = 1}^n \E[Y^*_1 - Y^*_0 \mid X = X_i].$$
\end{definition}

\begin{definition}[SATE: Sample Average Treatment Effect]
    \label{definition:sate} When $\mathbb{W} = \{0, 1\}$ and the observations are such that $Y_i = W_i\,Y^*_{1,i} + (1-W_i)\,Y^*_{0,i}$, where $(Y^*_{w,i})_{w \in \mathbb{W}}$ are the potential outcomes, for all $i \in \{1,\dots,n\}$, the SATE for the sample is given by $$\overline{\tau} = \frac{1}{n}\sum_{i = 1}^n Y^*_{1,i} - Y^*_{0,i}.$$
\end{definition}

Assumptions \ref{assumption:unconfoundedness}, \ref{assumption:pop_overlap}, and \ref{assumption:sutva} identify all of the above parameters except for SATE, which is not nonparametrically identified because it directly involves the unobserved counterfactual outcomes rather than expectations. However, note that PATE equals the expected value of both EATE and SATE. Other parameters, such as quantile treatment effects, are briefly considered in a later subsection, but the rest of the paper mostly focuses on aspects of the CATE function. Although it would be ideal to know CATE at every point on $\mathbb{X}$, an interpretable low-level summary of CATE would be much more useful in practice. For this purpose, an object called ``Best Summary of CATE'' (BATE) is defined as follows and includes the ``Group Average Treatment Effect'' (GATE) vector (defined below) and the PATE as special cases. GATEs are common parameters of interest in the literature because they are easily interpretable. The finite-dimensional BATE parameter is of interest for another obvious reason: estimates of the higher-dimensional (or infinite-dimensional) CATE function may have too much statistical uncertainty to be practically useful.

\begin{definition}[BATE: ``Best'' Summary of CATE Function]
    \label{definition:bate}
    If $\mathbb{W} = \{0,1\}$ and $\widetilde{Z}: \mathbb{X} \to \R^{d_Z}$ is a known function that converts $x \in \mathbb{X}$ into a finite-dimensional vector $\widetilde{Z}(x)\in \R^{d_Z}$ of basis functions so that $Z \equiv \widetilde{Z}(X)$ is not perfectly collinear almost surely, then BATE (``best'' summary of the conditional average treatment function, but ``best'' only in a particular sense) is given by
    $$\beta = \mathrm{arg\,min}_{\tilde{\beta} \,\in\, \R^{d_Z}}\,\E[(\tau(X) - \tilde{\beta}'\widetilde{Z}(X))^2] \text{, so } \E[Z(\tau(X) - Z'\beta)] = 0 \text{ and } \beta = \E[Z\,Z']^{-1}\E[Z\,\tau(X)],$$ 
    where $Z \equiv \widetilde{Z}(X)$. It also convenient to write $Z_i \equiv \widetilde{Z}(X_i)$ for all $i \in \{1,\dots,n\}$. When $Z$ contains binary indicators for population subgroups of interest, then BATE can be interpreted as the vector containing Group Average Treatment Effects (GATEs). When $\widetilde{Z}(\cdot) \equiv 1$, BATE equals PATE.
\end{definition}

Since CATE is a specific version of CAR, which is identified, it follows that BATE is also identified. The propensity of being treated is useful for expressing BATE in terms of observables. 
\begin{definition}[Propensity Score for Binary Treatment Indicator]
    \label{definition:ps_binary}
    When $\mathbb{W} = \{0,1\}$, let the term ``propensity score'' refer to $e(x) \equiv p(1,x) = 1 - p(0,x)$ for all $x \in \mathbb{X}$ without loss of generality, since $p(0,X) + p(1,X) = 1$ holds almost surely. In addition, let $e_i \equiv e(X_i)$ for all $i \in \{1, \dots, n\}$.
\end{definition}
\noindent Specifically, since $e_i = \E[W_i \mid X_i]$ and the expected value of $\mathcal{I}_i \equiv \frac{(W_i - e_i)\,Y_i}{e_i\,(1-e_i)}= \frac{W_i\,Y_i}{e_i} - \frac{(1-W_i)\,Y_i}{1 - e_i}$ given $X_i$ is $\E[\,\mathcal{I}_i \mid X_i] = \E\Big[\frac{W_i\,Y_i}{e_i} \Big\lvert \,X_i\Big] - \E\Big[\frac{(1-W_i)\,Y_i}{1 - e_i} \Big\lvert \,X_i\Big] = \frac{e_i\,\E[Y^*_{1,i} \mid X_i]}{e_i} - \frac{(1-e_i)\,\E[Y^*_{0,i} \mid X_i]}{1 - e_i}= \tau(X_i)$,
$$\E[Z_i\,Z'_i]^{-1}\E[Z_i\,\mathcal{I}_i] =  \E[Z_i\,Z'_i]^{-1}\E[Z_i\,\E[\,\mathcal{I}_i \mid X_i]] =\E[Z_i\,Z'_i]^{-1}\E[Z_i\,\tau(X_i)] = \beta$$
is identified under Assumptions \ref{assumption:unconfoundedness}--\ref{assumption:pop_overlap}. However, these conditions alone do not guarantee that the sample analog of $\E[Z_i\,Z'_i]^{-1}\E[Z_i\,\mathcal{I}_i]$ is $\sqrt{n}$-consistent and asymptotically normal, even in the simple case where $\widetilde{Z}(\cdot) \equiv 1$ \citep{khan2010irregular}. This motivates the next section.

\clearpage 

\clearpage

\section{Non-Inverse Alternatives to Inverse Probability Weighting}
\label{section:alternatives}

If strict overlap (in Definition \ref{definition:strict_overlap}) holds in addition to Assumptions \ref{assumption:unconfoundedness}--\ref{assumption:iid}, then it is quite straightforward to prove that $\hE[Z_i\,Z'_i]^{-1}\hE[Z_i\,\mathcal{I}_i]$, where $\hE[\,\cdot\,] \equiv \frac{1}{n}\sum_{i = 1}^n [\,\cdot\,]$, is a $\sqrt{n}$-consistent and asymptotically normal estimator of $\beta$ in the case where $W$ is binary. This estimator is based on a simple regression of $\mathcal{I}_i$ on $Z_i$; strict overlap (together with the other assumptions) ensures that $0 \prec \mathbb{V}[Z_i(\mathcal{I}_i - Z'_i\beta)] \prec \infty$, where $\mathbb{V}[\,\cdot\,]$ denotes the covariance matrix. The regressand $\mathcal{I}_i$ is called a ``pseudo-outcome'' \citep[see, e.g.,][]{kennedy2022towards} based on ``Inverse Probability Weighting'' (IPW), which is a widely used tool for causal inference. There are two versions of it \citep[see, e.g,][]{hirano2003efficient, hahn1998role} that are practically different but are based on the same underlying idea.

\begin{definition}[IPW: Inverse Probability Weighting]
    \label{definition:ipw}
    For all $w \in \mathbb{W}$, let $\widetilde{K}(W - w)$ denote either $\mathbb{I}\{W - w = 0\}$ or $\mathrm{lim}_{\,h\,\downarrow\,0} \,K_h(W-w)$, depending on whether $\mathbb{W}$ is a discrete set or a continuum, respectively. Then, IPW refers to the following related concepts that can both identify CAR:
    $$\frac{\E[\,\widetilde{K}(W-w)\,Y \mid X = x]}{p(w, x)} = \mu(w, x) \,\,\text{ and }\,\, \mu(w, x) = \E\Bigg[\frac{\widetilde{K}(W-w)\,Y}{p(w, x)}\,\Bigg\lvert\, X = x \Bigg]$$
    for all $(w,x) \in \mathbb{W} \times \mathbb{X}$. They are mathematically equivalent but have different practical implications for empirical work. For identifying CAC, IPW refers to the following related concepts:
    $$\theta[\widetilde{\mathbb{W}}, \kappa](x) = \sum_{w \,\in\, \widetilde{\mathbb{W}}} \,\,\kappa_w\,\frac{\widetilde{\mu}(w,x)}{p(w, x)}, \text{ where } \widetilde{\mu}(w,x) \equiv \E[\,\widetilde{K}(W-w)\,Y \mid X = x],$$
    is the sum of $|\widetilde{\mathbb{W}}|$ number of inverse-weighted conditional expectations \citep[see, e.g.,][]{hahn1998role}; and
    $$\theta[\widetilde{\mathbb{W}}, \kappa](x) = \E[\, \mathcal{I}[\widetilde{\mathbb{W}}, \kappa; p, x] \mid X = x], \text{ where } \mathcal{I}[\widetilde{\mathbb{W}}, \kappa; p, x] \equiv \sum_{w \,\in\, \widetilde{\mathbb{W}}} \,\,\kappa_w\,\frac{\widetilde{K}(W-w)\,Y}{p(w, x)},$$
    involves only one conditional expectation: that of the IPW ``pseudo-outcome'' $\mathcal{I}[\widetilde{\mathbb{W}}, \kappa; p, x]$, which is designed for the parameter $\theta[\widetilde{\mathbb{W}}, \kappa]$ \citep[see, e.g.,][]{hirano2003efficient,kennedy2022towards}.
\end{definition}

\begin{remark}
    When $\mathbb{W} = \widetilde{\mathbb{W}} = \{0,1\}$ and $\kappa = (-1, 1)$, then $\theta[\widetilde{\mathbb{W}}, \kappa] \equiv \tau$ (CATE). For all $x\in\mathbb{X}$,
    $$\tau(x) = \E[\,\mathcal{I}[\widetilde{\mathbb{W}}, \kappa; p, x] \mid X = x], \text{ where } \mathcal{I}[\widetilde{\mathbb{W}}, \kappa; p, x] = \frac{W\,Y}{e(x)} - \frac{(1-W)\,Y}{1 - e(x)} = \frac{[W - e(x)]\,Y}{e(x)[1-e(x)]}.$$
\end{remark}

In the above expressions, the propensity score is a nuisance function but plays a central role in the identification of CAR and CAC. However, it is possible to identify CAC (and thus also CAR) in another manner, for which the propensity score is relatively less central. A widely used strategy in the literature \citep{kennedy2022towards,semenova2021debiased,su2019non,chernozhukov2022locally}, based on the proposal of \cite{robins1994estimation}, augments the IPW pseudo-outcome to make identification of CAC robust to misspecification of either the propensity score function or the conditional means of the potential outcomes. The introduction of these additional nuisance functions limits the importance of the propensity score for identifying CAC.

\begin{definition}[AIPW: Augmented Inverse Probability Weighting]
    \label{definition:aipw}
    Let $\gamma$ and $\varphi$ be functions from $\mathbb{W} \times \mathbb{X}$ to $\R$ such that either $\gamma(w, x) = \E[Y \mid W = w,\, X = x]$ or $\varphi(w, x) = \E[\,\widetilde{K}(W - w) \mid X = x]$ holds for all $(w, x) \in \mathbb{W} \times \mathbb{X}$. Then, following AIPW concept can identify CAC (and thus CAR):
    $$\theta[\widetilde{\mathbb{W}}, \kappa](x) = \E[\, \mathcal{A}[\widetilde{\mathbb{W}}, \kappa; \gamma, \varphi, x] \mid X = x], \text{ where } \mathcal{A}[\widetilde{\mathbb{W}}, \kappa; \gamma, \varphi, x] \equiv \sum_{w \,\in\, \widetilde{\mathbb{W}}} \,\,\kappa_w\,\frac{\mathcal{A}_w[\widetilde{\mathbb{W}}, \kappa; \gamma, \varphi, x]}{\varphi(w, x)}$$
    and $\mathcal{A}_w[\widetilde{\mathbb{W}}, \kappa; \gamma, \varphi, x] \equiv \widetilde{K}(W - w)\,Y + \gamma(w, x)[- \widetilde{K}(W - w) + \varphi(w, x)]$ so that
    $$\mathcal{A}[\widetilde{\mathbb{W}}, \kappa; \gamma, \varphi, x] = \sum_{w \,\in\, \widetilde{\mathbb{W}}} \,\,\kappa_w\,\Bigg[\gamma(w,x) + \frac{\widetilde{K}(W - w)}{\varphi(w,x)}\,[Y - \gamma(w,x)]\Bigg]$$
    is the AIPW ``pseudo-outcome.'' For all $x \in \mathbb{X}$, note that
    $$\E[\, \mathcal{A}[\widetilde{\mathbb{W}}, \kappa; \gamma, \varphi, x] \mid X = x] = \sum_{w \,\in\, \widetilde{\mathbb{W}}} \,\,\kappa_w\,\Bigg[ \gamma(w,x) + \frac{p(w, x)\,\mu(w,x)}{\varphi(w,x)} - \frac{p(w, x)\,\gamma(w,x)}{\varphi(w,x)} \Bigg]$$ equals $\theta[\widetilde{\mathbb{W}}, \kappa](x)$ if $\gamma \equiv \mu$ or $\varphi \equiv p$, and so AIPW has a ``double robustness'' property, i.e., it is ``doubly robust'' to misspecification of either $\varphi$ or $\gamma$ (but not both).
\end{definition}

\begin{remark}
    When $\mathbb{W} = \widetilde{\mathbb{W}} = \{0,1\}$ and $\kappa = (-1, 1)$, then $\theta[\widetilde{\mathbb{W}}, \kappa] \equiv \tau$ (CATE). For all $x\in\mathbb{X}$,
    $\tau(x) = \E[\,\mathcal{A}[\widetilde{\mathbb{W}}, \kappa; \gamma, \varphi, x] \mid X = x]$, where $\varphi$ is such that $1 - \varphi(0,\cdot) = \varphi(1,\cdot) \equiv \varphi_1(\cdot)$ and $$\mathcal{A}[\widetilde{\mathbb{W}}, \kappa; \gamma, \varphi, x] = \gamma(1,x) + \frac{W\,(Y - \gamma(1,x))}{\varphi_1(x)} - \gamma(0,x) - \frac{(1-W)\,(Y - \gamma(0,x))}{1 - \varphi_1(x)} $$
    $$\implies \mathcal{A}[\widetilde{\mathbb{W}}, \kappa; \gamma, \varphi, x] = [\gamma(1,x) - \gamma(0,x)] + \frac{[W - \varphi_1(x)]\,[Y - \gamma(W, x)]}{\varphi_1(x)[1-\varphi_1(x)]}.$$
\end{remark}

\vspace*{4mm}

To understand the implications of limited overlap (i.e., scenario where strict overlap does not necessarily hold) for IPW or AIPW pseudo-outcome regressions \citep[see, e.g.,][]{kennedy2022towards}, it is useful to consider the case with binary treatment rather than the multivalued treatment setting. This helps in keeping the focus on key ideas rather than cumbersome notation. Hence, I deliberately do not return to the (conceptually similar) multivalued treatment setting until later in Section \ref{section:spw}.

For flexible estimation of CATE in the binary treatment setting, the latest methods in the literature \citep[see, e.g.,][]{kennedy2022towards} involve nonparametric regressions of $\mathcal{A}[\{0, 1\}, (-1, 1); \widehat{\gamma}, \widehat{\varphi}, X_i]$ on $X_i$, where $\widehat{\gamma}$ and $\widehat{\varphi}$ are nonparametric estimates of the respective nuisance parameters that are constructed using a separate auxiliary sample. Convergence rates of such estimators crucially depend on the risk of the ``oracle estimator,'' which nonparametrically regresses $\mathcal{A}[\{0, 1\}, (-1, 1); \gamma, \varphi, X_i]$ on $X_i$. This ``oracle risk'' further depends on the conditional variance of $\mathcal{A}[\{0, 1\}, (-1, 1); \gamma, \varphi, X]$, which is the AIPW pseudo-outcome. Similarly, the conditional variance of the IPW pseudo-outcome $\mathcal{I}[\{0, 1\}, (-1, 1); \gamma, \varphi, X]$ determines the risk of the oracle estimator that nonparametrically regresses $\mathcal{I}[\{0, 1\}, (-1, 1); p, X_i]$ on $X_i$. However, without strict overlap, the conditional variances of the IPW and AIPW pseudo-outcomes may not necessarily be bounded, as discussed below. Thus, the associated oracle estimators may themselves be unstable under limited overlap.

\begin{remark}
    For all $x \in \mathbb{X}$, the conditional variance of the IPW pseudo-outcome is given by
    $$\mathbb{V}[\,\mathcal{I}[\{0, 1\}, (-1, 1); p, X] \mid X = x] = \mathbb{V}\Bigg[\frac{(W-e(X))Y}{e(X)[1-e(X)]} \,\Bigg\lvert X = x\Bigg] = \frac{\mathbb{V}[(W - e(X))Y \mid X = x]}{e(x)^2[1-e(x)]^2},$$
    where the numerator $\mathbb{V}[(W - e(X))Y \mid X = x]$ is uniformly bounded because $(W - e(X))Y$ is a bounded random variable by Assumptions \ref{assumption:unconfoundedness}--\ref{assumption:iid}. Letting $\sigma^2_w(x) = \mathbb{V}[Y^*_w \mid X = x]$ for $w \in \{0,1\}$, $\mathbb{V}[\,\mathcal{I}[\{0, 1\}, (-1, 1); p, X] \mid X = x] = \E[(W-e(X))Y/[e(X)(1-e(X))] \mid X = x]$ simplifies to
    $$\frac{e(x)[1-e(x)]^2 \E[(Y^*_1)^2 \mid X = x]}{e(x)^2[1-e(x)]^2} + \frac{[1-e(x)]e(x)^2\E[(Y^*_0)^2 \mid X = x]}{e(x)^2[1-e(x)]^2} - \mathbb{E}\Bigg[\frac{(W-e(X))Y}{e(X)[1-e(X)]} \,\Bigg\lvert X = x\Bigg]^2 $$
    $$= \frac{[1-e(x)]\,[\sigma_1^2(x) + \mu_1^2(x)] + e(x)\,[\sigma_0^2(x) + \mu_0^2(x)] - e(x)[1-e(x)]\tau(x)^2}{e(x)[1-e(x)]} \equiv \frac{\mathcal{C}_\mathcal{I}(x)}{e(x)[1-e(x)]}.$$
    The numerator $\mathcal{C}_\mathcal{I}(x)$ in the above expression is bounded above (due to Assumptions \ref{assumption:unconfoundedness}--\ref{assumption:iid}), but the reciprocal of the denominator $e(x)[1-e(x)]$ is not bounded, unless strict overlap holds. Therefore, in general, the conditional variance of the IPW pseudo-outcome may be arbitrarily large under limited overlap (unless obscure restrictions, such as $\mathcal{C}_\mathcal{I}(x) \propto e(x)[1-e(x)]$, hold). Similarly, for all $x\in\mathbb{X}$, the conditional variance of the AIPW pseudo-outcome
    $$\mathbb{V}[\,\mathcal{A}[\{0, 1\}, (-1, 1); \gamma, \varphi, X] \mid X = x] = \frac{\mathcal{C}_\mathcal{A}(x)}{\varphi_1(x)^2[1-\varphi_1(x)]^2}$$ may not necessarily be bounded if $\varphi$ is not uniformly bounded away from zero, even though the numerator $\mathcal{C}_\mathcal{A}(x) \equiv \mathbb{V}[ \varphi_1(X)[1-\varphi_1(X)][\gamma(1,X) - \gamma(0,X)] + [W - \varphi_1(X)]\,[Y - \gamma(W, X)] \mid X = x]$ itself is bounded (by Assumptions \ref{assumption:unconfoundedness}--\ref{assumption:iid} and the associated restrictions on $\varphi$ and $\gamma$). 
\end{remark}

The above arguments, related to the points made by \cite{khan2010irregular} and others, show that the (augmented) IPW technique is a double-edged sword, because even simple propensity score models \citep[see, e.g.,][]{heiler2021valid}, such as $e(X) = X = \Phi(2\,X^*)$ with $X^* \sim \mathcal{N}(0,1)$, violate strict overlap. The main component of the (A)IPW pseudo-outcome (i.e., the denominator equalling the conditional variance of the binary treatment indicator) is powerful and empirically useful (under strong overlap) because it allows for nonparametric estimation of PATE without strong restrictions on the heterogeneity in the unit-level treatment effects \citep{hirano2003efficient}. However, the inverse weighting component is also the reason for the potential instability of the (A)IPW technique in many empirically relevant scenarios (involving limited overlap). One need not look beyond a very simple setting to understand the issue: a situation where the propensity score function is known and the CATE has a linear form. My approach to the problem in this setting, as demonstrated in this section, illustrates and motivates my general framework in Section \ref{section:spw} that tackles the challenge from semi/non-parametric perspectives. Hence, throughout this section, I make the following assumption and also treat the propensity score function $e(\cdot)$ as known.

\begin{assumption}[Linearity of CATE]
    \label{assumption:linear_cate}
    For all $x \in \mathbb{X}$, CATE equals $\tau(x) = \beta'\widetilde{Z}(x)$, where $\beta \in \R^{d_Z}$ is an unknown $d_Z$-dimensional parameter and $\widetilde{Z}: \mathbb{X} \to \R^{d_Z}$ is a known function that generates a finite-dimensional basis vector in $\R^{d_Z}$ such that $Z \equiv \widetilde{Z}(X)$ is not perfectly collinear almost surely. In this case, note that $\beta$ coincides with the BATE parameter in Definition \ref{definition:bate}. 
\end{assumption}

I now define a class of ``General Probability Weighting'' (GPW) estimators that may be used to estimate the above parameter of interest $\beta$. This estimator class includes the usual IPW estimator but also ``Non-Inverse Probability Weighting'' (NPW) estimators that are defined below. As shown and demonstrated later, NPW estimators are asymptotically normal under Assumptions \ref{assumption:unconfoundedness}--\ref{assumption:iid} and Assumption \ref{assumption:linear_cate} even if the IPW estimator may not be $\sqrt{n}$-consistent and asymptotically normal.

\begin{definition}[GPW: General Probability Weighting]
    \label{definition:gpw_estimator} Let $e_i \equiv e(X_i)$ and $Z_i \equiv \widetilde{Z}(X_i)$ for all $i \in \{1,\dots,n\}$. Then, for any $\nu \in \R$, the associated GPW estimator $\widehat{\beta}_\nu$ (with index $\nu$) is given by
    $$\widehat{\beta}_\nu = \hE[(e_i(1-e_i))^{\nu + 1}Z_i\,Z'_i]^{-1}\hE[(e_i(1-e_i))^{\nu}Z_i(W_i-e_i)Y_i] = \hE[Z_i(\nu)Z_i(\nu)']^{-1}\hE[Z_i(\nu)Y_i(\nu)],$$
    where $\hE[\,\cdot\,] \equiv \frac{1}{n}\sum_{i = 1}^n [\,\cdot\,]$, $\displaystyle\,Y_i(\nu) \equiv \frac{(W_i-e_i)Y_i}{\sqrt{(e_i(1-e_i))^{1 - \nu}}}$, and $\,Z_i(\nu) \equiv Z_i\sqrt{(e_i(1-e_i))^{\nu + 1}}$.
\end{definition}

\begin{definition}[NPW: Non-Inverse Probability Weighting]
    \label{definition:npw_estimator} The class of NPW estimators is defined as $\{\widehat{\beta}_\nu: \nu \geq 0\}$, which includes every GPW estimator with a nonnegative index, i.e., $\widehat{\beta}_\nu$ with $\nu \geq 0$. Of these, the estimator $\widehat{\beta}_1$ is even more special because its index $\nu = 1$ is the lowest possible nonnegative index such that $Y_i(\nu)$ and $Z_i(\nu)$ (separately rather than just their product) do not involve any inverse weighting components. Thus, the NPW estimator $\widehat{\beta}_1$, which can be obtained by regressing $(W_i - e_i)Y_i$ on $e_i(1-e_i)\,Z_i$, may be called ``the'' NPW estimator.
\end{definition}

\begin{remark}
    Note that the GPW estimator $\widehat{\beta}_{-1}$ with index $\nu = -1$ is the ``usual'' IPW estimator. In addition, any GPW estimator can be expressed as a weighted version of the usual IPW estimator, and vice versa. Specifically, $\widehat{\beta}_\nu = \hE[\omega_i(\nu)\,Z_iZ'_i]^{-1}\hE\Big[\omega_i(\nu)Z_i\frac{(W-e_i)Y_i}{e_i(1-e_i)}\Big],$
    where the weights are given by $\omega_i(\nu) = (e_i(1-e_i))^{\nu + 1}$, for all $\nu \in \R$. Similarly, the IPW estimator is a weighted version of another GPW estimator: $\widehat{\beta}_{-1} = \hE\Big[\frac{1}{\omega_i(\nu)}Z_i(\nu)Z_i(\nu)'\Big]^{-1}\hE\Big[\frac{1}{\omega_i(\nu)}Z_i(\nu)Y_i(\nu)\Big]$ for any $\nu \in \R$.
\end{remark}

\begin{remark}
    As mentioned above, GPW estimators can be expressed as weighted versions of the IPW estimator. However, the GPW estimators should not be confused with Generalized Least Squares (GLS) estimators except in some very special cases (e.g., $\mathcal{C}_\mathcal{I}(x) = \sigma^2[e(x)(1-e(x))]^{-\nu}$ for all $x \in \mathbb{X}$ for some $\sigma^2 > 0$ so that the GPW estimator $\widehat{\beta}_\nu$ coincides with a GLS estimator).
\end{remark}

Even in the simple case where $Z_i = 1$ ``a.s.''~(almost surely), the IPW estimator $\widehat{\beta}_{-1}$ may not be $\sqrt{n}$-estimable \citep{khan2010irregular} under Assumptions \ref{assumption:unconfoundedness}--\ref{assumption:iid} and Assumption \ref{assumption:linear_cate}. Specifically, \cite{ma2020robust} show that the distribution of $\widehat{\beta}_{-1}$ depends on the distribution of the propensity scores and may thus be non-Gaussian (e.g., asymmetric L\'evy stable) asymptotically. More generally, the ``usual'' (multivariate) central limit theorems may not necessarily be applicable to the sample mean of $(e_i(1-e_i))^{\nu}Z_i(W_i-e_i)Y_i$ when $\nu < 0$, in which case the random variable may potentially have infinite variance (if overlap is weak\,/\,limited). In this case, generalized central limit theorems may need to be applied instead, but they typically imply asymptotic non-Gaussianity of $\hE[Z_i(\nu)Y_i(\nu)]$ (and thus of $\widehat{\beta}_\nu$) when $\nu < 0$ under limited overlap. However, I show in this section that the NPW estimators are asymptotically normal unlike the IPW estimators (or, more generally, GPW estimators with negative indices) using the following proposition.

\begin{proposition}
    \label{proposition:npw_cme}
    Under Assumptions \ref{assumption:unconfoundedness}--\ref{assumption:pop_overlap}, $\E[e(X)(1-e(X))\tau(X) - (W - e(X))Y \mid X] = 0$ a.s.
\end{proposition}
\begin{proof}
    Note that $\E[(W- e(X))Y \mid X = x] = e(x)(1-e(x))\mu_1(x) + (1-e(x))(-e(x))\mu_0(x) = e(x)(1-e(x))\tau(x) \implies e(x)(1-e(x))\tau(x) - \E[(W- e(X))Y \mid X = x] = 0$ for all $x \in \mathbb{X}$.
\end{proof}

\begin{remark}[Connection to the Anderson--Rubin Test for Models with Weak Instruments] There is an interesting connection between the usual Anderson--Rubin test in the literature on weak instrumental variable models \citep[see, e.g.,][]{andrews2019weak} and the above conditional moment equality. As discussed in Definition \ref{definition:ipw}, for all $x \in \mathbb{X}$, $\tau(x) = \theta[\{0,1\},(-1,1)](x)$ in the binary treatment case is as follows: $\tau(x) = \E[W\,Y \mid X = x]/e(x) - \E[(1-W)\,Y \mid X = x]/(1 - e(x))$, which \cite{hahn1998role} uses to construct a semiparametrically efficient estimator of PATE, and so $$\tau(x) = \frac{(1-e(x))\E[W\,Y\mid X = x] - e(x)\E[(1-W)\,Y \mid X = x]}{e(x)[1-e(x)]} \equiv \frac{\tau_{\mathrm{numerator}}(x)}{\tau_{\mathrm{denominator}}(x)}$$
is a ratio with the numerator $\tau_{\mathrm{numerator}}(x) \equiv \E[W\,Y \mid X = x] - \E[W \mid X = x]\,\E[Y \mid X = x]$ and the denominator $\tau_{\mathrm{denominator}}(x) = \E[W \mid X = x]\E[1 - W \mid X = x]$. In the limited overlap case, this denominator can be arbitrarily close to zero. The main parameter in the just-identified (single) weak instrumental variable model can also be expressed as a ratio with a denominator that can be very close to zero; in this case, the parameter of interest is only ``weakly identified'' and conventional statistical inference procedures can be misleading. A remedy that is often used is the Anderson--Rubin test, which is robust to such ``weak identification'' in a particular sense. The main trick that the Anderson--Rubin test uses is to simply rearrange the equation, just as how the above equation for CATE can be rearranged as follows: $\tau(x)\,\tau_{\mathrm{denominator}}(x) - \tau_{\mathrm{numerator}}(x) = 0$, i.e., $\tau(x)[e(x)(1-e(x))] - \E[(W- e(X))Y \mid X = x] = 0$, for all $x \in \mathbb{X}$, as in Proposition \ref{proposition:npw_cme}.    
\end{remark}

\begin{proposition}[Asymptotic Normality of Non-Inverse Probability Weighting (NPW) Estimators with $\nu \geq 0$]
    \label{proposition:npw_asym_normality}
    Under Assumptions \ref{assumption:unconfoundedness}--\ref{assumption:iid} and \ref{assumption:linear_cate}, $\sqrt{n}(\widehat{\beta}_\nu - \beta) \rightsquigarrow \mathcal{N}(0, \Sigma_\nu)$ for any $\nu \geq 0$, where $\Sigma_\nu = \Sigma_{\nu,1}^{-1} \Sigma_{\nu,2} \Sigma_{\nu,1}^{-1}$ is the covariance matrix such that $\Sigma_{\nu,1} = \E[(e_i(1-e_i))^{\nu+1}Z_iZ'_i]$ and $\Sigma_{\nu,2} = \mathbb{E}[Z_iZ'_i\{e_i(1-e_i)\}^{2\nu}\{(W_i - e_i)Y_i - e_i(1-e_i)Z'_i\beta\}^2]$. If $\widehat{\Sigma}_\nu$ is the sample analog of $\Sigma_\nu$ (using $\hE$ instead of $\E$ and $\widehat{\beta}_\nu$ instead of $\beta$), then $\widehat{\Sigma}^{-1/2}_\nu \sqrt{n}(\widehat{\beta}_\nu - \beta) \rightsquigarrow \mathcal{N}(0, I_{d_Z})$ for any $\nu \geq 0$.
\end{proposition}
\begin{proof}
    Note that $\widehat{\beta}_\nu = \hE[Z_i(\nu)Z_i(\nu)']^{-1}\hE[Z_i(\nu)Z_i(\nu)'\beta + Z_i(\nu)Y_i(\nu) - Z_i(\nu)Z_i(\nu)'\beta]$, and so
    $\sqrt{n}(\widehat{\beta}_\nu - \beta) = \hE[(e_i(1-e_i))^{\nu + 1}Z_i\,Z'_i]^{-1}\sqrt{n}\,\hE[(e_i(1-e_i))^{\nu}Z_i\{(W_i - e_i)Y_i - e_i(1-e_i)Z'_i\beta\}]$. Since $Z_i$ is not perfectly collinear (by Assumption \ref{assumption:linear_cate}) and $e_i > 0$ a.s. (by Assumption \ref{assumption:pop_overlap}), $Z_i(\nu)$ is also not perfectly collinear, and so $\hE[Z_i(\nu)Z_i(\nu)']$ is invertible. The random vector $(e_i(1-e_i))^{\nu}Z_i\{(W_i - e_i)Y_i - e_i(1-e_i)Z'_i\beta\}$ is bounded a.s.~with finite covariance matrix $\Sigma_{\nu,2}$ (by Assumptions \ref{assumption:bounded_outcomes}--\ref{assumption:bounded_covariates} and since $\mathbb{W} = \{0, 1\}$ so that $0 < e_i < 1 \implies 0 < (e_i(1-e_i))^\nu \leq 1$ a.s. for any $\nu \geq 0$) and also has mean zero (since $\E[(e_i(1-e_i))^{\nu}Z_i\,\E[\{(W_i - e_i)Y_i - e_i(1-e_i)Z'_i\beta\} \mid X_i]] = 0$ by Assumptions \ref{assumption:unconfoundedness}--\ref{assumption:pop_overlap} and Proposition \ref{proposition:npw_cme}). These results, together with Assumption \ref{assumption:iid}, imply that $\hE[Z_i(\nu)Z_i(\nu)']$ and thus $\widehat{\Sigma}_\nu$ converge in probability to $\Sigma_{\nu,1}$ and $\Sigma_{\nu}$, respectively, and therefore also imply that $\sqrt{n}(\widehat{\beta}_\nu - \beta) \rightsquigarrow \mathcal{N}(0, \Sigma_\nu)$ and $\widehat{\Sigma}^{-1/2}_\nu \sqrt{n}(\widehat{\beta}_\nu - \beta) \rightsquigarrow \mathcal{N}(0, I_{d_Z})$ for any $\nu \geq 0$.
\end{proof}

\begin{remark}
    Although the above result holds under limited overlap asymptotically, the matrix $\hE[Z_i(\nu)Z_i(\nu)']$ may (in some cases) be close to being non-invertible even in a reasonably large sample. This is essentially a ``weak identification'' problem. However, even in such scenarios, it is possible to construct confidence sets for $\beta$ that are robust to the strength of identification. For this purpose, one may invert, for example, the identification/singularity-robust Anderson--Rubin (SR-AR) tests \citep{andrews2019identification} in the context of the moment condition model $\E[Z_i\{e_i(1-e_i)\}^{\nu}\{(W_i - e_i)Y_i - e_i(1-e_i)Z'_i\beta\}] = 0$, which is the basis for all the GPW estimators.
\end{remark}

\begin{remark}
    Under Assumptions \ref{assumption:unconfoundedness}--\ref{assumption:iid} and \ref{assumption:linear_cate}, both the PATE $\tau^* = \E[\tau(X)] = \beta'\,\E[Z]$ and the EATE $\tau^\circ = \hE[\tau(X_i)] = \beta'\,\hE[Z_i]$ parameters can be consistently estimated using $\widehat{\beta}'_\nu\, \hE[Z_i]$, where $\nu \geq 0$, based on Proposition \ref{proposition:npw_asym_normality}. The delta method can be appropriately applied for inference. Alternatively, the method in Section 6.1 of \cite{rothe2017robust} can be adapted for inference on PATE.
\end{remark}

\begin{remark}
    Note that Proposition \ref{proposition:npw_asym_normality} does not make any claim regarding the efficiency of NPW estimators. In fact, in some cases, there exist more efficient estimators. For example, if $\mathbb{X}$ can be partitioned into strata such that the CATE is constant within each stratum, i.e., if $Z$ is a vector of binary indicators (for the strata), then the CATE within each stratum can be estimated using the ``overlap weighting'' estimator proposed by \cite{li2018balancing}. It has the smallest asymptotic variance when the potential outcomes exhibit homoskedasticity. In the simple case where $Z \equiv 1$, the overlap weighting estimator is just $\hE[(1-e_i)W_iY_i]/\hE[(1-e_i)W_i] - \hE[e_i(1-W_i)Y_i]/\hE[e_i(1-W_i)]$. When CATE is not constant over $\mathbb{X}$, both the overlap weighting estimator and the NPW estimator(s) using $Z \equiv 1$ do not consistently estimate PATE. Instead, they estimate a ``Weighted Average Treatment Effect'' (WATE). See \cite{hirano2003efficient} for further discussion on efficient estimation of WATE.
\end{remark}

Figures \ref{figure:ipw_vs_npw_example_coef1est_dist} through \ref{figure:npw_joint_density_example} in the Appendix illustrate Proposition \ref{proposition:npw_asym_normality} by contrasting the sampling distributions of the IPW estimator $\widehat{\beta}_{-1}$ and the NPW estimator $\widehat{\beta}_1$ when $n = 10^5$ (sample size) in the following simple setting with limited overlap: $X \sim \mathrm{Uniform}(0,1)$, $W \sim \mathrm{Bernoulli}(X^4)$, $(\upsilon_1, \upsilon_2) \sim \mathrm{Uniform}[(-2, 2)^2]$, $Y = 10\,(1-X^4) + X^4\,\upsilon_1 + W\,(\tau(X) + 2\,\upsilon_2)$, where $\tau(X) = \beta'Z$ with $\beta = (3, -2)$ and $Z = (1,X)$ so that $\tau^* = \E[\tau(X)] = 3 - 2\,\E[X] = 2$. As the figures show, the NPW estimator in this setting is approximately Gaussian but the IPW estimator is not.

The theoretical underpinning of GPW (and thus NPW) estimators is the moment condition $$\E[Z_i\{e_i(1-e_i)\}^{\nu}\{(W_i - e_i)Y_i - e_i(1-e_i)Z'_i\beta\}] = 0$$
for any $\nu \in \R$, since $\E[Z_i\{e_i(1-e_i)\}^{\nu}\,\E[(W_i - e_i)Y_i - e_i(1-e_i)Z'_i\beta \mid X_i]] = 0$ by Proposition~\ref{proposition:npw_cme}. However, an alternative moment condition is $\E[Z_i\{(W_i - e_i)Y_i - (W_i - e_i)^2Z'_i\beta\}] = 0$, which holds because $\E[(W_i - e_i)^2Z'_i\beta \mid X_i] = e_i(1-e_i)Z'_i\beta$ almost surely, resulting in the estimator
$\hE[(W_i - e_i)^2Z_iZ'_i]^{-1}\,\hE[Z_i(W_i-e_i)Y_i],$
which can be obtained by regressing $Y_i$ on $(W_i - e_i)Z_i$. (The connection between this estimator and the so-called ``\cite{robinson1988root} transformation'' is discussed in the next section.) The equality $\E[Z_i\{(W_i - e_i)Y_i - \frac{1}{2}[W_i(1-e_i)+e_i(1-W_i)]Z'_i\beta\}] = 0$ could also be exploited because $\E[\frac{1}{2}[W_i(1-e_i)+e_i(1-W_i)]Z'_i\beta \mid X_i] = e_i(1-e_i)Z'_i\beta$ almost surely, and so
$\hE[\frac{1}{2}[W_i(1-e_i)+e_i(1-W_i)]Z_iZ'_i]^{-1}\,\hE[Z_i(W_i-e_i)Y_i]$ is another possible estimator. Similar reasoning as in Proposition \ref{proposition:npw_asym_normality} can be used to establish the asymptotic normality of both of these alternative estimators even under limited overlap when Assumptions \ref{assumption:unconfoundedness}--\ref{assumption:iid} and \ref{assumption:linear_cate} hold. In addition, when one-sided overlap holds, e.g., if $1 - e_i$ is bounded away from zero almost surely, the estimator $\hE[W_i\,Z_iZ'_i]^{-1}\hE[Z_i(W_i - e_i)Y_i/(1-e_i)]$ is asymptotically normal under the same assumptions (even if may have bad properties in finite samples or under weak identification asymptotics) because $\E[W_i\tau(X_i) - (W_i - e_i)Y_i/(1-e_i) \mid X_i] = e_i\tau(X_i) - e_i(1-e_i)\tau(X_i)/(1-e_i) = 0$ almost surely. All of these estimators, including the NPW estimator(s), are based on a more general principle called ``Stable Probability Weighting'' (SPW), which I develop in the next section.

\clearpage

\section{Stable Probability Weighting for Estimation and Inference}
\label{section:spw}

Under Assumptions \ref{assumption:unconfoundedness}--\ref{assumption:pop_overlap}, the Conditional Average Treatment Effect (CATE) function $\tau(\cdot)$ is nonparametrically identified without assuming strict overlap and without using propensity scores because
$\E[Y \mid W = 1,\, X = x] - \E[Y \mid W = 0,\, X = x] = \E[Y^*_1 \mid X = x] - \E[Y^*_0 \mid X = x] = \mu_1(x) - \mu_0(x) = \mu(1,x) - \mu(0,x) = \tau(x)$ for all $x \in \mathbb{X}$. This suggests that one could potentially learn about CATE by simply estimating $\mu$ through a nonparametric regression of $Y$ on $(W, X)$ and then use a plug-in version of $\mu(1,x) - \mu(0,x)$ to estimate CATE. However, the resulting naive plug-in estimator may be too noisy even if the CATE function itself is smooth and even when strict overlap holds; see Section 2.2 of \cite{kennedy2022towards}. Such an approach may be especially problematic from an inferential perspective under weak overlap; see Section 3.2 in the Supporting Information of \cite{kunzel2019metalearners}. An alternative is to use a parametric (or semiparametric) model for $\mu$, but this approach may be highly sensitive to model specification. A more fruitful approach is to treat $\mu$ as a nuisance parameter and to focus on the target parameter $\tau$ in addition to exploiting the data structure (i.e., unconfoundedness) more fully. \cite{kennedy2022towards} and others (cited in Section \ref{section:introduction}) take such an approach under the assumption of strict overlap. However, to the best of my knowledge, the empirically relevant setting of limited overlap is largely ignored in the existing literature on heterogeneous causal effects. Hence, this section develops the relevant large-sample theory.

In Subsection \ref{subsection:basic_spw}, I develop a basic form of a general principle called ``Stable Probability Weighting'' (SPW) for learning about CATE under limited overlap. I then further develop a more robust version of it called Augmented SPW (ASPW) in Subsection \ref{subsection:aspw}. In Subsection \ref{subsection:spw_multivalued}, I generalize the (A)SPW framework for learning about the Conditional Average Contrast (CAC) function as well as distributional or quantile treatment effects in settings with multivalued treatments. In Subsection \ref{subsection:large_sample_spw}, I discuss methods for estimating and inferring parameters of interest from parametric, semiparametric, nonparametric, machine learning, and policy learning perspectives.

\subsection{Main Conditional Moment Restrictions Robust to Limited Overlap}
\label{subsection:basic_spw}

Without further ado, I define the main Stable Probability Weighting (SPW) principle for CATE.

\begin{definition}[SPW for CATE: Stable Probability Weighting for CATE]
    \label{definition:spw_cate}
    The generalized residual $\Psi(Y, X, W; \tau, e, \eta)$ satisfies the Stable Probability Weighting (SPW) principle for CATE $\tau(\cdot)$ if $$\Psi(Y, X, W; \tau, e, \eta) = \mathcal{S}(X, W; e)\,\tau(X) - \tilde{s}(X; e)\,(W - e(X))\,Y + \xi(X, W; e, \eta)$$
    such that $\mathcal{S}$ and $\xi$ are bounded real-valued functions on $\mathbb{X} \times \mathbb{W}$ that depend on $e$ and $\eta$ (a bounded vector-valued function on $\mathbb{X}$) and satisfy the following conditions: $\E[\xi(X, W; e, \eta) \mid X] = 0$ a.s.; $s(X; e) \equiv \E[\mathcal{S}(X, W; e) \mid X]$ and $0 < \tilde{s}(X; e) \equiv s(X; e)/[e(X)(1-e(X))] \leq M$ a.s.~for some $M > 0$.
\end{definition}

\begin{remark}
    In the above definition, the boundedness restrictions are imposed on the random variables $\mathcal{S}(X, W; e)$ and $\xi(X, W; e, \eta)$ for expositional ease, but it is possible to re-define SPW using milder restrictions, such as $\E[\mathcal{S}(X, W; e)^2 \mid X] < \infty$ a.s. and $\E[\xi(X, W; e, \eta)^2 \mid X] < \infty$ a.s.
\end{remark}

\begin{theorem}[Conditional Moments of the SPW-Based Generalized Residual]
    \label{theorem:spw_moments} If the generalized residual $\Psi(Y, X, W; \tau, e, \eta)$ satisfies the SPW principle (in Definition \ref{definition:spw_cate}) and Assumptions \ref{assumption:unconfoundedness}--\ref{assumption:bounded_covariates} hold, then $\E[\Psi(Y, X, W; \tau, e, \eta) \mid X] = 0$ and $\mathbb{V}[\Psi(Y, X, W; \tau, e, \eta) \mid X] < \infty$ almost surely.
\end{theorem}
\begin{proof}
    By Proposition \ref{proposition:npw_cme}, $\E[\tilde{s}(X; e)\,(W - e(X))\,Y \mid X] = s(X; e)\,[e(X)(1-e(X))]^{-1+1}\tau(X)$ a.s., and so the result $\E[\Psi(Y, X, W; \tau, e, \eta) \mid X] = 0$ (a.s.) follows because $\E[\xi(X, W; e, \eta) \mid X] = 0$ and $\E[\mathcal{S}(X, W; e)\,\tau(X) \mid X] = \E[\mathcal{S}(X, W; e) \mid X]\,\tau(X) = s(X; e)\,\tau(X)$ a.s.; in addition, the boundedness assumptions imply that (a.s.) $\mathbb{V}[\Psi(Y, X, W; \tau, e, \eta) \mid X] = \E[\Psi(Y, X, W; \tau, e, \eta)^2 \mid X] < \infty$.
\end{proof}

\begin{example}[IPW Under Strict Overlap Satisfies SPW Principle]
    If strict overlap (as per Definition \ref{definition:strict_overlap}) holds so that $0 < \underline{p}^* < e(X) < 1 - \underline{p}^* < 1$ a.s., then the IPW technique, which uses $\mathcal{S}\equiv 1$ and $\xi \equiv 0$, satisfies the SPW principle with $M = (\underline{p}^*)^{-2}$. Thus, SPW is broader than IPW.
\end{example}

\begin{example}[SPW-Based Generalized Residual Under One-Sided Limited Overlap]
    Suppose a weaker version of strict overlap holds, i.e., $e(X) < 1 - \underline{p}^* < 1$ a.s. without loss of generality. Since there is still the possibility of one-sided limited overlap (if $e(\cdot)$ is not bounded away from zero), IPW does not satisfy the SPW principle in this setting. However, an alternative that sets $\mathcal{S}(X, W; e) = W$ and $\xi \equiv 0$ satisfies the SPW principle with $M = (\underline{p}^*)^{-1}$ because $\E[W \mid X]/[e(X)(1-e(X))] = e(X)/[e(X)(1-e(X))] = 1/(1-e(X)) \leq (\underline{p}^*)^{-1}$. In this case, the generalized residual is given by $\Psi(Y, X, W; \tau, e, \xi) = W\,\tau(X) - (W - e(X))\,Y/(1-e(X))$. Another alternative that satisfies the SPW principle sets $\mathcal{S}(X, W; e) = W$ and $\xi(X, W; e, \mu_0) = (W - e(X))\mu_0(X)/(1-e(X))$, resulting in the generalized residual $\Psi(Y, X, W; \tau, e, \xi) = W\,\tau(X) - (W - e(X))(Y - \mu_0(X))/(1-e(X))$. This alternative is even better because it has a certain double robustness property discussed later.
\end{example}

\begin{theorem}[GNPW for CATE: Generalized Non-Inverse Probability Weighting for CATE]
    \label{theorem:gnpw_cate}
    Let $$e(X)^{\nu_1}(1-e(X))^{\nu_2}\Big[\big(\vartheta_1\,W+\vartheta_2\,e(X)+\vartheta_3\,W\,e(X)+\vartheta_4\,e(X)^2\big)\,\tau(X) - (W - e(X))\,Y\Big] + \xi(X, W; e, \eta)$$
    be the generalized residual $\Psi(Y, X, W; \tau, e, \eta)$, where $\xi: \mathbb{X} \times \mathbb{W} \to \R$ is a bounded function that depends on a bounded vector-valued nuisance function $\eta$ (on $\mathbb{X})$, such that $\E[\xi(X, W; e, \eta) \mid X] = 0$ a.s.~and the real numbers $(\nu_1, \nu_2, \vartheta_1, \vartheta_2, \vartheta_3, \vartheta_4)$ satisfy the following constraints: $\nu_1 \geq 0$, $\nu_2 \geq 0$, $\vartheta_1 + \vartheta_2 = 1$, and $\vartheta_3 + \vartheta_4 = -1$. Then, under Assumptions \ref{assumption:unconfoundedness}--\ref{assumption:bounded_covariates}, $\E[\Psi(Y, X, W; \tau, e, \eta) \mid X] = 0$ and $\mathbb{V}[\Psi(Y, X, W; \tau, e, \eta) \mid X] < \infty$ a.s. (even if strict overlap does not hold).
\end{theorem}
\begin{proof}
    The result follows by Theorem \ref{theorem:spw_moments} since $\Psi(Y, X, W; \tau, e, \eta)$ satisfies Definition \ref{definition:spw_cate} with $\mathcal{S}(X, W; e) = e(X)^{\nu_1}(1-e(X))^{\nu_2}\big(\vartheta_1\,W+\vartheta_2\,e(X)+\vartheta_3\,W\,e(X)+\vartheta_4\,e(X)^2\big)$ and $M = 1$: a.s., $\tilde{s}(X; e) = e(X)^{\nu_1-1}(1-e(X))^{\nu_2-1}\big[(\vartheta_1 + \vartheta_2)e(X)+(\vartheta_3+\vartheta_4)e(X)^2\big] = e(X)^{\nu_1}(1-e(X))^{\nu_2} \leq 1$.
\end{proof}

\begin{example}[NPW: Non-Inverse Probability Weighting]
    When $\nu_1 = \nu_2 = \nu \geq 0$, $\vartheta_1 = 0$, $\vartheta_2 = 1$, $\vartheta_3 = 0$, $\vartheta_4 = -1$, and $\xi \equiv 0$, the GNPW residual reduces to the Non-Inverse Probability Weighting (NPW) residual $[e(X)(1-e(X))]^{\nu}[e(X)(1-e(X))\tau(X) - (W - e(X))Y]$, which is the basis for the NPW estimators in Section \ref{section:alternatives}.
\end{example}

\begin{example}[Robinson Transformation]
    When $\nu_1 = \nu_2 = 0$, $\vartheta_1 = 1$, $\vartheta_2 = 0$, $\vartheta_3 = -2$, $\vartheta_4 = 1$, and $\xi(X, W; e, \eta) = (W-e(X))\eta(X)$ with $\eta(X) \equiv \E[Y \mid X]$, the corresponding GNPW residual reduces to $(W - e(X))^2\tau(X) - (W - e(X))(Y - \eta(X))$, which is the same residual based on ``Robinson transformation'' \citep{robinson1988root} used in machine learning-based causal inference.\footnote{The Robinson transformation-based residual is used in versions of the popular ``R-Learner'' \citep{kennedy2022towards, kennedy2022minimax, nie2021quasi, zhao2022selective, semenova2022estimation, chernozhukov2018double, oprescu2019orthogonal}. This suggests that the R-Learner can be made to work even under limited overlap (as discussed further in Subsection \ref{subsection:large_sample_spw}), although the existing literature \citep[see, e.g.,][]{kennedy2022minimax,kennedy2022towards, nie2021quasi} uses the assumption of strict overlap to derive the statistical properties of the R-Learner. Even if the R-Learner may be made robust to limited overlap, the R-Learner is not robust to misspecification of the propensity score function $e(\cdot)$, as discussed in Subsection \ref{subsection:aspw}. However, there are other (augmented) SPW-based alternatives (discussed later) that do not have that drawback. (Interestingly, \cite{wager2020causal} notes that the sample analog of the quantity $\E[(Y - \eta(X))(W - e(X))]/\E[(W - e(X))^2]$ is a practical choice for estimating PATE under limited overlap when there is no (or low) treatment effect heterogeneity. However, that estimator is inconsistent when there is substantial treatment effect heterogeneity. Nevertheless, as discussed in Subsection \ref{subsection:large_sample_spw}, the SPW principle can be used to consistently estimate BATE and PATE under limited overlap and heterogeneous treatment effects.)}
\end{example}

\begin{example}
    When $\nu_1 = \nu_2 = 0$, $\vartheta_1 = 0.5$, $\vartheta_2 = 0.5$, $\vartheta_3 = -1$, $\vartheta_4 = 0$, and $\xi \equiv 0$, the GNPW residual can be expressed as $0.5[W(1-e(X)) + e(X)(1-W)]\tau(X) - (W-e(X))Y$, which is the basis for another estimator discussed near the end of Section \ref{section:alternatives}. Alternatively, if $\vartheta_1 = 1$, $\vartheta_2 = 0$, $\vartheta_3 = -1$, and $\vartheta_4 = 0$, then the associated GNPW residual is $W(1-e(X))\tau(X) - (W - e(X))Y$. Another alternative GNPW residual is $e(X)(1-W)\tau(X) - (W - e(X))Y$, which results from setting $\vartheta_1 = 0$, $\vartheta_2 = 1$, $\vartheta_3 = -1$, and $\vartheta_4 = 0$. Of course, there are infinite other possibilities.
\end{example}

It would be desirable for the generalized residual $\Psi(Y, X, W; \tau, e, \eta)$ to satisfy an important property called ``Neyman orthogonality'' \citep{chernozhukov2017double, chernozhukov2018double, mackey2018orthogonal}, which enables application of some recent (but increasingly popular) nonparametric or machine learning methods, in addition to having zero conditional mean and finite conditional variance. The residual $\Psi(Y, X, W; \tau, e, \eta)$ satisfies Neyman orthogonality if the G\^ateaux derivative of its conditional expectation is zero almost surely with respect the last two (nuisance) arguments (evaluated at their true values $e$ and $\eta$) so that the conditional moment equality $\E[\Psi(Y, X, W; \tau, e, \eta) \mid X] = 0$ (a.s.) is insensitive to local deviations in the values of the nuisance parameters, as defined below.

\begin{definition}[Neyman Orthogonality of the Generalized Residual for CATE]
    \label{definition:neyman_orth} The generalized residual $\Psi(Y, X, W; \tau, e, \eta)$, where $(e, \eta)$ lie in the appropriate Banach space $\mathcal{B}_{(e,\eta)}$, satisfies the Neyman orthogonality condition if $\frac{d}{dt}\,\E[\Psi(Y, X, W; \tau, (e, \eta)+t(h_e, h_\eta)) \mid X]\,\big\lvert_{t= 0} = 0$ almost surely, where $t \in [0, 1)$ and $(h_e, h_\eta) \equiv (\tilde{e} - e, \tilde{\eta} - \eta) \in \mathcal{B}_{(e,\eta)}$, so that the residual is Neyman orthogonal.
\end{definition}

\begin{theorem}[Neyman Orthogonal Generalized Residual Under One-Sided Limited Overlap]
    \label{theorem:neyman_one_sided}
    Suppose that Assumptions \ref{assumption:unconfoundedness}--\ref{assumption:bounded_covariates} hold. If $0 < e(X) < 1 - \underline{p}^* < 1$ a.s., then the generalized residual $\Psi(Y, X, W; \tau, e, \mu_0) = W\,\tau(X) - (W - e(X))(Y - \mu_0(X))/(1- e(X))$ satisfies the SPW principle and Neyman orthogonality. If $0 < \underline{p}^* < e(X) < 1$ a.s., then an SPW principle-based Neyman orthogonal residual is $\Psi(Y, X, W; \tau, e, \mu_1) = (1-W)\,\tau(X) - (W - e(X))(Y - \mu_1(X))/e(X)$.
\end{theorem}

\begin{proof}
    Suppose that $0 < e(X) < 1 - \underline{p}^* < 1$ almost surely. Then, the generalized residual $\Psi(Y, X, W; \tau, e, \mu_0) = W\,\tau(X) - (W - e(X))(Y - \mu_0(X))/(1- e(X))$ satisfies the SPW principle in Definition \ref{definition:spw_cate} with $\mathcal{S}(X, W; e) = W$, $\xi(X, W; e, \eta) = (W - e(X))\eta(X)/(1-e(X))$, $\eta \equiv \mu_0$, and $M = (\underline{p}^*)^{-1}$ because $\tilde{s}(X; e) = \E[W \mid X]/[e(X)(1-e(X))] = e(X)/[e(X)(1-e(X))] = 1/(1-e(X)) \leq (\underline{p}^*)^{-1}$ almost surely. Let $E_t(x) = \E[\Psi(Y, X, W; \tau, (e, \eta)+t(h_e, h_\eta)) \mid X = x]$, where $t \in [0, 1)$ and $(h_e, h_\eta) \in \mathcal{B}_{(e,\eta)}$, for all $x \in \mathbb{X}$. Note that $E_t$ depends on $(e, \eta, h_e, h_\eta)$, where $\eta \equiv \mu_0$, but that dependence is suppressed for expositional ease. It is straightforward to show that $E_t \equiv e\tau - e(\mu_1 - \mu_0 - th_\eta) - (1-e)(-e-th_e)(\mu_0 - \mu_0 - th_\eta)/(1-e-th_e)$, which simplifies to $E_t \equiv teh_\eta - (1-e)(teh_\eta + t^2h_eh_\eta)/(1-e-th_e)$, and so $\frac{d}{dt}E_t \big\lvert_{t = 0} \equiv eh_\eta - eh_\eta \equiv 0$, proving that the generalized residual is Neyman orthogonal. Similar reasoning can be used to prove the last statement in the theorem (for the case where $0 < \underline{p}^* < e(X) < 1$ almost surely).
\end{proof}

\begin{theorem}[Neyman Orthogonal Generalized Residual Under Two-Sided Limited Overlap]
    \label{theorem:neyman_two_sided}
    Suppose that Assumptions \ref{assumption:unconfoundedness}--\ref{assumption:bounded_covariates} hold. Let $\Psi(Y, X, W; \tau, e, (\mu_0, \mu_1))$ be the GNPW residual
    $$\mathcal{S}(X, W; e)\,\tau(X) - (W - e(X))\,\big[Y - \mu_0(X) - \big(\vartheta_2+\vartheta_4\,e(X)\big)\big(\mu_1(X)-\mu_0(X)\big)\big],$$
    where $\mathcal{S}(X, W; e) \equiv \vartheta_1\,W+\vartheta_2\,e(X)+\vartheta_3\,W\,e(X)+\vartheta_4\,e(X)^2$ such that $(\vartheta_1, \vartheta_2, \vartheta_3, \vartheta_4) \in \R^4$ with $\vartheta_1 + \vartheta_2 = 1$ and $\vartheta_3 + \vartheta_4 = -1$. Then, the generalized residual $\Psi(Y, X, W; \tau, e, (\mu_0, \mu_1))$ satisfies the SPW principle and Neyman orthogonality.
\end{theorem}
\begin{proof}
    The generalized residual $\Psi(Y, X, W; \tau, e, (\mu_0, \mu_1))$ is a GNPW (generalized non-inverse probability weighting) residual that satisfies the SPW principle, because $\Psi(Y, X, W; \tau, e, (\mu_0, \mu_1))$ satisfies the conditions of Theorem \ref{theorem:gnpw_cate} with $\eta \equiv (\mu_0, \mu_1)$, $\nu_1 = \nu_2 = 0$, and $\xi(X, W; e, \eta) = (W-e(X))[\mu_0(X) + (\vartheta_2+\vartheta_4\,e(X))(\mu_1(X)-\mu_0(X))]$ such that (a.s.) $\E[\xi(X, W; e, \eta) \mid X] = 0$. To show Neyman orthogonality, let $E_t(x) = \E[\Psi(Y, X, W; \tau, (e, \mu_0, \mu_1)+t(h_e, h_{\mu_0}, h_{\mu_1})) \mid X = x]$, where $t \in [0, 1)$ and $(h_e, h_{\mu_0}, h_{\mu_1}) \in \mathcal{B}_{(e,\eta)}$, for all $x \in \mathbb{X}$. Note that the dependence of $E_t$ on $(e, \mu_0, \mu_1, h_e, h_{\mu_0}, h_{\mu_1})$ is suppressed for expositional ease. It is straightforward to show that $E_t \equiv [(\vartheta_1 + \vartheta_2)e + \vartheta_2th_e + (\vartheta_3 + \vartheta_4)e^2 + \vartheta_3teh_e + \vartheta_4(2teh_e + t^2h_e^2)](\mu_1 - \mu_0) - {e(1-e-th_e)\mu_1} + (1-e)(e+th_e)\mu_0 + th_e[-(\vartheta_2 + \vartheta_4e + \vartheta_4th_e)(\mu_1 + th_{\mu_1} - \mu_0 - th_{\mu_0}) - \mu_0 - th_{\mu_0}]$, which simplifies to $E_t \equiv t^2[\vartheta_4 h_e^2(\mu_1 - \mu_0) + (\vartheta_2 + \vartheta_4e)h_e(h_{\mu_1} - h_{\mu_0}) + h_e^2(\mu_1 - \mu_0) - h_eh_{\mu_0}] + t^3[\vartheta_4h_e^2(h_{\mu_1} - h_{\mu_0})]$ using the equalities $\vartheta_1 + \vartheta_2 = 1$ and $\vartheta_3 + \vartheta_4 = -1$. Thus, $\frac{d}{dt}E_t \big\lvert_{t = 0} \equiv 0$, proving that the GNPW residual $\Psi(Y, X, W; \tau, e, (\mu_0, \mu_1))$ is Neyman orthogonal while satisfying the SPW principle.
\end{proof}

\begin{example}[Orthogonal Non-Inverse Probability Weighting]
    Setting $\vartheta_2 = 1$ and $\vartheta_4 = -1$ (so that $\vartheta_1 = \vartheta_3 = 0$) in Theorem \ref{theorem:neyman_two_sided} results in the following Neyman orthogonal NPW residual: $\Psi(Y, X, W; \tau, e, (\mu_0, \mu_1)) = e(X)(1-e(X))\tau(X) - (W - e(X))[Y - (1-e(X))\mu_1(X) - e(X)\mu_0(X)]$.
\end{example}

\begin{example}[Robinson Transformation]
    Setting $\vartheta_1 = 1$, $\vartheta_2 = 0$, $\vartheta_3 = -2$, and $\vartheta_4 = 1$ gives $\Psi(Y, X, W; \tau, e, (\mu_0, \mu_1)) = (W - e(X))^2\tau(X) - (W - e(X))[Y - e(X)\mu_1(X) - (1-e(X))\mu_0(X)]$ using Theorem \ref{theorem:neyman_two_sided}. At the true values of the parameters, this residual coincides with the Robinson transformation-based residual $\Psi(Y, X, W; \tau, e, \breve{\eta}) = (W - e(X))^2\tau(X) - (W - e(X))(Y - \breve{\eta}(X))$, which uses the nuisance function $\breve{\eta}(X) \equiv \E[Y \mid X]$ instead of the vector-valued nuisance function $(\mu_0, \mu_1)$. Thus, the $\breve{\eta}$-based residual $\Psi(Y, X, W; \tau, e, \breve{\eta})$ is Neyman orthogonal. However, it is highly sensitive to misspecification of the propensity score function, as shown in Subsection \ref{subsection:aspw}. 
\end{example}

\begin{example}
    Setting $\vartheta_1 = 1$, $\vartheta_2 = 0$, $\vartheta_3 = -1$, and $\vartheta_4 = 0$ results in the orthogonalized GNPW residual $W(1-e(X))\tau(X) - (W - e(X))(Y - \mu_0(X))$. Alternatively, setting $\vartheta_1 = 0$, $\vartheta_2 = 1$, $\vartheta_3 = -1$, and $\vartheta_4 = 0$ gives the orthogonal residual $e(X)(1-W)\tau(X) - (W - e(X))(Y - \mu_1(X))$.
\end{example}

The discussion so far has focused on the SPW principle, which relies on the propensity score function $e(\cdot)$, but it is also possible to construct generalized residuals that are robust to limited overlap without using $e(\cdot)$. For this purpose, it is useful to define a broader ``Stable Residualization Principle'' (SRP), which nests the SPW principle.

\begin{definition}[SRP for CATE: Stable Residualization Principle for CATE]
    \label{definition:srp_cate}
    The generalized residual $\Psi(Y, X, W; \tau, e, \eta)$ satisfies the Stable Residualization Principle (SRP) for CATE if
    $$\Psi(Y, X, W; \tau, e, \eta) = \Psi_1(X, W; e, \eta)\,\tau(X) - \Psi_2(X, W; e, \eta)\,Y - \Psi_3(X, W; e, \eta)$$
    such that $\Psi_1, \Psi_2, \Psi_3$ are bounded real-valued functions on $\mathbb{X} \times \mathbb{W}$ that depend on $e$ and $\eta$ (a bounded vector-valued function on $\mathbb{X}$) such that the following hold a.s.: $\E[\Psi_1(X, W; e, \eta) \mid X] \neq 0$; and $\E[\Psi_1(X, W; e, \eta) \mid X]\,\tau(X) = \E[\Psi_2(X, W; e, \eta)\,\mu(W,X) \mid X] + \E[\Psi_3(X, W; e, \eta) \mid X]$.
\end{definition}

\begin{theorem}[Conditional Moments of the SRP-Based Generalized Residual]
    If the generalized residual $\Psi(Y, X, W; \tau, e, \eta)$ satisfies the SRP condition (in Definition \ref{definition:srp_cate}) and Assumptions \ref{assumption:unconfoundedness}--\ref{assumption:bounded_covariates} hold, then $\E[\Psi(Y, X, W; \tau, e, \eta) \mid X] = 0$ and $\mathbb{V}[\Psi(Y, X, W; \tau, e, \eta) \mid X] < \infty$ almost surely.
\end{theorem}
\begin{proof}
    Note that $\E[\Psi(Y, X, W; \tau, e, \eta) \mid X] = \E[\Psi_2(X, W; e, \eta)\,(Y - \mu(W,X)) \mid X] = 0$ a.s.~because $\E[\Psi_2(X, W; e, \eta)\,Y \,|\, X] = \E[\E[\Psi_2(X, W; e, \eta)\,Y \,|\, X, W] \,|\, X] = \E[\Psi_2(X, W; e, \eta)\,\E[Y \,|\, X, W] \,|\, X]$, where $\E[Y \,|\, X, W] = \E[Y^*_W \,|\, X, W] = \E[Y^*_W \,|\, X] = \mu(W, X)$ by Assumptions \ref{assumption:unconfoundedness}--\ref{assumption:pop_overlap}. In addition, $\mathbb{V}[\Psi(Y, X, W; \tau, e, \eta) \mid X] = \E[\Psi(Y, X, W; \tau, e, \eta)^2 \mid X] < \infty$ a.s.~by Assumptions \ref{assumption:bounded_outcomes}--\ref{assumption:bounded_covariates}.
\end{proof}

\begin{example}[SPW Satisfies SRP]
    Under Assumptions \ref{assumption:unconfoundedness}--\ref{assumption:bounded_covariates}, if the residual $\Psi(Y, X, W; \tau, e, \eta)$ obeys the SPW principle, then it also satisfies the SRP principle with $\Psi_1(X, W; e, \eta) = \mathcal{S}(X, W; e)$, $\Psi_2(X, W; e, \eta) = \tilde{s}(X; e)\,(W-e(X))$, and $\Psi_3(X, W; e, \eta) = -\xi(X, W; e, \eta)$.
\end{example}

\begin{example}[Generalized Robinson Transformation Satisfies SRP]
    Suppose $\eta(X) \equiv \E[Y \mid X]$, $\Psi_1(X, W; e, \eta) = |W - e(X)|^\nu$, $\Psi_2(X, W; e, \eta) =\mathrm{sgn}(2W - 1)\,|W - e(X)|^{\nu-1}$, and  $\Psi_3(X, W; e, \eta) =-\Psi_2(X, W; e, \eta)\,\eta(X)$ so that $|W - e(X)|^\nu\tau(X) - \mathrm{sgn}(2W - 1)\,|W - e(X)|^{\nu-1}(Y - \eta(X))$ is the generalized residual $\Psi(Y, X, W; \tau, e, \eta)$. When $\nu = 2$, it equals the usual Robinson transformation-based residual, which underpins the ``R-Learner'' \citep{kennedy2022minimax}, and is thus nested within the SPW framework. More generally, when $\nu \geq 1$, the residual $\Psi(Y, X, W; \tau, e, \eta)$, based on a generalized version of Robinson transformation, satisfies the SRP condition under Assumptions \ref{assumption:unconfoundedness}--\ref{assumption:bounded_covariates} (even without strict overlap). However, when strict overlap holds, $\Psi(Y, X, W; \tau, e, \eta)$ satisfies SRP for any $\nu \in \R$; e.g., setting $\nu = 0$ gives the ``U-Learner'' based on $\tau(X) - (Y - \eta(X))/(W - e(X))$.
\end{example}

\begin{example}[Residualization Satisfying SRP Without Propensity Scores]
    Suppose $\eta \equiv \mu_0$, $\Psi_1(X, W; e, \eta) = W$, $\Psi_2(X, W; e, \eta) = 1$, and  $\Psi_3(X, W; e, \eta) =-\eta(X)$ so that the generalized residual $\Psi(Y, X, W; e, \eta) = W\,\tau(X) - (Y - \mu_0(X))$ does not depend on $e(\cdot)$. Under Assumptions \ref{assumption:unconfoundedness}--\ref{assumption:bounded_covariates}, $\E[\Psi(Y, X, W; e, \eta) \mid X] = e(X)\tau(X) - e(X)\mu_1(X) - (1-e(X))\mu_0(X) + \mu_0(X) = 0$ almost surely with finite conditional variance. Similar reasoning can be used to show that the generalized residual $\Psi(Y, X, W; e, \eta) = (W-1)\,\tau(X) - (Y - \mu_1(X))$ with $\eta \equiv \mu_1$ also satisfies SRP under the same assumptions. More generally, when $\eta \equiv (\mu_0, \mu_1)$, the generalized residual 
    $$\Psi(Y, X, W; e, (\mu_0, \mu_1)) = [\vartheta_1\,W + \vartheta_2(W-1)]\,\tau(X) - (\vartheta_1+\vartheta_2)\,Y + \vartheta_1\,\mu_0(X)) + \vartheta_2\,\mu_1(X)$$
    satisfies the SRP condition under Assumptions \ref{assumption:unconfoundedness}--\ref{assumption:bounded_covariates} for any $(\vartheta_1, \vartheta_2) \in \R^2_{\geq 0}\setminus \{(0,0)\}$. Thus, the above residual, which does not rely on the propensity score function, is robust to limited overlap but does not have Neyman orthogonality or robustness to misspecification of the nuisance functions.
\end{example}

The generalized residuals discussed so far treat the CATE function $\tau(\cdot)$ as the target parameter and the functions $(\mu_0(\cdot), \mu_1(\cdot))$ as nuisance parameters (in addition to $e(\cdot)$, the propensity score function). However, when $(\mu_0(\cdot), \mu_1(\cdot))$ is itself of interest, another SPW principle can be used.

\begin{definition}[Stable Probability Weighting for Conditional Treatment and Control Means]
    \label{definition:spw_ctm_ccm}
    The generalized residual vector $\Psi(Y, X, W; \mu, e, \eta) = (\widetilde{\Psi}_w(Y, X, W; \mu_w, e, \eta))_{w \,\in\, \{0, 1\}}$ satisfies the SPW principle for CTM and CCM (Conditional Treatment\,/\,Control Mean) $(\mu_w(\cdot))_{w\, \in\, \{0, 1\}}$ if $$\widetilde{\Psi}_w(Y, X, W; \mu_w, e, \eta) = \mathcal{S}_w(X, W; e)\,\mu_w(X) - \tilde{s}_w(X; e)\,\mathbb{I}\{W = w \}Y + \xi_w(X, W; e, \eta)$$
    such that $\mathcal{S}_w$ and $\xi_w$ are bounded real-valued functions on $\mathbb{X} \times \mathbb{W}$ that depend on $e$ and $\eta$ (a bounded vector-valued function on $\mathbb{X}$) and satisfy the following conditions for each $w\in \{0,1\}$: $\E[\xi_w(X, W; e, \eta) \mid X] = 0$ almost surely; and $\tilde{s}_w(X; e) \equiv s_w(X; e)/[e(X)^w(1-e(X))^{1-w}]$, where $s_w(X; e) \equiv \E[\mathcal{S}_w(X, W; e) \mid X]$, such that $0 < \tilde{s}_w(X; e) \leq M_w$ almost surely for some $M_w > 0$.
\end{definition}

\begin{example}
    The residual vector $(\mu_0(X) - (1-W)Y/(1-e(X)), \mu_1(X) - W\,Y/e(X))$ satisfies the SPW principle for CTM and CCM with $\xi_w \equiv 0$, $\mathcal{S}_w(X, W; e) = 1$, and $M_w = (\underline{p}^*)^{-1}$ for $w \in \{0, 1\}$ under strict overlap, in which case $1/e(X) \leq (\underline{p}^*)^{-1}$ and $1/[1-e(X)] \leq (\underline{p}^*)^{-1}$. A more efficient method would use $\mathcal{S}_w(X, W; e) = w\,\overline{e}\,W/e(X) + (1-w) (1-\overline{e})(1-W)/(1-e(X))$, where $\overline{e} = \E[e(X)]$, so that $\tilde{s}_w(X; e) = w\,\overline{e}/e(X) + (1-w) (1-\overline{e})/(1-e(X)) \leq M_w$, where $M_w = w\,\overline{e}/\underline{p}^* + (1-w)(1-\overline{e})/\underline{p}^* < (\underline{p}^*)^{-1}$ for $w \in \{0,1\}$. In the literature, $\mathcal{S}_w(X, W; e)$ is called a ``stabilized weight'' \citep[see][]{ai2021unified,hernan2020causal} in the context of ``marginal structural models'' under strict overlap. This choice, along with $\xi_0 \equiv \xi_1 \equiv 0$, leads to the residual vector $\Psi(Y, X, W; \mu, e, \eta) = ((1-\overline{e})(1-W)(\mu_0(X) - Y)/(1-e(X)), \overline{e}\,W\,(\mu_1(X) - Y)/e(X))$.
\end{example}

\begin{example}[Overlap Weights Satisfy SPW Principle for CTM\,/\,CCM Under Limited Overlap]
    \cite{li2019propensity} and \cite{li2018balancing} propose ``overlap weights'' for balancing covariates between the treatment and control groups and for estimating a (weighted) ATE for the ``overlap population.'' The overlap weights are given by $\mathcal{S}_w(X, W; e) = w\,W(1-e(X)) + (1-w)(1-W)\,e(X)$ for $w \in \{0,1\}$. Then, $\tilde{s}_w(X; e) = w\,(1-e(X)) + (1-w)\,e(X) \leq 1$ a.s.~(even under limited overlap), leading to the residuals $(e(X)(1-W)(\mu_0(X) - Y), W(1-e(X))(\mu_1(X) - Y))$ when $\xi_0 \equiv \xi_1 \equiv 0$. 
\end{example}

\begin{example}[Alternative Stable Probability Weights for CTM\,/\,CCM Under Limited Overlap]
    Under limited overlap, the simpler weights $\mathcal{S}_w(X, W; e) = w\,W + (1-w)(1-W)$ for $w \in \{0,1\}$ also satisfy the SPW principle because $\tilde{s}_w(X; e) = w\cdot 1 + (1-w)\cdot 1 \leq 1$ a.s., leading to the residual vector $((1-W)(\mu_0(X) - Y), W(\mu_1(X) - Y))$ when $\xi_0 \equiv \xi_1 \equiv 0$. Another possibility is to use the weights $\mathcal{S}_w(X, W; e) = w\,e(X) + (1-w)(1-e(X))$ for $w \in \{0,1\}$ along with $\xi_w(X, W; e, \eta) =w(W-e(X))\eta_1(X) + (1-w)\,(e(X)-W)\eta_0(X)]$ for $w \in \{0, 1\}$ such that ${\eta(X) \equiv (\eta_0(X), \eta_1(X)) \equiv (\E[Y \mid W = 0, X], \E[Y \mid W = 1, X])}$, resulting in the residual vector $((1-e(X))\mu_0(X) - (1-W)Y + (e(X)-W)\eta_0(X),\,\, e(X)\mu_1(X) - W\,Y + (W-e(X))\eta_1(X))$.
\end{example}

\subsection{Augmented Conditional Moment Restrictions for Double Robustness}
\label{subsection:aspw}

I now return to case where CATE is of interest and CTM\,/\,CCM are treated as nuisance components. Robustifying the conditional moment restrictions (for CATE) to the nuisance parameters is an important consideration in the context of semi/non-parametric (as well as parametric) statistical methods. While the Neyman orthogonality property discussed in the previous subsection provides a certain form of ``local'' robustness (i.e., low sensitivity to marginally minute errors in the nuisance components around their true values), it does not cover the case where there may be systematic mistakes in one of the nuisance functions. It would be desirable to construct conditional moment restrictions that have ``double robustness'' to misspecification of one of the nuisance components. There is a basic version of such property (assessed at the true value of CATE) and also a more global version (assessed in the space where CATE lies). Both versions are formalized below.

\begin{definition}[BDR: Basic Double Robustness]
    \label{definition:bdr}
        The generalized residual $\Psi(Y, X, W; \theta, \zeta, \eta)$, where $\theta \in \mathcal{B}_\theta$ is the target parameter and $(\zeta, \eta) \in \mathcal{B}_\zeta \times \mathcal{B}_\eta$ are nuisance parameters in the appropriate parameter spaces, satisfies BDR and is thus ``doubly robust'' (in a basic sense) if
        $$\tilde{\theta} = \theta \iff \E[\Psi(Y, X, W; \tilde{\theta}, \tilde{\zeta}, \eta) \mid X] = \E[\Psi(Y, X, W; \tilde{\theta}, \zeta, \tilde{\eta}) \mid X] = 0$$
        almost surely for all $\tilde{\theta} \in \mathcal{B}_\theta$, $\tilde{\zeta} \in \mathcal{B}_\zeta$, and $\tilde{\eta} \in \mathcal{B}_\eta$, where $\theta$, $\zeta$, and $\eta$ are the true parameter values.
\end{definition}

\begin{definition}[GDR: Global Double Robustness]
    \label{definition:gdr}
    The generalized residual $\Psi(Y, X, W; \theta, \zeta, \eta)$, where $\theta \in \mathcal{B}_\theta$ is the target parameter and $(\zeta, \eta) \in \mathcal{B}_\zeta \times \mathcal{B}_\eta$ are nuisance parameters in the appropriate parameter spaces, satisfies GDR and is thus ``globally double robust'' if
    $$\E[\Psi(Y, X, W; \tilde{\theta}, \zeta, \eta) \mid X] = \E[\Psi(Y, X, W; \tilde{\theta}, \tilde{\zeta}, \eta) \mid X] = \E[\Psi(Y, X, W; \tilde{\theta}, \zeta, \tilde{\eta}) \mid X]$$
    almost surely for all $\tilde{\theta} \in \mathcal{B}_\theta$, $\tilde{\zeta} \in \mathcal{B}_\zeta$, and $\tilde{\eta} \in \mathcal{B}_\eta$ in addition to satisfying the BDR property.
\end{definition}

The BDR property in Definition \ref{definition:bdr} is simply a conditional version of the unconditional version usually defined in the literature on ``doubly robust'' methods \citep{robins2001comment,rothe2019properties,chaudhuri2019review}. Similarly, the GDR property is a generalization of the unconditional version defined by \cite{chernozhukov2022locally} and is much stronger than BDR.

\begin{definition}[ASPW for CATE: Augmented Stable Probability Weighting for CATE]
    \label{definition:aspw_cate}
    The generalized residual $\Psi(Y, X, W; \tau, e, \eta)$ satisfies the Augmented SPW (ASPW) principle for CATE $\tau(\cdot)$ if $\Psi(Y, X, W; \tau, e, \eta)$ satisfies the SPW principle as well as the BDR principle.
\end{definition}

\begin{remark}[Robinson Transformation-Based Residual Does Not Have Double Robustness]
    \label{remark:robinson_not_dr}
    Recall that $\Psi(Y, X, W; \tau, e, \eta) = (W - e(X))^2\tau(X) - (W - e(X))(Y - \eta(X))$, where $\eta(X) \equiv \E[Y \mid X]$, is the widely used residual based on Robinson transformation. Under Assumptions \ref{assumption:unconfoundedness}--\ref{assumption:bounded_covariates}, it is easy to show that $\E[\Psi(Y, X, W; \tau, e, \tilde{\eta}) \mid X] = 0$ a.s.~because $\E[(W - e(X))\tilde{\eta}(X) \mid X] = (e(X) - e(X))\tilde{\eta}(X) = 0$ a.s.~for any $\tilde{\eta} \in \mathcal{B}_\eta$, and so the Robinson transformation is robust to misspecification of $\eta$. However, the residual is not robust to misspecification of the propensity score function $e(\cdot)$. To see this, let $E_{\tilde{e}}(X) \equiv \E[\Psi(Y, X, W; \tau, \tilde{e}, \eta) \mid X]$ for any $\tilde{e} \in \mathcal{B}_e$. Then, $E_{\tilde{e}} \equiv [e(1-\tilde{e})^2 + (1-e)\tilde{e}^2]\tau - e(1-\tilde{e})(\mu_1 - \mu_0 - e\,\tau) + (1-e)\tilde{e}(\mu_0 - \mu_0 - e\,\tau)$, which simplifies to $E_{\tilde{e}} \equiv [e - 2e\tilde{e} + \tilde{e}^2]\tau - e(1-\tilde{e})(1-e)\tau - e(1-e)\tilde{e}\,\tau \equiv [e - 2e\tilde{e} + \tilde{e}^2]\tau - [e - e^2]\tau$, and so $E_{\tilde{e}} \equiv (e - \tilde{e})^2\tau \not\equiv 0$ in general if $e \not\equiv \tilde{e}$ (unless $\tau \equiv 0$). Therefore, the Robinson transformation-based residual does not satisfy BDR in general by not being robust to systematic (i.e., non-local) mistakes in the argument dedicated to $e(\cdot)$, even though it is robust to misspecification of $\eta(\cdot)$. This fact helps explain why ``there seems to be some important asymmetry in the role of these two nuisance functions for the CATE'' \citep{kennedy2022minimax} when using the Robinson transformation.
\end{remark}

The above remark shows that the Neyman orthogonal generalized residual based on the Robinson transformation does not satisfy the ASPW principle, even under strict overlap. When overlap is strong, the AIPW residual (defined as CATE minus the AIPW pseudo-outcome) can be used instead; see Definition \ref{definition:aipw} and the remark below it. However, even under limited overlap, there is a simple SPW-based way to fix the usual form of Robinson transformation in order to achieve BDR. There are also several other ways to construct generalized residuals that satisfy the ASPW principle under one-sided or two-sided limited overlap, as discussed below.

\begin{theorem}[Globally Double Robust Generalized Residual Under One-Sided Limited Overlap]
    \label{theorem:gdr_one_sided}
    Suppose that Assumptions \ref{assumption:unconfoundedness}--\ref{assumption:bounded_covariates} hold. If $0 < e(X) < 1 - \underline{p}^* < 1$ a.s., then the generalized residual $\Psi(Y, X, W; \tau, e, \mu_0) = W\,\tau(X) - (W - e(X))(Y - \mu_0(X))/(1- e(X))$ satisfies the ASPW principle and GDR in addition to Neyman orthogonality. If $0 < \underline{p}^* < e(X) < 1$ a.s., then $\Psi(Y, X, W; \tau, e, \mu_1) = (1-W)\,\tau(X) - (W - e(X))(Y - \mu_1(X))/e(X)$ has those properties.
\end{theorem}

\begin{proof}
    Suppose that $0 < e(X) < 1 - \underline{p}^* < 1$ almost surely. Then, the generalized residual $\Psi(Y, X, W; \tau, e, \mu_0) = W\,\tau(X) - (W - e(X))(Y - \mu_0(X))/(1- e(X))$ satisfies SPW and Neyman orthogonality by Theorem \ref{theorem:neyman_one_sided}. To show GDR, let $E[\tilde{\tau}, \tilde{e}, \tilde{\mu}_0](X) \equiv \E[\Psi(Y, X, W; \tilde{\tau}, \tilde{e}, \tilde{\mu}_0) \mid X]$ for any $\tilde{\tau} \in \mathcal{B}_\tau$, $\tilde{e} \in \mathcal{B}_e$, and $\tilde{\eta} \in \mathcal{B}_\eta$. Then, $E[\tilde{\tau}, \tilde{e}, \tilde{\mu}_0] \equiv e\,\tilde{\tau} - e(\mu_1 - \tilde{\mu}_0) + (1-e)\tilde{e}(\mu_0 - \tilde{\mu}_0)/(1 - \tilde{e})$, and so it is straightforward to see that $E[\tilde{\tau}, e, \tilde{\mu}_0] \equiv E[\tilde{\tau}, \tilde{e}, \mu_0] \equiv E[\tilde{\tau}, e, \mu_0] \equiv e(\tilde{\tau} - \tau) \equiv 0$ if and only if $\tilde{\tau} \equiv \tau$ because $e$ is positive (by Assumption \ref{assumption:pop_overlap}). Therefore, the generalized residual also satisfies GDR and the ASPW principle. Similar reasoning can be used to prove the last
    statement in the theorem (for the case where $0 < \underline{p}^* < e(X) < 1$ almost surely).
\end{proof}

\begin{theorem}[Doubly Robust GNPW Generalized Residual Under Two-Sided Limited Overlap]
    \label{theorem:gnpw_dr}
    Suppose that Assumptions \ref{assumption:unconfoundedness}--\ref{assumption:bounded_covariates} hold. Let $\Psi(Y, X, W; \tau, e, (\mu_0, \mu_1))$ be the GNPW residual
    $$\mathcal{S}(X, W; e)\,\tau(X) - (W - e(X))\,\big[Y - \mu_0(X) - \big(\vartheta_2+\vartheta_4\,e(X)\big)\big(\mu_1(X)-\mu_0(X)\big)\big],$$
    where $\mathcal{S}(X, W; e) \equiv \vartheta_1\,W+\vartheta_2\,e(X)+\vartheta_3\,W\,e(X)+\vartheta_4\,e(X)^2$ such that $(\vartheta_1, \vartheta_2, \vartheta_3, \vartheta_4) \in \R^4$ with $\vartheta_1 + \vartheta_2 = 1$ and $\vartheta_3 + \vartheta_4 = -1$. Then, the generalized residual $\Psi(Y, X, W; \tau, e, (\mu_0, \mu_1))$ is doubly robust and satisfies the ASPW principle in addition to Neyman orthogonality.
\end{theorem}
\begin{proof}
    The generalized residual $\Psi(Y, X, W; \tau, e, \eta)$, where $\eta \equiv (\mu_0, \mu_1)$, is a GNPW (generalized non-inverse probability weighting) residual that satisfies the SPW principle and Neyman orthogonality by Theorem \ref{theorem:neyman_two_sided}. To show BDR, let $E[\tilde{\tau}, \tilde{e}, \tilde{\eta}](X) \equiv \E[\Psi(Y, X, W; \tilde{\tau}, \tilde{e}, \tilde{\eta}) \mid X]$ for any $\tilde{\tau} \in \mathcal{B}_\tau$, $\tilde{e} \in \mathcal{B}_e$, and $\tilde{\eta} \in \mathcal{B}_\eta$. Then, it follows that $E[\tilde{\tau}, \tilde{e}, \tilde{\eta}] \equiv (\vartheta_1e + \vartheta_2 \tilde{e} + \vartheta_3 e\tilde{e} + \vartheta_4 \tilde{e}^2)\tilde{\tau} - e(1-\tilde{e})[\mu_1 - \tilde{\mu}_0 - (\vartheta_2 + \vartheta_4\tilde{e})(\tilde{\mu}_1 - \tilde{\mu}_0)] + (1-e)\tilde{e}[\mu_0 - \tilde{\mu}_0 - (\vartheta_2 + \vartheta_4 \tilde{e})(\tilde{\mu}_1 - \tilde{\mu}_0)]$, implying that $E[\tilde{\tau}, e, \tilde{\eta}] \equiv e(1-e)(\tilde{\tau} - \tau)$ and $E[\tilde{\tau}, \tilde{e}, \eta] \equiv (\vartheta_1 e + \vartheta_2 \tilde{e} + \vartheta_3 e \tilde{e} + \vartheta_4 \tilde{e}^2)(\tilde{\tau} - \tau)$. Since $e \in (0, 1)$ and $\tilde{e} \in (0, 1)$, $E[\tilde{\tau}, e, \tilde{\eta}] \equiv E[\tilde{\tau}, \tilde{e}, \eta] \equiv 0 \iff \tilde{\tau} \equiv \tau$, satisfying the ASPW principle.
\end{proof}

The above proof not only shows that GNPW residual has basic double robustness (BDR) but also implicitly shows that it does not have global double robustness (GDR) in general. Nevertheless, it does satisfy the ASPW principle. The above theorem also provides a way to fix the Robinson transformation-based residual (RT residual) to make it doubly robust. The generalized residual $(W - e(X))^2\tau(X) - (W - e(X))(Y - e(X)\mu_1(X) - (1-e(X))\mu_0(X))$ is the GNPW residual with $(\vartheta_1, \vartheta_2, \vartheta_3, \vartheta_4) = (1, 0, -2, 1)$ and has a vector-valued nuisance function $(\mu_0(\cdot), \mu_1(\cdot))$ in addition to $e(\cdot)$. Since this residual satisfies the ASPW principle, it can be used instead of the usual RT residual given by $(W - e(X))^2\tau(X) - (W - e(X))(Y - \eta(X))$, which only uses one nuisance function $\eta(X) \equiv \E[Y \mid X]$ in addition to $e(\cdot)$. As argued in Remark \ref{remark:robinson_not_dr}, the RT residual does not have double robustness because correctly specifying $\eta(\cdot)$, which implicitly involves the true value of the propensity score function, while misspecifying $e(\cdot)$ leads to a systematic imbalance that leads to a nonzero conditional expectation of the RT residual in that case. Thus, surprisingly, when $e(\cdot)$ is misspecified as $\tilde{e}(\cdot)$, it is better to use the residual $(W - \tilde{e}(X))^2\tau(X) - (W - \tilde{e}(X))(Y - \tilde{e}(X)\mu_1(X) - (1-\tilde{e}(X))\mu_0(X))$ than to use $(W - \tilde{e}(X))^2\tau(X) - (W - \tilde{e}(X))(Y - \eta(X))$, where $\eta(X) \equiv e(X)\mu_1(X) + (1-e(X))\mu_0(X)$ is the correct conditional expectation rather than the incorrect version $\tilde{e}(X)\mu_1(X) + (1-\tilde{e}(X))\mu_0(X)$; the latter is nevertheless deliberately exploited by the GNPW residual to achieve double robustness. Another generalized residual that satisfies the ASPW principle (but not GDR) is given below.

\begin{theorem}[Weighted AIPW Residual Satisfies Neyman Orthogonality and ASPW Principle]
    \label{theorem:aspw_bdr}
    If Assumptions \ref{assumption:unconfoundedness}--\ref{assumption:bounded_covariates}, then, the weighted AIPW residual $\Psi(Y, X, W; \tau, e, (\mu_0, \mu_1))$ equal to $$e(X)(1-e(X))\tau(X) - e(X)(1-e(X))(\mu_1(X) - \mu_0(X)) - (W - e(X))(Y - W\mu_1(X) - (1-W)\mu_0(X))$$ is doubly robust and satisfies the ASPW principle in addition to Neyman orthogonality.
\end{theorem}

\begin{proof}
    Let $E[\tilde{\tau}, \tilde{e}, \tilde{\eta}](X) \equiv \E[\Psi(Y, X, W; \tilde{\tau}, \tilde{e}, \tilde{\eta}) \mid X]$, where $\tilde{\eta} \equiv (\tilde{\mu}_0, \tilde{\mu}_1)$, for any $\tilde{\tau} \in \mathcal{B}_\tau$, $\tilde{e} \in \mathcal{B}_e$, and $\tilde{\eta} \in \mathcal{B}_\eta$. Then, $E[\tilde{\tau}, \tilde{e}, \tilde{\eta}] \equiv \tilde{e}(1-\tilde{e})(\tilde{\tau} - \tilde{\mu}_1 + \tilde{\mu}_0) - e(1-\tilde{e})(\mu_1 - \tilde{\mu}_1) + (1-e)\tilde{e}(\mu_0 - \tilde{\mu}_0)$. Note that $E[\tilde{\tau}, e, \tilde{\eta}] \equiv e(1-e)(\tilde{\tau} - \tau)$ and $E[\tilde{\tau}, \tilde{e}, \eta] \equiv \tilde{e}(1-\tilde{e})(\tilde{\tau} - \tau)$. Since both $e(1-e)$ and $\tilde{e}(1-\tilde{e})$ lie in $(0, 0.25)$, $E[\tilde{\tau}, e, \tilde{\eta}] \equiv E[\tilde{\tau}, \tilde{e}, \eta] \equiv 0 \iff \tilde{\tau} \equiv \tau$, satisfying the BDR property. To show Neyman orthogonality, let $E_t(x) = \E[\Psi(Y, X, W; \tau, (e, \mu_0, \mu_1)+t(h_e, h_{\mu_0}, h_{\mu_1})) \mid X = x]$, where $t \in [0, 1)$ and $(h_e, h_{\mu_0}, h_{\mu_1}) \in \mathcal{B}_{(e,\eta)}$, for all $x \in \mathbb{X}$. Note that the dependence of $E_t$ on $(e, \mu_0, \mu_1, h_e, h_{\mu_0}, h_{\mu_1})$ is suppressed for expositional ease. It is straightforward to show that $E_t \equiv [(e+th_e) - (e + th_e)^2][\tau - \mu_1 - th_{\mu_1} + \mu_0 + th_{\mu_0}] - e(1-e-th_e)(- th_{\mu_1}) +(1-e)(e + th_e)(- th_{\mu_0})$, which simplifies to $E_t \equiv t^2[(e-1)h_eh_{\mu_1} - eh_eh_{\mu_0}] + t^3[h_e^2(h_{\mu_1} - h_{\mu_0})]$. Thus, $\frac{d}{dt}E_t \big\lvert_{t = 0} \equiv 0$, proving that the weighted AIPW residual $\Psi(Y, X, W; \tau, e, (\mu_0, \mu_1))$ is also Neyman orthogonal. Of course, it also satisfies the SPW principle with $\mathcal{S}(X, W; e) = e(X)(1-e(X))$ so that $\tilde{s}(X; e) = 1$ a.s.~and $\xi(X, W; e, \eta) = (W - e(X))(W\mu_1(X) + (1-W)\mu_0(X)) - e(X)(1-e(X))(\mu_1(X) - \mu_0(X))$.
\end{proof}

The residuals discussed in Theorems \ref{theorem:gnpw_dr} and \ref{theorem:aspw_bdr} are valid under two-sided limited overlap and satisfy the BDR condition but not the GDR condition in general. The reason is that they multiply the main parameter of interest $\tau(\cdot)$ by a function of $e(\cdot)$. However, there are two alternatives that do satisfy the GDR principle and the ASPW principle in addition to satisfying Neyman orthogonality.

\begin{theorem}[Stablizing AIPW Produces a Neyman Orthogonal SPW-Based Residual with GDR]
    \label{theorem:aspw_gdr}
    If Assumptions \ref{assumption:unconfoundedness}--\ref{assumption:bounded_covariates}, then, the stabilized AIPW residual $\Psi(Y, X, W; \tau, e, (\mu_0, \mu_1))$ equal to $$r(X)(1-r(X))\big[\tau(X) - ([\mu_1 - \mu_0](X)) - (W - e(X))(Y - W\mu_1(X) - (1-W)\mu_0(X))/[e(X)(1-e(X))]\big],$$ where $r: \mathbb{X} \to (0, 1)$ is some known function such that $r(X)(1-r(X))/[e(X)(1-e(X))] \leq M$ a.s.~for some $M > 0$, satisfies the GDR condition, the ASPW principle, and Neyman orthogonality.
\end{theorem}

\begin{proof}
    Let $E[\tilde{\tau}, \tilde{e}, \tilde{\eta}](X) \equiv \E[\Psi(Y, X, W; \tilde{\tau}, \tilde{e}, \tilde{\eta}) \mid X]$, where $\tilde{\eta} \equiv (\tilde{\mu}_0, \tilde{\mu}_1)$, for any $\tilde{\tau} \in \mathcal{B}_\tau$, $\tilde{e} \in \mathcal{B}_e$, and $\tilde{\eta} \in \mathcal{B}_\eta$. Then, $E[\tilde{\tau}, \tilde{e}, \tilde{\eta}] \equiv r(1-r)[(\tilde{\tau} - \tilde{\mu}_1 + \tilde{\mu}_0) - e(\mu_1 - \tilde{\mu}_1)/\tilde{e} + (1-e)(\mu_0 - \tilde{\mu}_0)/(1 - \tilde{e})]$. Note that $E[\tilde{\tau}, e, \tilde{\eta}] \equiv E[\tilde{\tau}, \tilde{e}, \eta] \equiv r(1-r)(\tilde{\tau} - \tau)$. In addition, $E[\tilde{\tau}, e, \tilde{\eta}] \equiv E[\tilde{\tau}, \tilde{e}, \eta] \equiv 0 \iff \tilde{\tau} \equiv \tau$ since $r(1-r) \in (0, 0.25)$. Therefore, the stabilized AIPW residual satisfies the GDR property. To show Neyman orthogonality, let $E_t(x) = \E[\Psi(Y, X, W; \tau, (e, \mu_0, \mu_1)+t(h_e, h_{\mu_0}, h_{\mu_1})) \mid X = x]$, where $t \in [0, 1)$ and $(h_e, h_{\mu_0}, h_{\mu_1}) \in \mathcal{B}_{(e,\eta)}$, for all $x \in \mathbb{X}$. Note that the dependence of $E_t$ on $(e, \mu_0, \mu_1, h_e, h_{\mu_0}, h_{\mu_1})$ is suppressed for expositional ease. It is straightforward to show that $E_t \equiv r(1-r)[\tau - \mu_1 - th_{\mu_1} + \mu_0 + th_{\mu_0} - t(h_{\mu_1} - h_{\mu_0}) + eth_{\mu_1}/(e + th_e) - (1-e)th_{\mu_0}/(1-e-th_e)]$, implying that $\frac{d}{dt}E_t \equiv r(1-r)[h_{\mu_0} - h_{\mu_1} + e^2h_{\mu_1}/(e+th_e)^2 - (1-e)^2h_{\mu_0}/(1-e-th_e)^2]$, and so $\frac{d}{dt}E_t \big\lvert_{t = 0} \equiv 0$, proving that the stabilized AIPW residual $\Psi(Y, X, W; \tau, e, (\mu_0, \mu_1))$ is also Neyman orthogonal. In addition, it satisfies the SPW principle with $\mathcal{S}(X, W; e) = r(X)(1-r(X))$ and $\xi(X, W; e, \eta) = r(X)(1-r(X))\Big[\frac{W - e(X)}{[e(X)(1-e(X))}(W\mu_1(X) + (1-W)\mu_0(X)) - (\mu_1(X) - \mu_0(X))\Big]$.
\end{proof}

\begin{theorem}[GDR and Orthogonal ASPW Residual With Known Regions of Limited Overlap]
    \label{theorem:hybrid_gdr}
    Suppose the regions of limited overlap are known. Then, this knowledge can be used to construct a binary function $r: \mathbb{X} \to \{0, 1\}$ such that $r(x) = 1$ if $e(x) < \underline{p}^*$, $r(x) = 0$ if $e(x) > 1 - \underline{p}^*$, and $r(x) \in \{0, 1\}$ otherwise for all $x \in \mathbb{X}$, where $\underline{p}^* \in (0, 0.5)$ is some constant. Then, the residual $$\Psi(Y, X, W; \tau, e, \eta) = \mathcal{S}(X, W; e)\,\tau(X) - \tilde{s}(X; e)\,(W - e(X))[Y - r(X)\,\mu_0(X) - (1-r(X))\,\mu_1(X)],$$ where $\mathcal{S}(X, W; e) = W\,r(X) + (1-W)(1-r(X))$, $\tilde{s}(X; e) = r(X)/(1-e(X)) + (1-r(X))/e(X)$, and $\eta \equiv (\mu_0, \mu_1)$, satisfies the GDR condition, the ASPW principle, and Neyman orthogonality under Assumptions \ref{assumption:unconfoundedness}--\ref{assumption:bounded_covariates}. Note that this result treats $r(\cdot)$ as known, not the exact values of $e(\cdot)$.
\end{theorem}

\begin{proof}
    The covariate space $\mathbb{X}$ can partitioned based on the binary function $r(\cdot)$. Then, Theorem \ref{theorem:gdr_one_sided} can be applied separately on the bifurcated covariate spaces to prove the above result.
\end{proof}

\subsection{Augmented Stable Probability Weighting for Multivalued Treatments}
\label{subsection:spw_multivalued}

In this section, I show that the SPW principle readily extends to the case where general conditional structural parameters (defined using conditional outcome means or quantiles or distribution functions) are of interest in settings with multivalued treatments. However, to keep the discussion concise, the version of (Augmented) Stable Probability Weighting that I formulate for this purpose is less broad than the version developed previously for CATE.

Recall that the Conditional Average Response (CAR) and the Conditional Response Distribution (CRD) functions are given by $\mu(w, x) = \mathbb{E}[Y^*_w \mid X = x]$ and $\varrho(u; w, x) = \mathbb{P}\{Y^*_w \leq u \mid X = x\}$, respectively, for all $(u, w, x) \in \R \times \mathbb{W} \times \mathbb{X}$. The Conditional Average Contrast (CAC) function is given by $\textstyle \theta[\widetilde{\mathbb{W}}, \kappa] = \sum_{w \,\in\, \widetilde{\mathbb{W}}} \,\,\kappa_w\,\mu_w$ when $\widetilde{\mathbb{W}}$ is a finite subset of $\mathbb{W}$ and $\kappa \equiv (\kappa_w)_{w\,\in\,\mathbb{W}}$ is a $|\widetilde{\mathbb{W}}|$-dimensional vector of real constants. In addition, let $q_v(w,x) = \mathrm{inf}\{u \in \R: \varrho(u; w, x) \geq v\}$ be the Conditional Quantile Response (CQR) function for all $v \in (0, 1)$ and $(w, x) \in \mathbb{W} \times \mathbb{X}$. In the below discussion, $\widetilde{K}(W - w)$ denotes either $\mathbb{I}\{W - w = 0\}$ or $\mathrm{lim}_{\,h\,\downarrow\,0} \,K_h(W-w)$ for all $w \in \mathbb{W}$, depending on whether $\mathbb{W}$ is a discrete set or a continuum, respectively.

\begin{theorem}[ASPW for CAC of Multivalued Treatments].
    \label{theorem:aspw_cac}
     Let $\widetilde{\mathcal{S}}(\cdot; \varphi)$ be a positive bounded real-valued function on $\mathbb{X}$ that may depend on a bounded nuisance function $\varphi: \mathbb{W} \times \mathbb{X} \to \R_{> 0}$ such that $\widetilde{\mathcal{S}}(X; \varphi) \leq M\prod_{w \,\in\, \widetilde{\mathbb{W}}} \,\varphi(w, X)$ a.s.~for some $M > 0$. For $\theta \equiv \theta[\widetilde{\mathbb{W}}, \kappa]$, let $$\Psi(Y, X, W; \theta, \varphi, \gamma) = \widetilde{\mathcal{S}}(X; \varphi)\Bigg\{\sum_{w \,\in\, \widetilde{\mathbb{W}}} \,\,\kappa_w\,\Bigg[\gamma(w,X) + \frac{\widetilde{K}(W - w)}{\varphi(w,X)}\,[Y - \gamma(w,X)]\Bigg] - \theta[\widetilde{\mathbb{W}}, \kappa](X)\Bigg\}$$ be the ASPW generalized residual for CAC such that $\gamma: \mathbb{W} \times \mathbb{X} \to \R$ is a bounded real-valued nuisance parameter. Then, $\E[\Psi(Y, X, W; \theta, \varphi, \gamma) \mid X] = 0$ and $\mathbb{V}[\Psi(Y, X, W; \theta, \varphi, \gamma) \mid X] < \infty$ a.s.~under Assumptions \ref{assumption:unconfoundedness}--\ref{assumption:bounded_covariates} if $\varphi \equiv p$ or $\gamma \equiv \mu$.
\end{theorem}

\begin{proof}
    The zero conditional expectation of the ASPW residual follows from the double robustness property described in Definition \ref{definition:aipw}. The bounded conditional variance follows from the fact that $\widetilde{\mathcal{S}}(X; \varphi)\,/\,\big(\prod_{w \, \in\, \widetilde{\mathbb{W}}} \,\varphi(w, X)\big) \leq M$ a.s.~(in addition to the other boundedness assumptions).
\end{proof}

\begin{example}
    Theorem \ref{theorem:aspw_cac} can be satisfied by setting $\widetilde{\mathcal{S}}(X; \varphi) = \prod_{w \in \widetilde{\mathbb{W}}} \,\varphi(w, X)$ so that $M = 1$.
\end{example}

\begin{remark}[ASPW for CAR or CRD of Multivalued Treatments]
    Note that setting $\varphi \equiv p$ and $\gamma \equiv 0$ in the above residual gives the non-augmented SPW residual for CAC. When CAR is of interest, $\kappa$ in the parameter $\theta[\widetilde{\mathbb{W}}, \kappa]$ should be set to the appropriate standard unit vector. When the CRD (or a distributional treatment effect parameter) is of interest, the corresponding parameter can be expressed as $\sum_{w \,\in\, \widetilde{\mathbb{W}}} \,\,\kappa_w\,\varrho(u; w, \cdot)$ for any $u \in \R$, and so $Y$ and $\mu$ in the above theorem should be replaced by $\mathbb{I}\{Y \leq u\}$ and $\varrho(u; \cdot, \cdot)$, respectively.
\end{remark}

\begin{theorem}[ASPW for CQR of Multivalued Treatments] Suppose Assumptions \ref{assumption:unconfoundedness}--\ref{assumption:bounded_covariates} hold. Let $\Psi_{v, w}(Y, X, W; q, \varphi, \gamma)$, for some chosen values $(v, w) \in (0, 1) \times \mathbb{W}$, be the generalized residual $$\tilde{K}(W - w)[\mathbb{I}\{Y \leq q_v(w, X)\} - \gamma(q_v(w, X); w, X)] + \varphi(w, X)[\gamma(q_v(w, X); w, X) - v]$$ such that $\gamma: \R \times \mathbb{W} \times \mathbb{X} \to [0, 1]$ and $\varphi: \mathbb{W} \times \mathbb{X} \to \R_{> 0}$ are bounded nuisance functions. Then, $\E[\Psi_{v, w}(Y, X, W; q, p, \varrho) \mid X] = 0$ and $\mathbb{V}[\Psi_{v, w}(Y, X, W; q, p, \varrho) \mid X] < \infty$ a.s.~if $\varphi \equiv p$ or $\gamma \equiv \varrho$.  
\end{theorem}

\begin{proof}
    Note that $\E[\Psi_{v, w}(Y, X, W; q, \varphi, \gamma) \mid X] = p(w, X)[\varrho(q_v(w,X); w, X) - \gamma(q_v(w, X); w, X)] + \varphi(w, X)[\gamma(q_v(w, X); w, X) - \varrho(q_v(w,X); w, X)]$. Therefore, this equals zero almost surely if $\varphi \equiv p$ or $\gamma \equiv \varrho$. In addition, the boundedness assumptions imply bounded conditional variance.
\end{proof}

\begin{remark}[ASPW for Conditional Quantile Contrasts]
    Note that setting $\varphi \equiv p$ and $\gamma \equiv 0$ in the above residual gives the non-augmented SPW residual for CQR. If the Conditional Quantile Contrast (CQC) defined as $\sum_{w \,\in\, \widetilde{\mathbb{W}}} \,\,\kappa_w\,q_v(w, \cdot)$ is of interest for a given $v \in (0, 1)$, then it can be identified using the vector-valued residual $\Psi_v(Y, X, W; q, \varphi, \gamma) \equiv (\Psi_{v, w}(Y, X, W; q, \varphi, \gamma))_{w \in \widetilde{\mathbb{W}}}$.
\end{remark}

\subsection{Large-Sample Statistical Methods Using Stable Probability Weighting}
\label{subsection:large_sample_spw}

In this subsection, I discuss how the (A)SPW principle can be empirically implemented using several existing methods for conditional moment models. Thanks to this rich literature, this paper does not attempt to reinvent the wheel for large-sample estimation and inference.\footnote{Of course, future research can try to improve methods for efficient identification/singularity-robust semiparametric or nonparametric estimation and inference using infinite conditional moment equalities based on exogeneity.}

Let $\rho(Y, X, W; g)$, where $g$ includes the main parameters of interest and any nuisance parameters, be a (finite-dimensional) vector containing the (A)SPW-based generalized residuals for the parameters of interest as well as generalized residuals that identify the nuisance functions such that $\E[\rho(Y, X, W; g) \mid X] = 0$ almost surely. For example, when $\mathbb{W} = \{0, 1\}$ and CATE is of interest, the required conditional moment restriction is satisfied by setting $g = (\tau, e, \mu_0, \mu_1)$ and $\rho(Y, X, W; g) = (W - e(X),\,(1-W)(Y - \mu_0(X)),\,W(Y - \mu_1(X)),\,\Psi(Y, X, W; \tau, e, (\mu_0, \mu_1)))$, where $\Psi(Y, X, W; \tau, e, (\mu_0, \mu_1))$ is an (A)SPW generalized residual for CATE.\footnote{Sometimes a few nuisance functions, such as the $r(\cdot)$ functions in Theorems \ref{theorem:aspw_gdr} and \ref{theorem:hybrid_gdr} need to be excluded from $g$. When $r(\cdot)$ is unknown in such cases, data-driven assumptions on them can be made using auxiliary sample that is not used for the main analysis. However, the main analysis will then be conditional on the auxiliary sample.} BATE and PATE satisfy $\E[\tau(X) - Z'\beta \mid Z] = 0$ (a.s.) and $\E[\tau(X) - \tau^*] = 0$. \cite{chen2014local} provide a sufficient condition for (local) identification of $g$.\footnote{For example, invertibility of $\E[(e_i(1-e_i))^{\nu + 1}Z_iZ'_i]$, i.e., no perfect collinearity in $\sqrt{(e_i(1-e_i))^{\nu + 1}}Z_i$, where $\nu \geq 0$, is crucial for NPW-based identification of $\beta$ in Section \ref{section:alternatives} when CATE has linearity and $e(\cdot)$ is known.} Typically, further restrictions on $g$ (or further assumptions on the setting) are needed to satisfy the sufficient identification condition of \cite{chen2014local}. I cannot list them out with high specificity because they depend on the empirical context.

Nonparametric estimation of $g$ is feasible when $\E[\rho(Y, X, W; g) \mid X] = 0$ (a.s.) and other regularity conditions hold. One such important condition is that $\E[\lVert\rho(Y, X, W; g)\rVert^2 \mid X]$ is bounded, as required by \cite{ai2003efficient} and \cite{newey2003instrumental} for smooth generalized residuals $\rho(\cdot)$. For example, when $g = (\tau, e, \mu_0, \mu_1)$, boundedness of $\E[\lVert\rho(Y, X, W; g)\rVert^2 \mid X]$ under Assumptions \ref{assumption:unconfoundedness}--\ref{assumption:bounded_covariates} crucially depends on the boundedness of $\E[\Psi(Y, X, W; \tau, e, (\mu_0, \mu_1))^2 \mid X]$, which holds when the (A)SPW principle is used. The sieve minimum distance (SMD) estimation procedure of \cite{ai2003efficient} or other alternatives can be implemented for consistent estimation of $\tau$ and to obtain $\sqrt{n}$-consistent and asymptotically normal estimators of $(\tau^*, \beta)$. When the main parameters of interest are finite-dimensional, two-step semiparametric procedures may be used.\footnote{See, e.g., \cite{bravo2020two,chernozhukov2022locally,cattaneo2010efficient,chenliao2015sieve,ai2012semiparametric,ackerberg2012practical,ackerberg2014asymptotic,otsu2007penalized,hahn2018nonparametric,andrews1994asymptotics,newey1994asymptotic,chernozhukov2018double}.} When all the parameters are finite-dimensional, there exist simpler methods.\footnote{See, e.g., \cite{dominguez2004consistent,kitamura2004empirical,antoine2014conditional,andrews2013inference,jun2012testing,jun2009semiparametric,andrews2017inference}. If unconditional moment equalities identify the parameters, one can also use the methods of \cite{andrews2019identification,andrews2012estimation,andrews2016conditional}.}

In the general case where $\rho(\cdot)$ is possibly nonsmooth, penalized SMD estimators \citep{chen2012estimation} can be used, and sieve-based tests \citep{chen2015sieve,hong2017inference} help with inference on functionals of $g$. \cite{chen2016methods} provide a useful survey of other available methods. \cite{chernozhukov2022constrained} propose valid procedures when constraints, such as shape restrictions, are imposed on $g$. When $\tau$ (or the main structural parameter) is weakly identified but a strongly identified functional of $\tau$, such as $(\tau^*, \beta)$, is of interest, the penalized minimax (adversarial) approach of \cite{bennett2022inference} can be used for inference. Some recent high-dimensional methods \citep{dong2021high,chang2018new,chang2021high,nekipelov022regularized} may also be used.\footnote{High-dimensional statistical methods can be used by turning the conditional moment restrictions into unconditional moment equalities using high-dimensional instruments and parameters. For example, when PATE, BATE, and CATE are of interest so that $g = (\tau^*, \beta, \tau, e, \mu_0, \mu_1)$, high-dimensional logit and linear regression models may be specified for $e$ and $(\mu_0, \mu_1, \tau)$, respectively, and the corresponding high-dimensional regressors may be used as instruments. Then, under some restrictions on the high-dimensionality of the model (e.g., parameter sparsity), estimation and inference are possible using (penalized) sieve generalized method of moments \citep{dong2021high} or using high-dimensional penalized empirical likelihood methods \citep{chang2018new,chang2021high}. Regularized machine learning methods \citep{nekipelov022regularized} may also be used based on (A)SPW generalized residuals.} ASPW-based nonparametric estimators of CATE may be used for policy learning using, e.g., \citeauthor{kitagawa2018should}'s (\citeyear{kitagawa2018should}) hybrid Empirical Welfare Maximization (EWM) approach.

When $g$ only has infinite-dimensional components, e.g., $g = (\tau, e, \mu_0, \mu_1)$, and the main parameter evaluated at some $x \in \mathbb{X}$, e.g., $\tau(x)$, is of interest, then local estimation and inference methods are available. Kernel and local polynomial methods can be used for estimation.\footnote{See, e.g., \cite{lewbel2007local,carroll1998local,han2020identification,zhang2003local}.} In this context of localized moment restrictions, \cite{andrews2014nonparametric} and \cite{xu2020inference} provide useful inference methods that are robust to identification strength. Machine learning-based methods, such as orthogonal random forests \citep{oprescu2019orthogonal}, may also be used based on ASPW residuals.

\clearpage

\section{Finite-Sample Stable Probability Weighting and Inference}
\label{section:fpw}

In this section, I develop finite-sample estimation and inference methods for heterogeneous average treatment effects of finite multivalued treatments $\mathbb{W}$. The setup for finite-sample results uses a ``fixed design'' that treats the covariates as given.\footnote{This type of conditioning on covariates is quite common in finite-sample statistical theory. For the results on design-based inference, I also condition on the (realized or latent) potential outcomes, following common practice.} I assume that the dataset $\mathcal{D}_n \equiv (Y_i, X_i, W_i)_{i = 1}^n$ has independent observations and follows the model specified in Assumptions \ref{assumption:high_dim_strata}--\ref{assumption:stratified_car} below.

\begin{assumption}[High-Dimensional Strata]
    \label{assumption:high_dim_strata}
   The covariate set $\mathbb{X}_n = \bigcup_{i = 1}^n \{X_i\}$, which is assumed to be $\{1, \cdots, K_n\}$ without loss of generality, is such that $\sum_{i = 1}^n \mathbb{I}\{X_i = k\} \geq 2$ for all $k \in \mathbb{X}_n$.
\end{assumption}

\begin{assumption}[Unconfoundedness, SUTVA, and Overlap]
    \label{assumption:fs_unconfoundedness}
    For all $i \in \{1, \dots, n\}$, let $Y^*_{w,i}$ be the potential outcome under treatment $w \in \mathbb{W}$. Then, for all $i \in \{1, \dots, n\}$, $Y_i = Y^*_{W_i,i}$ almost surely, and $\mathbb{P}\{W_i = w \mid X_i, (Y^*_{w,i})_{w \in \mathbb{W}}\} = \mathbb{P}\{W_i = w \mid X_i\} \equiv \lambda_{w,X_i} \in (0, 1)$ such that $\sum_{w\in\mathbb{W}} \lambda_{w,X_i} = 1$.
\end{assumption}

\begin{assumption}[Stratified and Bounded Conditional Average Responses]
    \label{assumption:stratified_car}
    For all $w \in \mathbb{W}$ and $i \in \{1, \dots, n\}$, $\E[Y^*_{w,i} \mid X_i] = \mu_w(\tilde{Z}(X_i))$, where $\tilde{Z}$ is a surjective function that maps $\mathbb{X}_n$ to a finite set $\mathbb{Z}_n$, and $\mu_w: \mathbb{Z}_n \to [\underline{C}_w, \overline{C}_w]$ is a function that maps $\mathbb{Z}_n$ to a known interval $[\underline{C}_w, \overline{C}_w] \subset \R$.
\end{assumption}

The above assumptions essentially specify a (potentially) high-dimensional model (conditional on covariates) with fine strata $\mathbb{X}_n$ (and $\mathbb{Z}_n$) that could depend on the sample size.\footnote{Note that the strata (i.e., the covariate set $\mathbb{X}_n$) in Assumption \ref{assumption:high_dim_strata} is different from the set $\mathbb{X}$ used for large-sample theory. Since the finite-sample analysis is conditional on covariates, the researcher may use unsupervised feature learning to form the strata $\mathbb{X}_n$ based on a richer set of background variables. For example, if a scalar background variable is available, it can be partitioned (say, based on quantiles) into disjoint bins, which then are the strata.} Henceforth, I assume without loss of generality that $|\mathbb{Z}_n | = 1$ so that $\E[Y^*_{w,i} \mid X_i] = \mu_w$ for some real number $\mu_w \in [\underline{C}_w, \overline{C}_w]$ for all $w\in \mathbb{W}$ and $i \in \{1, \dots, n\}$.\footnote{Assumption \ref{assumption:fs_unconfoundedness} uses a partitioning-based high-dimensional model (similar to a a zero-degree piecewise polynomial or spline) for the propensity score with potentially $(|\mathbb{W}| - 1) \times K_n$ unknown nuisance parameters. Assumption \ref{assumption:stratified_car} uses (potentially) coarser strata for the Conditional Average Response (CAR) function but still could be high-dimensional. If we are interested in estimating CAR at a particular value $z \in \mathbb{Z}_n$, then only the units $i$ with $X_i \in \tilde{Z}^{-1}(z)$ need to be used, and those units would then comprise the dataset for finite-sample analysis. Since the dataset $\mathcal{D}_n$ could be redefined to contain only those units, we may assume without loss of generality that $\tilde{Z} \equiv 1$ so that $|\mathbb{Z}_n | = 1$.} I am interested in ``unbiased'' estimation of the Conditional Average Contrast (CAC) parameter $\theta = \sum_{w \in \mathbb{W}} \kappa_w\,\mu_w$ for a given $(\kappa_w)_{w \in \mathbb{W}} \in \R^{|\mathbb{W}|}$.

When $K_n = 1$, an unbiased estimator of $\mu_w$ ($w \in \mathbb{W}$) is $\frac{1}{n}\sum_{i = 1}^n \mathbb{I}\{W_i = w\}\,Y_i/\lambda_{w,1}$, but this is not computable because $\lambda_{w,1}$ is unknown. A conventionally used alternative estimator of $\mu_w$ is the subsample mean\footnote{I do not call the subsample mean the IPW estimator for reasons that will become clear later.} $\big(\sum_{i = 1}^n \mathbb{I}\{W_i = w\}\,Y_i\big)/\big(\sum_{i = 1}^n \mathbb{I}\{W_i = w\}\big) = \frac{1}{n}\sum_{i = 1}^n \mathbb{I}\{W_i = w\}\,Y_i/\hat{\lambda}_{w,1}$, where $\hat{\lambda}_{w,1} = \frac{1}{n}\sum_{i = 1}^n \mathbb{I}\{W_i = w\}$. However, the expectation of this subsample mean does not exist because the denominator $\hat{\lambda}_{w,1}$ can be zero with positive probability, especially under limited overlap. Since unbiased point-estimation may not be generally feasible in this setting, I focus on set-estimators, for which I define a generalized notion of ``unbiasedness'' as follows.

\begin{definition}[Unbiased Set-Estimator]
    A set-estimator $[\widehat{\theta}_l, \widehat{\theta}_u] \subset \R$ of a parameter $\theta \in \R$, where $\widehat{\theta}_l \leq \widehat{\theta}_u$ (or $\widehat{\theta}_l = \widehat{\theta}_u$ for point estimators) a.s., is said to be unbiased for $\theta$ if $\E[\widehat{\theta}_l] \leq \theta \leq \E[\widehat{\theta}_u]$.
\end{definition}

To inform unbiased set-estimation, it is useful to compute the finite-sample bias of a (biased) estimator whose expectation does exist (unlike the subsample mean mentioned before). For this purpose, some sensible estimator (with a finite expectation) of the reciprocal of the propensity score is needed. However, in general, it is impossible to construct an unbiased estimator of the reciprocal of a Bernoulli parameter \citep{voinov1993unbiased}. Nevertheless, a (biased) shrinkage estimator \citep{fattorini2006applying} is sometimes used for practical purposes. I use a leave-one-out version of it below.

\begin{theorem}[Bias From Using a Leave-One-Out Shrinkage Estimator of Inverse Probability]
    \label{theorem:biased_estimator}
    Let $\mathcal{X}_k = \{i: X_i = k\}$ and $N_k = |\mathcal{X}_k|$ for all $k \in \mathbb{X}_n$. Let $\widehat{R}_{w,i} = N_{X_i}/(1 + \sum_{j \neq i\,|\, X_j = X_i} \mathbb{I}\{W_j = w\})$ for all $w\in \mathbb{W}$ and $i \in \{1,\dots,n\}$. Then, under Assumptions \ref{assumption:high_dim_strata}--\ref{assumption:stratified_car}, for all $w \in \mathbb{W}$ and $k \in \mathbb{X}_n$, the bias of the estimator $\widetilde{\mu}_{w,k} = \frac{1}{N_k}\sum_{i \in \mathcal{X}_k} \widehat{R}_{w,i}\,\mathbb{I}\{W_i = w\}\,Y_i$ is $\E[\widetilde{\mu}_{w,k}] - \mu_w = -(1 - \lambda_{w,k})^{N_k}\,\mu_w$.
\end{theorem}

\begin{proof}
    Note that $\E[\widehat{R}_{w,i}\,\mathbb{I}\{W_i = w\}\,Y_i] = \E[\E[\widehat{R}_{w,i}\,\mathbb{I}\{W_i = w\}\,Y_i \mid X_i]] = \E[\widehat{R}_{w,i} \mid X_i]\,\lambda_{w,k}\,\mu_w$ under the assumptions (and since $\widehat{R}_{w,i} \indep (Y_i, W_i) \mid X_i$) for all $i \in \mathcal{X}_k$ and any $(w, k) \in \mathbb{W} \times \mathbb{X}_n$. Since $\E[\widehat{R}_{w,i} \mid X_i] = \big[1-(1-\lambda_{w,k})^{N_k}\big]/\lambda_{w,k}$ by the calculations of \cite{chao1972negative}, $\E[\widehat{R}_{w,i}\,\mathbb{I}\{W_i = w\}\,Y_i] = \mu_w - (1-\lambda_{w,k})^{N_k}\mu_w$, and so the result follows.
\end{proof}

An interesting feature of the estimator $\widetilde{\mu}_{w,k}$ is that it is numerically equivalent to a modified subsample mean, i.e., $\widetilde{\mu}_{w,k} \equiv \big(\sum_{i \in \mathcal{X}_k} \mathbb{I}\{W_i = w\}\,Y_i\big)/\mathrm{max}\{1, \sum_{i \in \mathcal{X}_k} \mathbb{I}\{W_i = w\}\}$, which does have a finite expectation, for all $(w, k) \in \mathbb{W} \times \mathbb{X}_n$. If it is possible to construct an unbiased set-estimator of the bias of $\widetilde{\mu}_{w,k}$, then an unbiased set-estimator of $\mu_w$ can also be constructed for any $w \in \mathbb{W}$. I provide such a construction below but without pooling information across strata.

\begin{theorem}[Unpooled Unbiased Set-Estimators of Conditional Average Responses]
    \label{theorem:unpooled_set_est}
    For all $(w, k) \in \mathbb{W} \times \mathbb{X}_n$, let $\widehat{\mu}_{w,k}: \R \to \R$ be a function given by $\widehat{\mu}_{w,k}(t) = \widetilde{\mu}_{w,k} + t \big(\prod_{i \in \mathcal{X}_k} \mathbb{I}\{W_i \neq w\}\big)$. Then, an unbiased set-estimator of $\mu_w$ is given by $[\widehat{\mu}_{w,k}(\underline{C}_w),\,\widehat{\mu}_{w,k}(\overline{C}_w)]$, which reduces to the point-estimator $\widetilde{\mu}_{w,k}$ if $W_i = w$ for at least one $i \in \mathcal{X}_k$ but otherwise reduces to $[\underline{C}_w, \overline{C}_w]$, for all $(w, k) \in \mathbb{W} \times \mathbb{X}_n$ under Assumptions \ref{assumption:high_dim_strata}--\ref{assumption:stratified_car}.
\end{theorem}

\begin{proof}
    Under the assumptions, $\widehat{\mu}_{w,k}(\mu_w) \in [\widehat{\mu}_{w,k}(\underline{C}_w),\,\widehat{\mu}_{w,k}(\overline{C}_w)]$ almost surely, and so it follows that $\E[\widehat{\mu}_{w,k}(\underline{C}_w)] \leq \E[\widehat{\mu}_{w,k}(\mu_w)] \leq \E[\widehat{\mu}_{w,k}(\overline{C}_w)]$. Note that $\E[\prod_{i \in \mathcal{X}_k} \mathbb{I}\{W_i \neq w\}] = (1 - \lambda_{w,k})^{N_k}$, so $\E[\widehat{\mu}_{w,k}(\mu_w)] = \E[\widetilde{\mu}_{w,k}(\mu_w)] + \mu_w\,\E[\prod_{i \in \mathcal{X}_k} \mathbb{I}\{W_i \neq w\}] = \mu_w + (-1 + 1)(1-\lambda_{w,k})^{N_k}\mu_w = \mu_w$. Therefore, $[\widehat{\mu}_{w,k}(\underline{C}_w),\,\widehat{\mu}_{w,k}(\overline{C}_w)]$ is an unbiased set-estimator of $\mu_w$ for all $(w, k) \in \mathbb{W} \times \mathbb{X}_n$.
\end{proof}

Although the set-estimators provided above are unbiased, it may be more efficient to pool information across strata, especially when there is strong overlap in some strata. I do this using the ``Finite-Sample Stable Probability Weighting'' (FPW) set-estimator below and show that it is unbiased. It can be seen as a counterpart of ASPW estimators in the finite-sample setting.

\begin{theorem}[FPW: Finite-Sample Stable Probability Weighting Set-Estimator]
    \label{theorem:fpw}
    For all $w \in \mathbb{W}$, let $\widehat{\mathcal{U}}_w = [\bar{\mu}_w(\underline{C}_w), \bar{\mu}_w(\overline{C}_w)]$, where $\bar{\mu}_w: \R \to \R$ is given by $\bar{\mu}_w(t) = \frac{1}{n}\sum_{i = 1}^n \bar{\mu}_{w,i}(t)$ such that $$\bar{\mu}_{w,i}(t) \equiv \Bigg[ \widehat{R}_{w,i}\,\mathbb{I}\{W_i = w\}\,Y_i + \Bigg(\prod_{j \,|\, X_j = X_i} \mathbb{I}\{W_j \neq w\}\Bigg)\Bigg(\sum_{k \,\in\, \mathbb{X}_n\setminus \{X_i\}}\frac{N_k}{n - N_{X_i}}\,\widehat{\mu}_{w,k}(t) \Bigg)\Bigg]$$
    if $K_n > 1$ and such that $\bar{\mu}_{w,i}(t) \equiv \widehat{R}_{w,i}\,\mathbb{I}\{W_i = w\}\,Y_i + t\,(\prod_{\,j} \mathbb{I}\{W_j \neq w\})$ if $K_n = 1$. Let $\widehat{\theta}_l$ be the minimum and $\widehat{\theta}_u$ the minimum of the set $\big\{\sum_{w \in \mathbb{W}} \kappa_w\,v_w \,\big\lvert\, v_w \in \widehat{\mathcal{U}}_w \,\forall\,w\big\}$. Then, $[\widehat{\theta}_l, \widehat{\theta}_u]$ is an unbiased set-estimator of $\theta = \sum_{w \in \mathbb{W}} \kappa_w\,\mu_w$ under Assumptions \ref{assumption:high_dim_strata}--\ref{assumption:stratified_car}.
\end{theorem}

\begin{proof}
    For all $i \in \{1,\dots,n\}$ and $w \in \mathbb{W}$, $\E[\bar{\mu}_{w,i}(\mu_w)] = \mu_w$ by Theorems \ref{theorem:biased_estimator} and \ref{theorem:unpooled_set_est}. Thus, $\E[\bar{\mu}_w(\underline{C}_w)] \leq \mu_w \leq \E[\bar{\mu}_w(\overline{C}_w)]$, making $\widehat{\mathcal{U}}_w$ an unbiased set-estimator of $\mu_w$, for all $w \in \mathbb{W}$. Therefore, by construction, $[\widehat{\theta}_l, \widehat{\theta}_u]$ is an unbiased set-estimator of $\theta$ under the assumptions.
\end{proof}

\begin{remark}
    Note that the weights $\frac{N_k}{n - N_{X_i}}$ in the above formula may be replaced with any arbitrary set of nonnegative weights $\omega_{w,k,i} \geq 0$ such that $\sum_{k \,\in\, \mathbb{X}_n\setminus \{X_i\}} \omega_{w,k,i} = 1$ for all $i \in \{1, \dots, n\}$ and $w \in \mathbb{W}$. If it is known a priori that some strata $\mathcal{K} \subset \mathbb{X}_n$ have strong overlap, it may be practical to to set $\omega_{w,k,i}$ to highest values for $k \in \mathcal{K}$. However, the FPW set-estimator is more systematic.
\end{remark}

The contribution $\bar{\mu}_{w,i}(t)$ of each observation $i \in \{1, \dots, n\}$ to the FPW set-estimator $\widehat{\mathcal{U}}_w$ (when $\theta = \mu_w$) for some $w \in \mathbb{W}$ is very intuitive. When the stratum of $i$ has information about $\mu_w$, then $\bar{\mu}_{w,i}(t)$ uses a type of stabilized probability weighting $\widehat{R}_{w,i}\,\mathbb{I}\{W_i = w\}\,Y_i$ using the stable reciprocal weight $\widehat{R}_{w,i}$. Otherwise, $\bar{\mu}_{w,i}(t)$ uses a certain type of imputation based on information (about $\mu_w$) in other strata. Theorem \ref{theorem:fpw} shows that this sort of case-specific imputation is unbiased in a sense. The FPW approach can be interpreted as a finite-sample counterpart of the large-sample ASPW approach, which is conceptually similar.\footnote{My FPW approach is similar in spirit to that of \cite{lee2021bounding}, but they consider a broader setting where the parameters of interest may not even be point-identified, so the set-estimators they propose are generally wider. Under my assumptions, the parameters are indeed point-identified. Thus, in this setting, the FPW set-estimator, which actually reduces to a point-estimator for many practical purposes, is much simpler conceptually and computationally.} The FPW estimator can also be seen as a bias-corrected version of the Weighted Modified Difference (WMD) estimator $\big[\sum_{w\in \mathbb{W}} \kappa_w \big(\sum_{k \in \mathbb{X}_n} \frac{N_k}{n}\,\widetilde{\mu}_{w,k}\big)\big] \equiv \big[\sum_{w\in \mathbb{W}} \kappa_w \big(\sum_{k \in \mathbb{X}_n} \frac{N_k}{n}\,\big(\sum_{i \in \mathcal{X}_k} \mathbb{I}\{W_i = w\}\,Y_i\big)/\mathrm{max}\{1, \sum_{i \in \mathcal{X}_k} \mathbb{I}\{W_i = w\}\}\big)\big]$. Both the FPW and WMD estimators can be contrasted with the IPW estimator $\big[\sum_{w\in \mathbb{W}} \kappa_w \big(\sum_{k \in \mathbb{X}_n} \frac{N_k}{n}\,\breve{\mu}_{w,k}\big)\big]$, where $\breve{\mu}_{w,k} \equiv \frac{1}{N_k}\sum_{i \in \mathcal{X}_k} \mathbb{I}\{W_i = w\}\,Y_i/\mathrm{max}\{\widehat{P}_{w,i}, 1/(2N_k - 2)\}$ that uses the leave-one-out propensity score estimate $\widehat{P}_{w,i} \equiv (\sum_{j \neq i\,|\, X_j = X_i} \mathbb{I}\{W_j = w\})/(N_{X_i} - 1)$. It follows from Jensen's inequality\footnote{Note that $\E[1/\mathrm{max}\{\widehat{P}_{w,i}, [2(N_k - 1)]^{-1}\}] \geq 1/\E[\widehat{P}_{w,i}] = 1/\lambda_{w,X_i}$ if $\lambda_{w,X_i} \geq 1/[2(N_k - 1)]$, because the following function is convex for all $k \in \mathbb{X}_n$: $1/P$ if $P \geq 1/(N_k - 1)$, and $2(N_k-1) - (N_k - 1)^2P$ otherwise.} that $\breve{\mu}_{w,k}$ is a biased estimator of $\mu_w$ for all $(w, k) \in \mathbb{W} \times \mathbb{X}_n$. Figures \ref{figure:sampling_dist_n50_p2} through \ref{figure:sampling_dist_n500_p50} in the Appendix illustrate the unbiasedness property of the FPW estimator (compared with the IPW and WMD estimators that can have very large biases and variances) in a simple setting. 

The leave-one-out propensity score estimate $\widehat{P}_{w,i} \equiv (\sum_{j \neq i\,|\, X_j = X_i} \mathbb{I}\{W_j = w\})/(N_{X_i} - 1)$ is used by the FPW and IPW estimators (in different ways) for finite-sample estimation of the CAC parameter $\theta$. The IPW estimator is a point-estimator but it is biased. The FPW is unbiased but it is a set-estimator. A natural question is whether it is possible to obtain an unbiased point-estimator (potentially based on $\widehat{P}_{w,i}$) for any particular forms of CAC $\theta = \sum_{w \in \mathbb{W}} \kappa_w\,\mu_w$, where the constants $(\kappa_w)_{w \in \mathbb{W}} \in \R^{|\mathbb{W}|}$ may themselves potentially depend on unknown (nuisance) parameters. The below theorem addresses this question and also paves the way for practical finite-sample inference.

\begin{theorem}[Unbiased Point-Estimators of Scaled Average Treatment Effects]
    \label{theorem:unbiased_scaled_ate}
    Let $(a, b) \in \mathbb{W}^2$ such that $a \neq b$ and $\kappa_a = - \kappa_b = \frac{1}{n}\sum_{i = 1}^n \lambda_{a,X_i}\,\lambda_{b,X_i} \equiv \mathcal{G}_{a,b}$ while $\kappa_{\tilde{w}} = 0$ for $\tilde{w} \in \mathbb{W} \setminus \{a, b\}$. Then, $\widehat{T}_{a,b} \equiv \frac{1}{n}\sum_{i = 1}^n \big[\widehat{P}_{b,i}\,\mathbb{I}\{W_i = a\} - \widehat{P}_{a,i}\,\mathbb{I}\{W_i = b\}\big]Y_i$ is an unbiased point-estimator of the CAC parameter given by $\theta \equiv \sum_{w \in \mathbb{W}} \kappa_w\,\mu_w = \mathcal{G}_{a,b}(\mu_a - \mu_b)$ so that the moment equality $\E[\widehat{\Psi}_{a,b}] = 0$ holds with the residual $\widehat{\Psi}_{a,b} \equiv \mathcal{G}_{a,b}(\mu_a - \mu_b) - \widehat{T}_{a,b}$ under Assumptions \ref{assumption:high_dim_strata}--\ref{assumption:stratified_car}.
\end{theorem}

\begin{proof}
    Since $\E[\widehat{P}_{b,i}\,\mathbb{I}\{W_i = a\}Y_i - \widehat{P}_{a,i}\,\mathbb{I}\{W_i = b\}Y_i] = \lambda_{b,X_i}\lambda_{a,X_i}\mu_a - \lambda_{a,X_i}\lambda_{b,X_i}\mu_b$ for all units $i \in \{1, \dots, n\}$ under the assumptions, it follows that $\E[\widehat{T}_{a,b}] = (\frac{1}{n}\sum_{i = 1}^n \lambda_{a,X_i}\,\lambda_{b,X_i})(\mu_a - \mu_b)$.
\end{proof}

Note that the moment equality in the above theorem strongly resembles the conditional moment equality for SPW generalized residuals; this is not surprising because they are based on the same underlying idea. Clearly, Theorem \ref{theorem:unbiased_scaled_ate} implies that $\widehat{T}_{a,b}/\mathcal{G}_{a,b}$ is an (infeasible) unbiased point-estimator of the average treatment effect $\mu_a - \mu_b$. Despite the infeasibility of $\widehat{T}_{a,b}/\mathcal{G}_{a,b}$ for finite-sample estimation, $\widehat{T}_{a,b}/\mathcal{G}_{a,b}$ is highly useful for finite-sample inference. Suppose $T^{**}_{a,b}$ is a random draw from the finite-sample distribution corresponding to $\widehat{T}_{a,b}$. Then, despite the fact that $\widehat{T}_{a,b}/\mathcal{G}_{a,b}$ is latent and not computable, note that $\mathbb{P}\{T^{**}_{a,b}/\mathcal{G}_{a,b} \geq \widehat{T}_{a,b}/\mathcal{G}_{a,b}\} = \mathbb{P}\{T^{**}_{a,b} \geq \widehat{T}_{a,b}\}$. Thus, the $p$-values associated with the test statistics $\widehat{T}_{a,b}$ and $\widehat{T}_{a,b}/\mathcal{G}_{a,b}$ are indeed equivalent. 

I now propose a conceptually and computationally simple finite-sample inference method based on the above statistic to test null hypotheses regarding the average effects of multivalued treatments. Specifically, I am interested in ``weak'' null hypotheses that concern the average treatment effect vector $T^\circ = (T^\circ_w)_{w = 1}^L$, where $\mathbb{W} = \{0,1,\dots,L\}$ without loss of generality and $T^\circ_w \equiv \mu_w - \mu_0$, so that $T^\circ \equiv (\mu_1 - \mu_0, \dots, \mu_L - \mu_0)$. In contrast, ``sharp'' null hypotheses hypothesize a specific set of restricted values of all the missing (counterfactual) potential outcomes. Currently, there only exist asymptotically valid tests\footnote{See \cite{chung2013exact,chung2016multivariate,wu2021randomization}.} of weak null hypotheses in my general setting.\footnote{There exist some exceptions, such as the case of binary outcomes, for which \cite{rigdon2015randomization} and \cite{li2016exact} propose some methods. \cite{caughey2021randomization} provide useful methods for testing ``bounded'' null hypotheses with one-sided treatment effect heterogeneity, but the weak null hypotheses I consider are more general.} I use partial identification and a simple simulation-based method to bound the entire ``$p$-value function'' \citep{luo2021leveraging},  which may be called a finite-sample version of the ``confidence distribution'' \citep{xie2013confidence}, using only a single Monte Carlo sample.

As is standard in the literature on design-based finite-sample inference, I treat the following as fixed for inferential purposes: the (partially observed) data $\mathcal{D}^*_n \equiv ((Y^*_{w,i})_{w \in \mathbb{W}}, X_i)_{i = 1}^n$, and a set $\Lambda$ that partially identifies the nuisance parameters $\{\lambda_{w,k}: (w, k) \in \mathbb{W} \times \mathbb{X}_n\}$ by assuming bounds for them, implying the assumption that $\{\lambda_{w,k}: (w, k) \in \mathbb{W} \times \mathbb{X}_n\} \in \Lambda$. Thus, in the following results, the notation $\mathbb{P}_l\{\cdot\}$ is used to represent probability when the only source of randomness is the vector of treatment statuses that are generated according to the model $l \in \Lambda$.

For purposes of inference on the parameter $\mu_a - \mu_b$, I consider ``linear'' test statistics that can be expressed as $\breve{T}_{a,b} = \frac{1}{n}\sum_{i = 1}^n Q_i^{a,b}(\mathcal{W}_n)\,Y_i$, where the weight $Q_i^{a,b}(\mathcal{W}_n)$ may depend on the entire set $\mathcal{W}_n \equiv (X_i, W_i)_{i = 1}^n$ for all $i \in \{1, \dots, n\}$, for a distinct pair $(a, b) \in \mathbb{W}^2$. I prefer to choose $\breve{T}_{a,b} \equiv \widehat{T}_{a,b}$, but $\breve{T}_{a,b}$ may also be chosen to be the appropriate IPW estimator or WMD estimator. For testing weak null hypotheses, I use a flexible restriction on unit-level effect heterogeneity.

\begin{assumption}[Bounds on Unit-Level Effect Heterogeneity]
    \label{assumption:bounds_het}
    Let $\varepsilon^*_{w,i} \equiv (Y^*_{w,i} - Y^*_{0,i}) - (\mu_w - \mu_0)$ for all $w \in \mathbb{W} \setminus \{0\} \equiv \{1, \dots, L\}$ such that $\varepsilon^*_{w,i} \in [-c_w, c_w]$ for some $c_w \geq 0$ for all $i \in \{1, \dots, n\}$.
\end{assumption}

To avoid cumbersome notation, I henceforth concentrate on the case where $\mathbb{W} = \{0, 1\}$ in order to construct bounds for the $p$-value function in Theorem \ref{theorem:bounds_cd} below. However, a conceptually (and mechanically) similar construction is possible for the general case where $\mathbb{W} = \{0, \dots, L\}$.

\begin{theorem}[Partial Identification of the Finite-Sample Confidence Distribution]
    \label{theorem:bounds_cd}
    Let $\mathbb{W} = \{0, 1\}$ so that $\mu_1 - \mu_0 \equiv T^\circ \in \R$. Consider testing the ``weak'' null hypothesis $\mathcal{H}_0: T^\circ = \overline{T}\in \R$. Let $\widetilde{\mathcal{W}}^{l}_n \equiv (X_i, \widetilde{W}^{l}_i)_{i = 1}^n$ such that $(\widetilde{W}^{l}_i)_{i = 1}^n$ is a random vector that follows the model $l \in \Lambda$. In addition, let $\widetilde{\Omega}^l_1 \equiv \frac{1}{n}\sum_{i = 1}^n Q_i^{1,0}(\widetilde{\mathcal{W}}^l_n)\,Y_i$ and $\widetilde{\Omega}^l_2 \equiv \sum_{i = 1}^n \widetilde{U}^l_i$, where $\widetilde{U}^l_i \equiv \frac{1}{n}\, Q_i^{1,0}(\widetilde{\mathcal{W}}^l_n)\,(\widetilde{W}^l_i - W_i)$. Furthermore, let $\widetilde{\Omega}^l_3 \equiv \sum_{i = 1}^n \mathbb{I}\{\widetilde{U}^l_i \geq 0\}\, \widetilde{U}^l_i$ and $\widetilde{\Omega}^l_4 \equiv \sum_{i = 1}^n \mathbb{I}\{\widetilde{U}^l_i < 0\}\, \widetilde{U}^l_i$. Then, the $p$-value (as a function of $\overline{T}$) under Assumptions \ref{assumption:high_dim_strata}--\ref{assumption:bounds_het} is bounded by the extrema of the following set: $$\Big\{\mathbb{P}_l\{\widetilde{\Omega}^l_1 + \widetilde{\Omega}^l_2\,\overline{T} + \epsilon_3\, \widetilde{\Omega}^l_3 + \epsilon_4 \,\widetilde{\Omega}^l_4 \geq \widehat{T}_{1,0} \}: (\epsilon_3, \epsilon_4) \in \{-c_1, c_1\}^2,\, l \in \Lambda\Big\}.$$
\end{theorem}

\begin{proof}
    Under the assumptions and the null hypothesis $\mathcal{H}_0: T^\circ = \overline{T}$, note that $Y^*_{1,i} = Y^*_{0,i} + \overline{T} + \varepsilon^*_{1,i}$. Then, $Y_i = Y^*_{0,i} + W_i(Y^*_{1,i} - Y^*_{0,i}) = Y^*_{0,i} + W_i(\overline{T} + \varepsilon^*_{1,i})$ so that $Y^*_{w,i} = Y_i + (\overline{T} + \varepsilon^*_{1,i})(w - W_i)$ for $w \in \{0, 1\}$. Under $\widetilde{\mathcal{W}}^{l}_n$, the observed outcomes are $\widetilde{Y}^l_i = Y_i + \overline{T}(\widetilde{W}_i - W_i) + \varepsilon^*_{1,i}(\widetilde{W}_i - W_i)$ for all $i \in \{1, \dots, n\}$. Under model $l \in \Lambda$, the $p$-value is given by $\mathbb{P}_l\{\widetilde{T}^l \geq \widehat{T}_{1,0}\}$, where the random variable is $\widetilde{T}^l = \frac{1}{n}\sum_{i = 1}^n Q_i^{1,0}(\widetilde{\mathcal{W}}^l_n)\,\widetilde{Y}^l_i = \widetilde{\Omega}^l_1 + \widetilde{\Omega}^l_2\,\overline{T} + \sum_{i = 1}^n \varepsilon^*_{1,i}\,\widetilde{U}^l_i$ such that $\varepsilon^*_{1,i}\,\widetilde{U}^l_i$ lies between $-c_1\widetilde{U}^l_i$ and $c_1\widetilde{U}^l_i$ for all $i \in \{1, \dots, n\}$, resulting in the above bounds for the $p$-value function.
\end{proof}

Stochastic approximation may be used to compute $\mathbb{P}_l\{\cdot\}$ for all $l \in \Lambda$ without affecting the finite-sample validity of the test \citep{lehmann2005testing}. Only a single set of Monte Carlo draws are needed to bound the entire $p$-value function based on the above result. Valid confidence sets can also be computed by simply ``inverting'' \citep{luo2021leveraging} all possible $p$-value functions.

\clearpage

\small

\setstretch{1}
\setlength{\bibsep}{2pt}
\bibliographystyle{chicago}
\bibliography{references}

\clearpage

\section*{Appendix of Illustrations}

\begin{figure}[ht]
\caption{Sampling Distributions of IPW-Based $(1,0)'\widehat{\beta}_{-1}$ and NPW-Based $(1,0)'\widehat{\beta}_{1}$ Estimators}
\label{figure:ipw_vs_npw_example_coef1est_dist}
\vspace{-5mm}
\begin{center}
\includegraphics[scale=0.7]{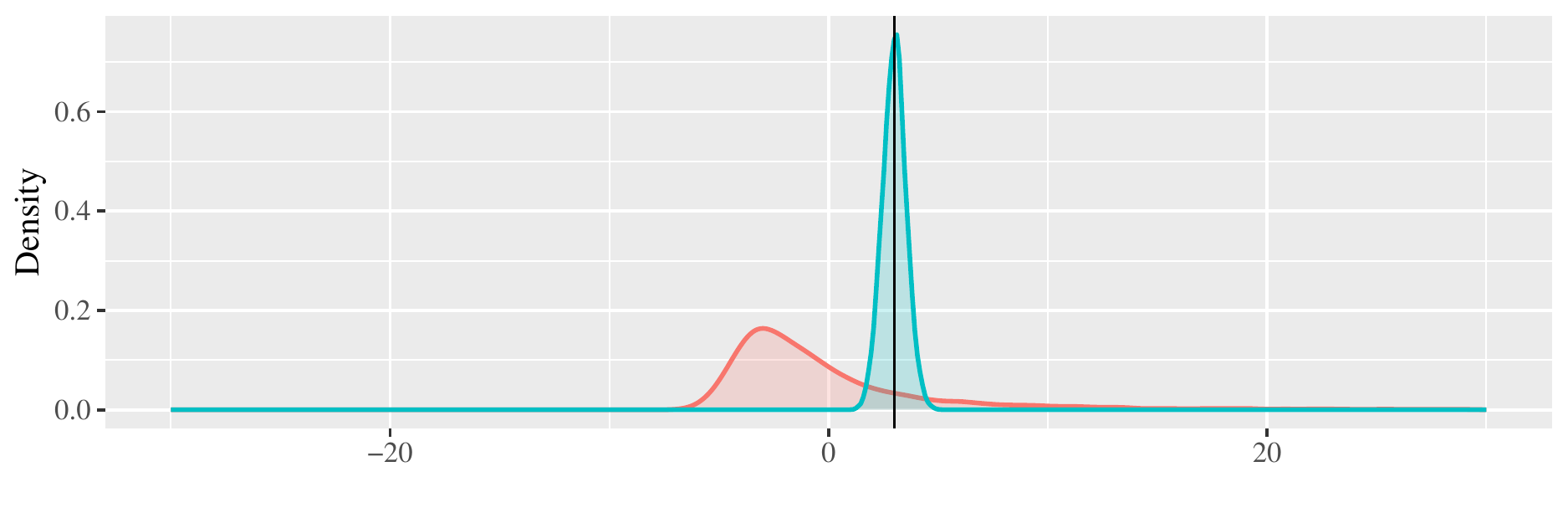}
\end{center}
\vspace{-5mm}
\caption{Sampling Distributions of IPW-Based $(0,1)'\widehat{\beta}_{-1}$ and NPW-Based $(0,1)'\widehat{\beta}_{1}$ Estimators}
\label{figure:ipw_vs_npw_example_coef2est_dist}
\vspace{-5mm}
\begin{center}
\includegraphics[scale=0.7]{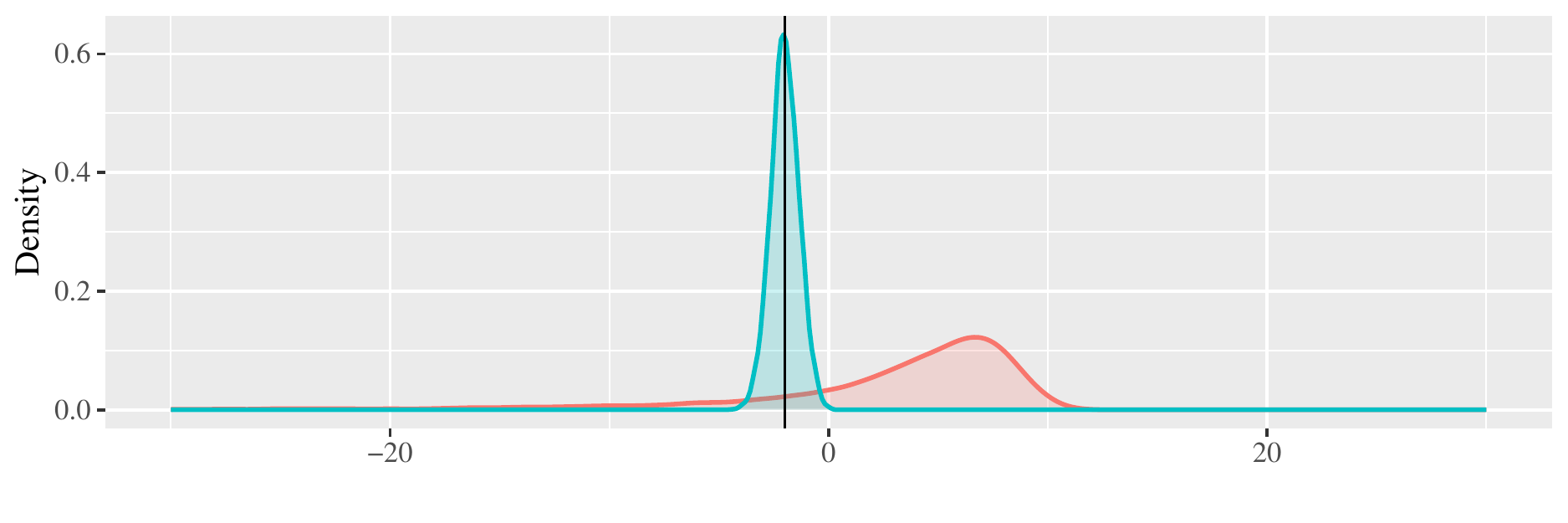}
\end{center}
\vspace{-5mm}
\caption{Sampling Distributions of IPW-Based $\hE[Z_i]'\widehat{\beta}_{-1}$ and NPW-Based $\hE[Z_i]'\widehat{\beta}_{1}$ Estimators}
\label{figure:ipw_vs_npw_example_ate_est_dist}
\vspace{-5mm}
\begin{center}
\includegraphics[scale=0.7]{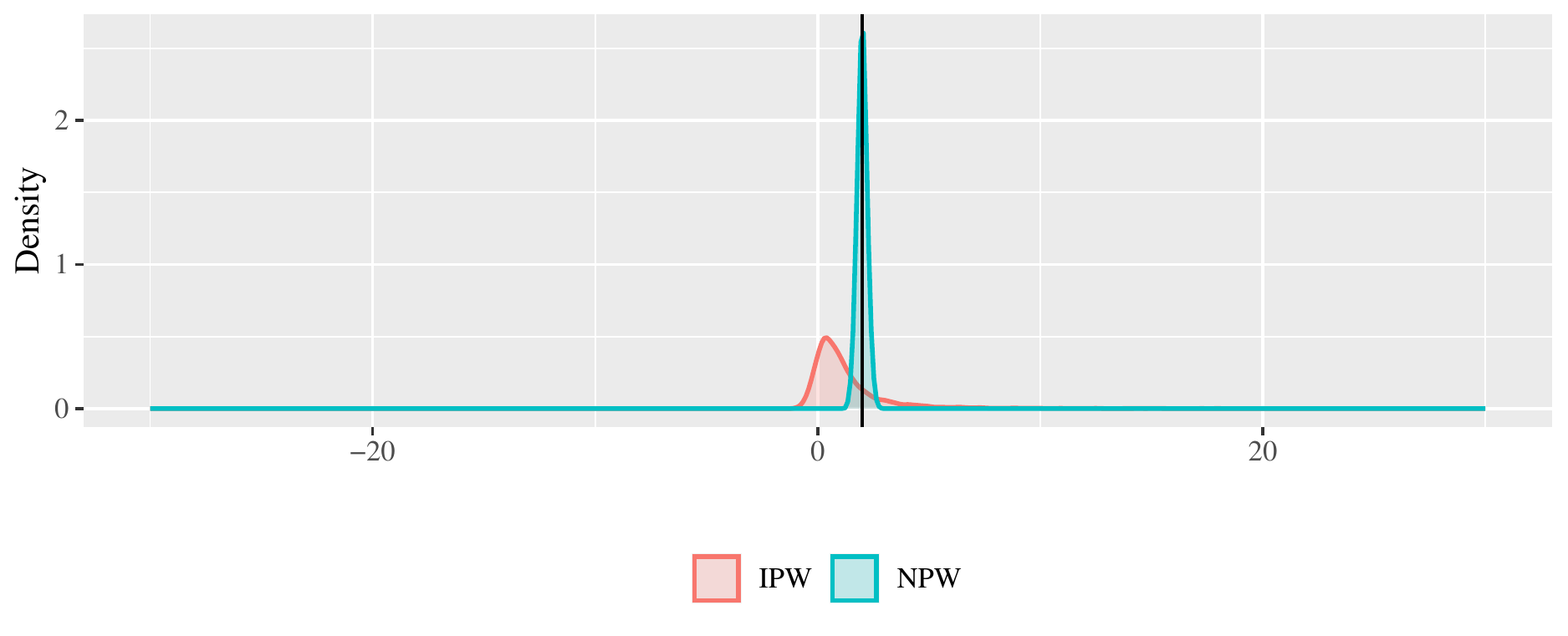}
\end{center}
\end{figure}
\vspace{-4mm}
\noindent \footnotesize Note: The above figures show the simulated densities of the sampling distributions of the IPW and NPW estimators of the parameters $(1,0)'\beta = 3$, $(0,1)'\beta = -2$, and $\E[Z]'\beta = 2$ (under the assumed data-generating process) in a setting with limited overlap. Black lines are used to represent the true values of these estimands. The true parameter coefficient values are marked using black lines. The simulation uses $n = 10^5$ and the following data-generating process: $X \sim \mathrm{Uniform}(0,1)$, $W \sim \mathrm{Bernoulli}(X^4)$, $(\upsilon_1, \upsilon_2) \sim \mathrm{Uniform}[(-2, 2)^2]$, $Y^*_0 = 10\,(1-X^4) + X^4\,\upsilon_1$, and $Y^*_1 = Y^*_0 + \tau(X) + 2\,\upsilon_2$, where $\tau(X) = \beta'Z$ with $\beta = (3, -2)$ and $Z = (1,X)$ so that $\tau^* = \E[Z]'\beta = 3 - 2\,\E[X] = 2$.

\begin{figure}[ht]
\caption{Contour Plot of a Sampling Distribution of the IPW Estimator $\widehat{\beta}_{-1}$}
\label{figure:ipw_joint_density_example}
\vspace{-5mm}
\begin{center}
\includegraphics[scale=0.9]{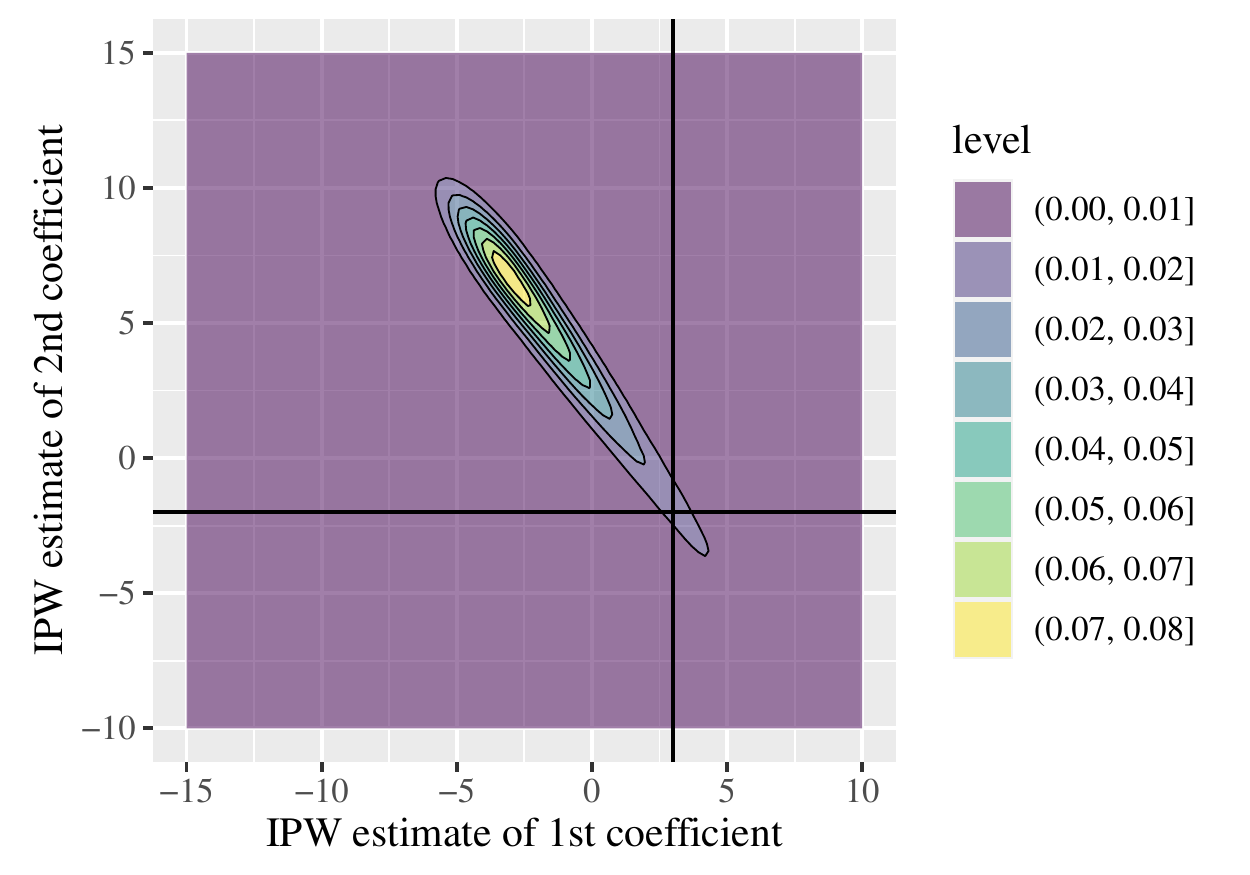}
\end{center}
\caption{Contour Plot of a Sampling Distribution of the NPW Estimator $\widehat{\beta}_{1}$}
\label{figure:npw_joint_density_example}
\vspace{-5mm}
\begin{center}
\includegraphics[scale=0.9]{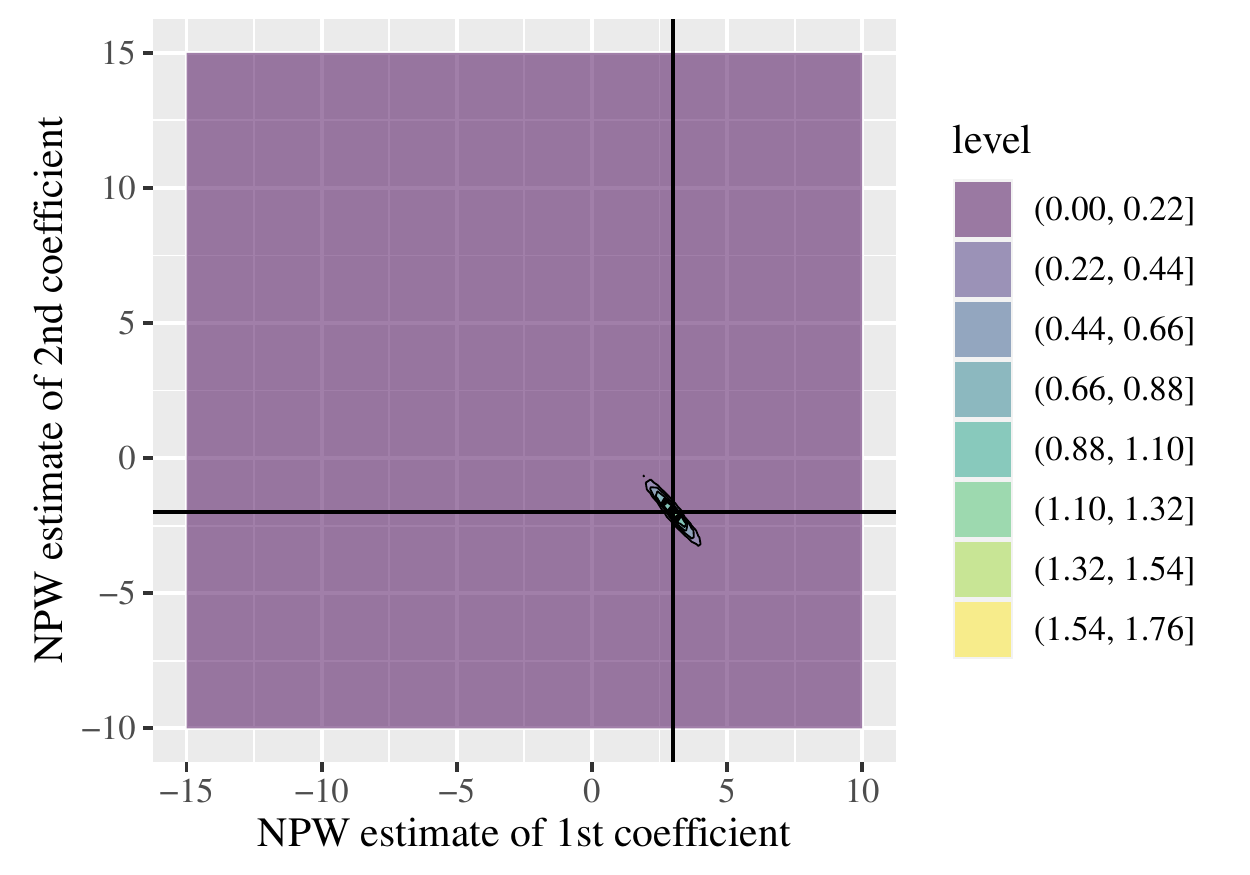}
\end{center}
\footnotesize Note: The above contour plots show the simulated joint densities of the sampling distributions of the IPW and NPW estimators in a setting with limited overlap. The horizontal and vertical axes are used for representing values of the first and second coefficient estimates of a two-dimensional coefficient parameter vector $\beta = (3, -2)$, respectively. The true parameter coefficient values are marked using black lines. The simulation uses $n = 10^5$ and the following data-generating process: $X \sim \mathrm{Uniform}(0,1)$, $W \sim \mathrm{Bernoulli}(X^4)$, $(\upsilon_1, \upsilon_2) \sim \mathrm{Uniform}[(-2, 2)^2]$, $Y^*_0 = 10\,(1-X^4) + X^4\,\upsilon_1$, and $Y^*_1 = Y^*_0 + \tau(X) + 2\,\upsilon_2$, where $\tau(X) = \beta'Z$ with $\beta = (3, -2)$ and $Z = (1,X)$ so that $\tau^* = 3 - 2\,\E[X] = 2$. 
\end{figure}

\begin{figure}[ht]
\caption{Distributions of Finite-Sample Estimators when $n = 50$ and  $\lambda_{1,0} = 1 - \lambda_{1,1} = 0.02$}
\label{figure:sampling_dist_n50_p2}
\vspace{-5mm}
\begin{center}
\includegraphics[scale=0.6]{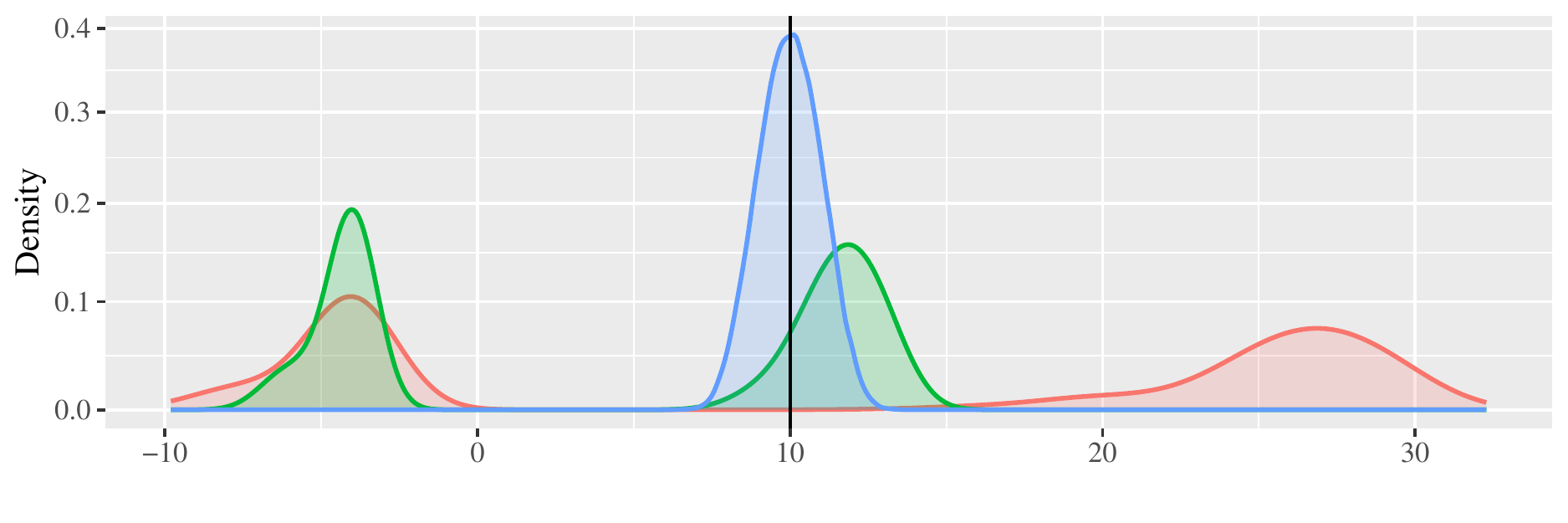}
\end{center}
\vspace{-7mm}
\caption{Distributions of Finite-Sample Estimators when $n = 50$ and  $\lambda_{1,0} = 1 - \lambda_{1,1} = 0.05$}
\vspace{-5mm}
\begin{center}
\includegraphics[scale=0.6]{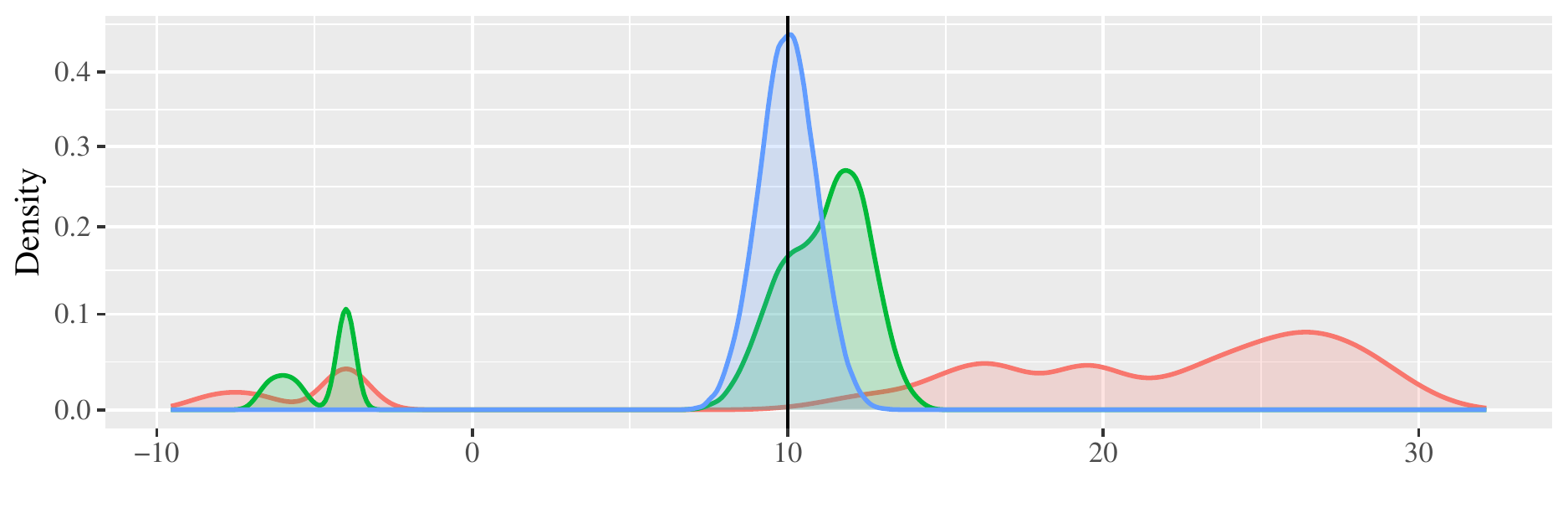}
\end{center}
\vspace{-7mm}
\caption{Distributions of Finite-Sample Estimators when $n = 50$ and  $\lambda_{1,0} = 1 - \lambda_{1,1} = 0.10$}
\vspace{-5mm}
\begin{center}
\includegraphics[scale=0.6]{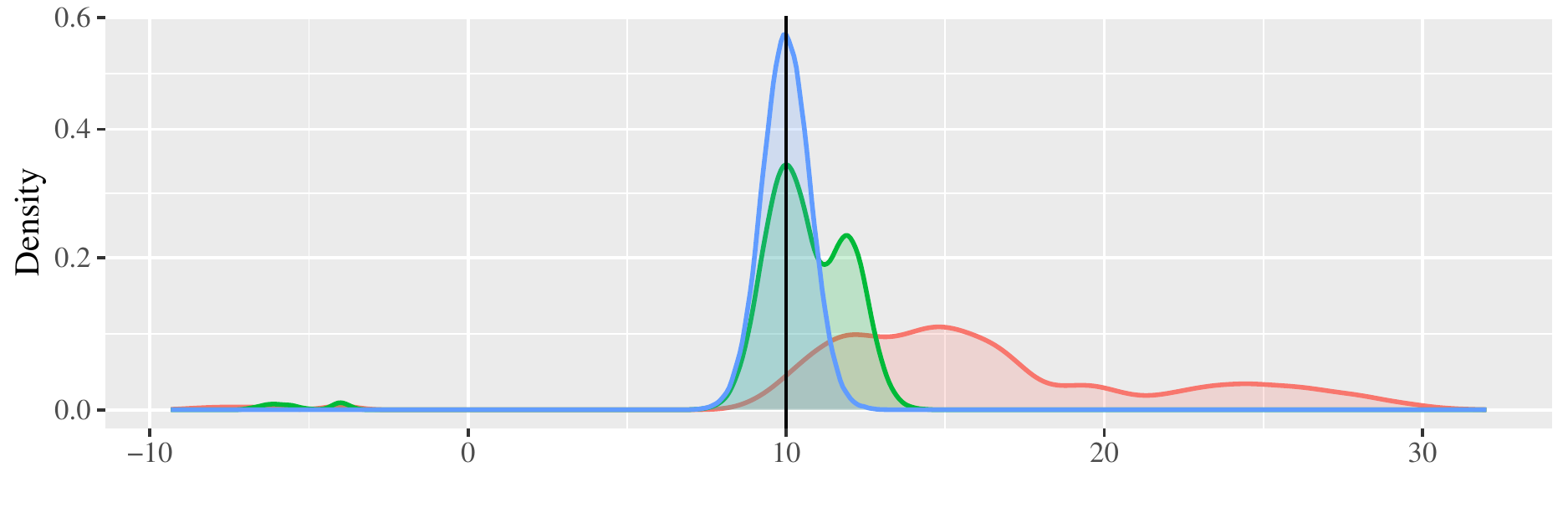}
\end{center}
\vspace{-7mm}
\caption{Distributions of Finite-Sample Estimators when $n = 50$ and  $\lambda_{1,0} = 1 - \lambda_{1,1} = 0.50$}
\vspace{-5mm}
\begin{center}
\includegraphics[scale=0.6]{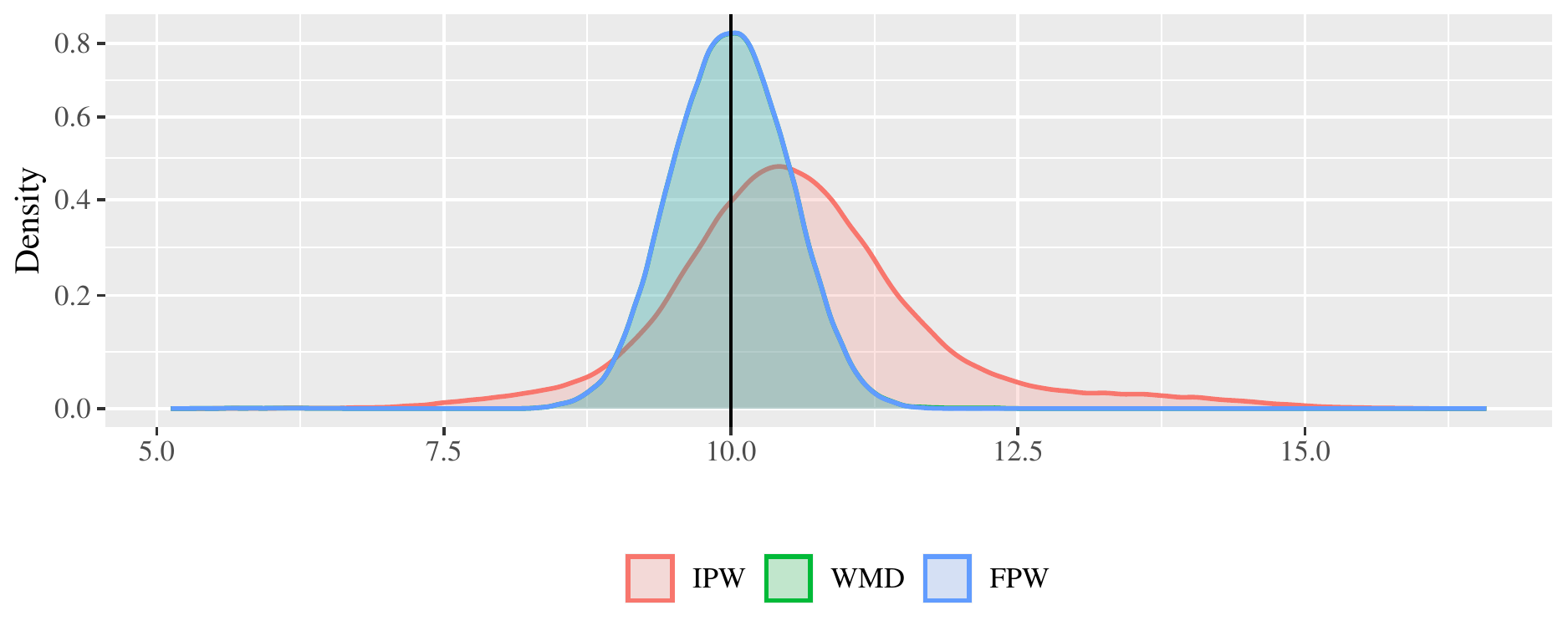}
\end{center}
\footnotesize Note: The above figures show the simulated densities of the sampling distributions of the IPW, WMD, and FPW estimators of the ATE $= 10$ for the specified nuisance parameters $(\lambda_{0,k},\lambda_{1,k})_{k \in \{0, 1\}}$ when $\mathcal{D}_n = (Y_i, X_i, W_i)_{i = 1}^n$ with $n = 50$, $W_i \in \mathbb{W} = \{0, 1\}$, and $X_i = 1\{i > 0.8\,n\}$ such that $Y_i = 10 + 2(1+X_i)\,\upsilon_{1,i} + W_i[10 + (1+2X_i)\,\upsilon_{2,i}]$, where $(\upsilon_{1,i},\upsilon_{2,i}) \sim \mathrm{Uniform}[(-1,1)^2]$, and $W_i \sim \mathrm{Bernoulli}(\lambda_{1,X_i})$, where $\lambda_{1,X_i} = 1 - \lambda_{0, X_i}$, for all $i \in \{1, \dots, n\}$. Although the FPW estimator is a set-estimator in general, it usually reduces to a point-estimator for the above choices of the nuisance parameters $\lambda_{1,0} = 1 - \lambda_{0, 0} = 1 - \lambda_{1,1} = \lambda_{0, 1}$. The black vertical lines mark the true ATE value $= 10$.
\end{figure}

\begin{figure}[ht]
\caption{Distributions of Finite-Sample Estimators when $n = 100$ and  $\lambda_{1,0} = 1 - \lambda_{1,1} = 0.02$}
\vspace{-5mm}
\begin{center}
\includegraphics[scale=0.6]{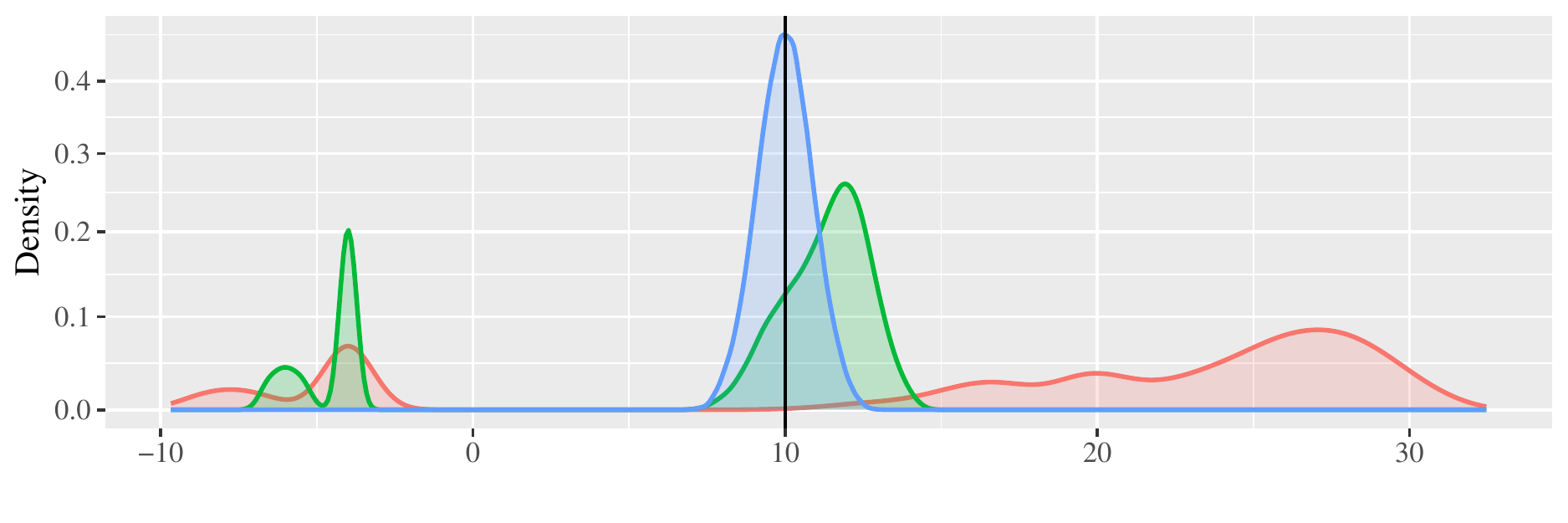}
\end{center}
\vspace{-7mm}
\caption{Distributions of Finite-Sample Estimators when $n = 100$ and  $\lambda_{1,0} = 1 - \lambda_{1,1} = 0.05$}
\vspace{-5mm}
\begin{center}
\includegraphics[scale=0.6]{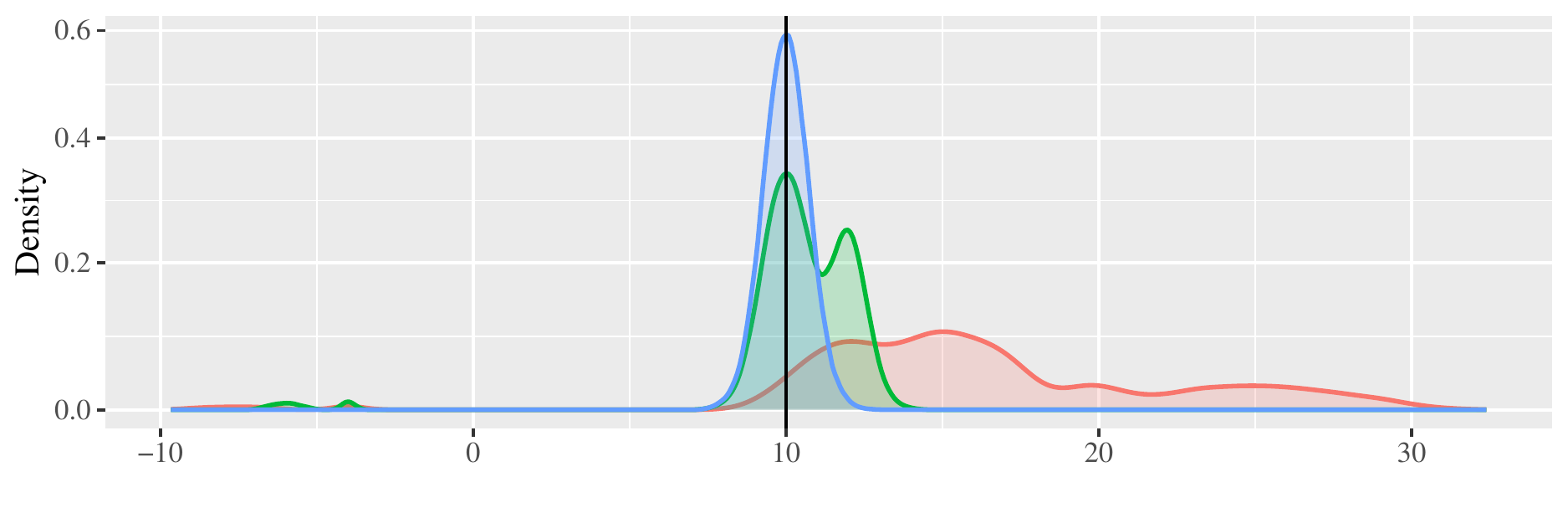}
\end{center}
\vspace{-7mm}
\caption{Distributions of Finite-Sample Estimators when $n = 100$ and  $\lambda_{1,0} = 1 - \lambda_{1,1} = 0.10$}
\vspace{-5mm}
\begin{center}
\includegraphics[scale=0.6]{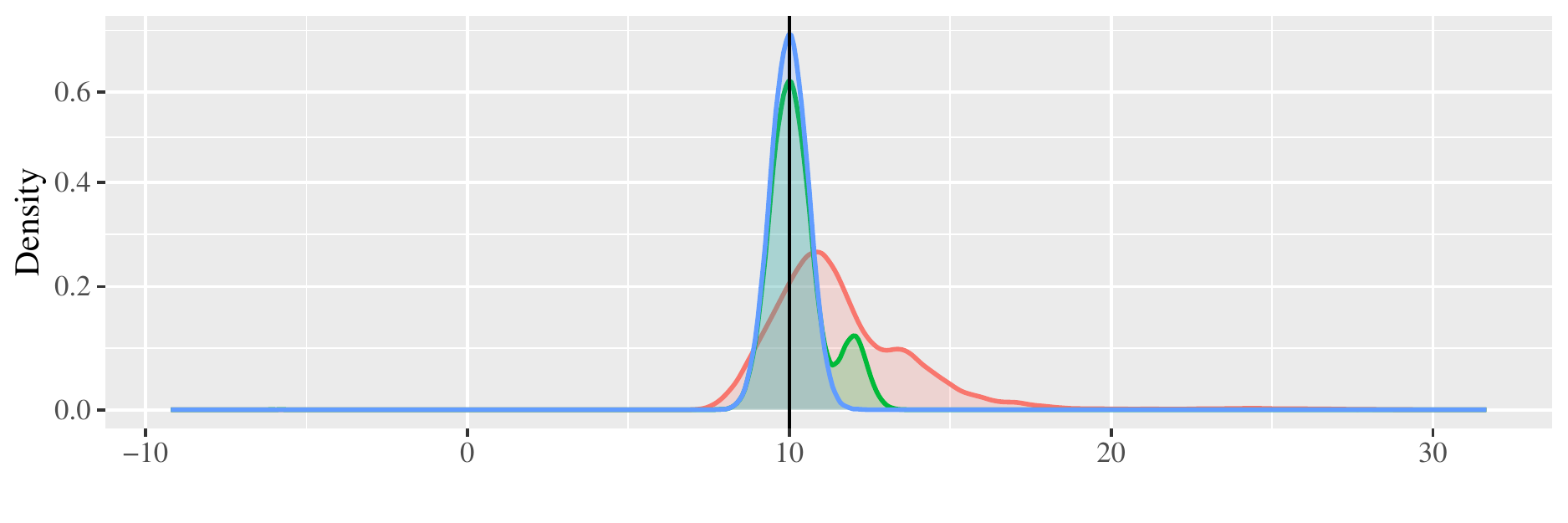}
\end{center}
\vspace{-7mm}
\caption{Distributions of Finite-Sample Estimators when $n = 100$ and  $\lambda_{1,0} = 1 - \lambda_{1,1} = 0.50$}
\vspace{-5mm}
\begin{center}
\includegraphics[scale=0.6]{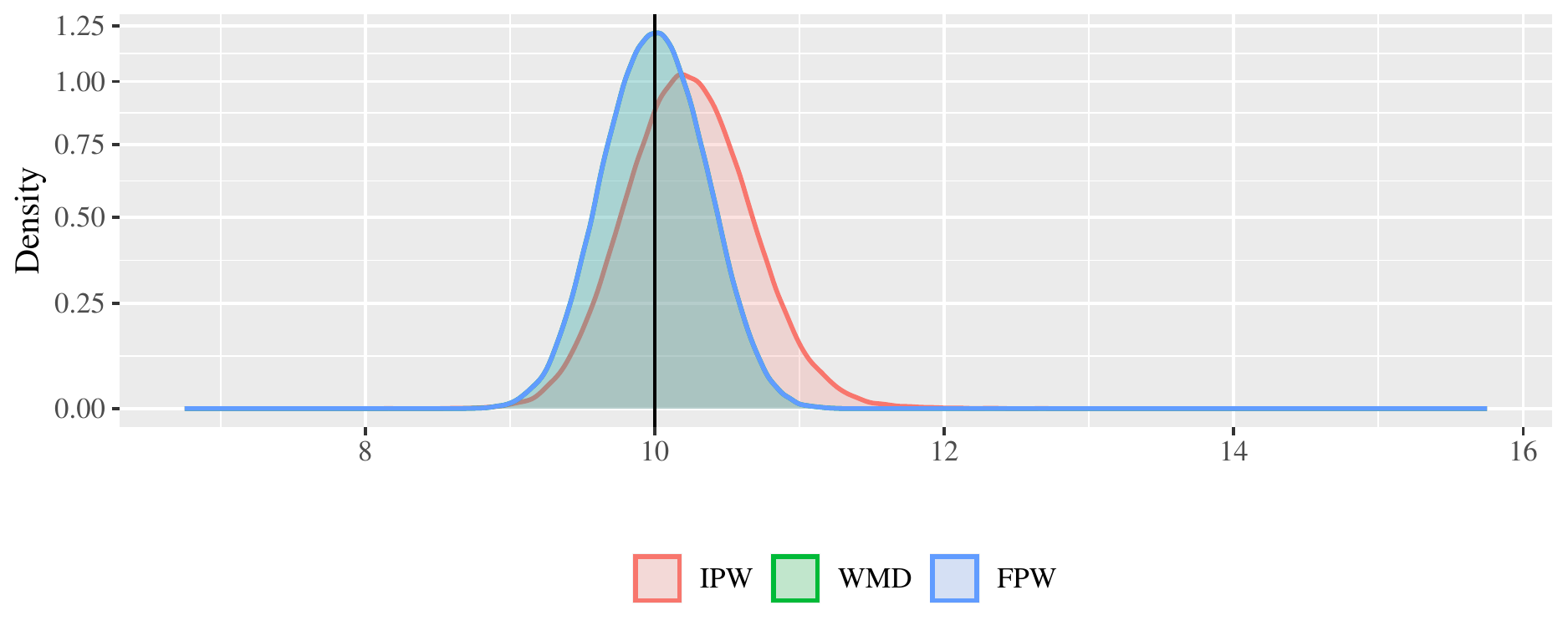}
\end{center}
\footnotesize Note: The above figures show the simulated densities of the sampling distributions of the IPW, WMD, and FPW estimators of the ATE $= 10$ for the specified nuisance parameters $(\lambda_{0,k},\lambda_{1,k})_{k \in \{0, 1\}}$ when $\mathcal{D}_n = (Y_i, X_i, W_i)_{i = 1}^n$ with $n = 100$, $W_i \in \mathbb{W} = \{0, 1\}$, and $X_i = 1\{i > 0.8\,n\}$ such that $Y_i = 10 + 2(1+X_i)\,\upsilon_{1,i} + W_i[10 + (1+2X_i)\,\upsilon_{2,i}]$, where $(\upsilon_{1,i},\upsilon_{2,i}) \sim \mathrm{Uniform}[(-1,1)^2]$, and $W_i \sim \mathrm{Bernoulli}(\lambda_{1,X_i})$, where $\lambda_{1,X_i} = 1 - \lambda_{0, X_i}$, for all $i \in \{1, \dots, n\}$. Although the FPW estimator is a set-estimator in general, it usually reduces to a point-estimator for the above choices of the nuisance parameters $\lambda_{1,0} = 1 - \lambda_{0, 0} = 1 - \lambda_{1,1} = \lambda_{0, 1}$. The black vertical lines mark the true ATE value $= 10$.
\end{figure}

\begin{figure}[ht]
\caption{Distributions of Finite-Sample Estimators when $n = 500$ and  $\lambda_{1,0} = 1 - \lambda_{1,1} = 0.002$}
\vspace{-5mm}
\begin{center}
\includegraphics[scale=0.6]{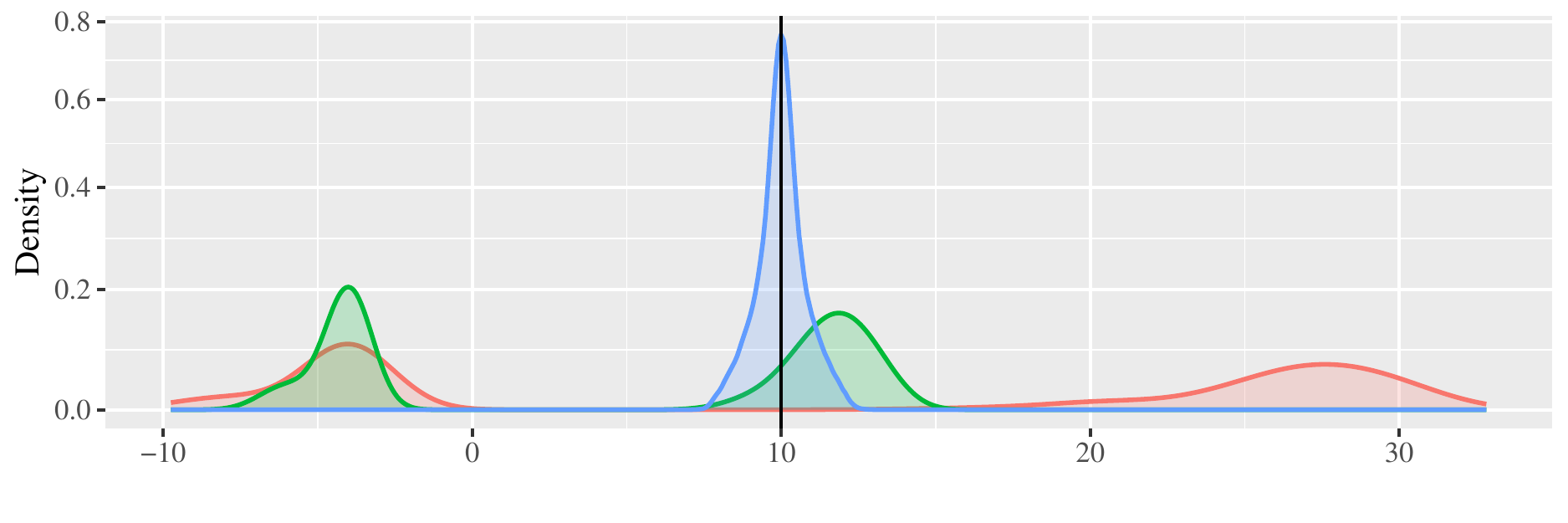}
\end{center}
\vspace{-7mm}
\caption{Distributions of Finite-Sample Estimators when $n = 500$ and  $\lambda_{1,0} = 1 - \lambda_{1,1} = 0.005$}
\vspace{-5mm}
\begin{center}
\includegraphics[scale=0.6]{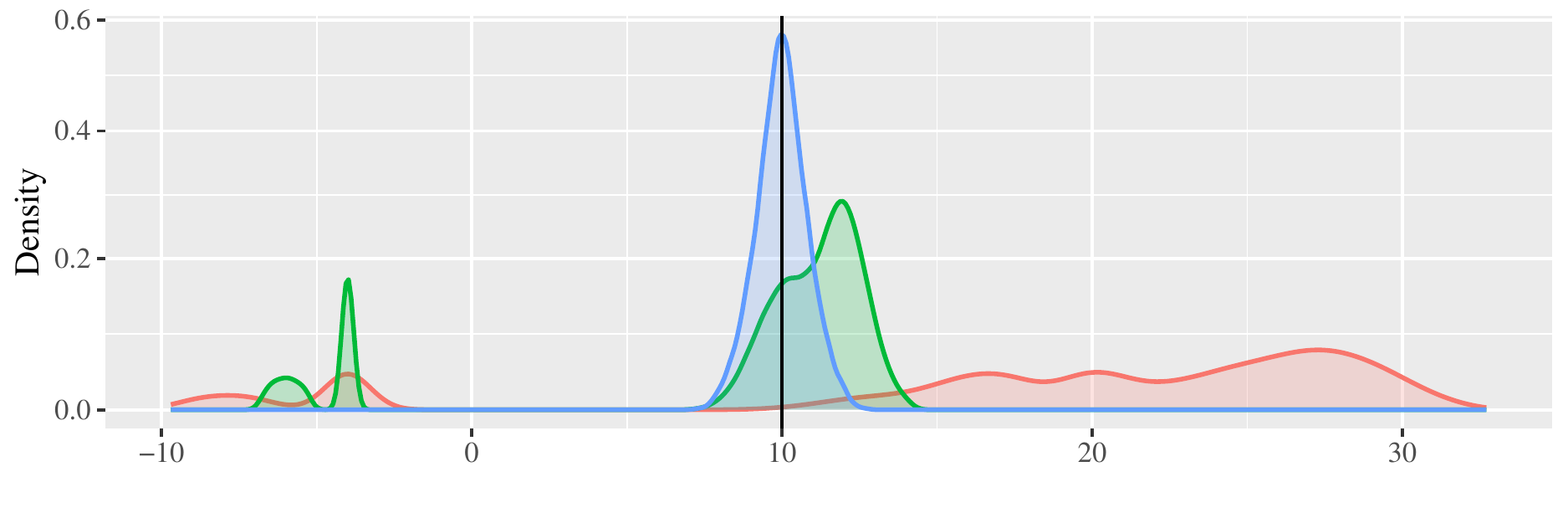}
\end{center}
\vspace{-7mm}
\caption{Distributions of Finite-Sample Estimators when $n = 500$ and  $\lambda_{1,0} = 1 - \lambda_{1,1} = 0.010$}
\vspace{-5mm}
\begin{center}
\includegraphics[scale=0.6]{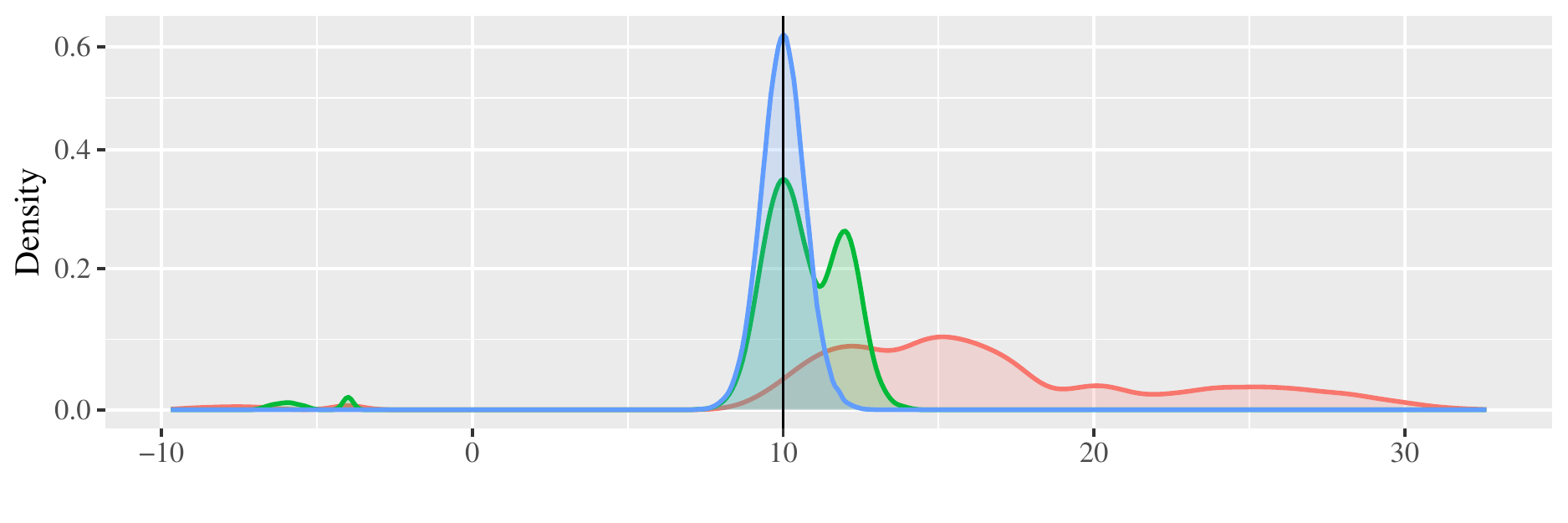}
\end{center}
\vspace{-7mm}
\caption{Distributions of Finite-Sample Estimators when $n = 500$ and  $\lambda_{1,0} = 1 - \lambda_{1,1} = 0.500$}
\label{figure:sampling_dist_n500_p50}
\vspace{-5mm}
\begin{center}
\includegraphics[scale=0.6]{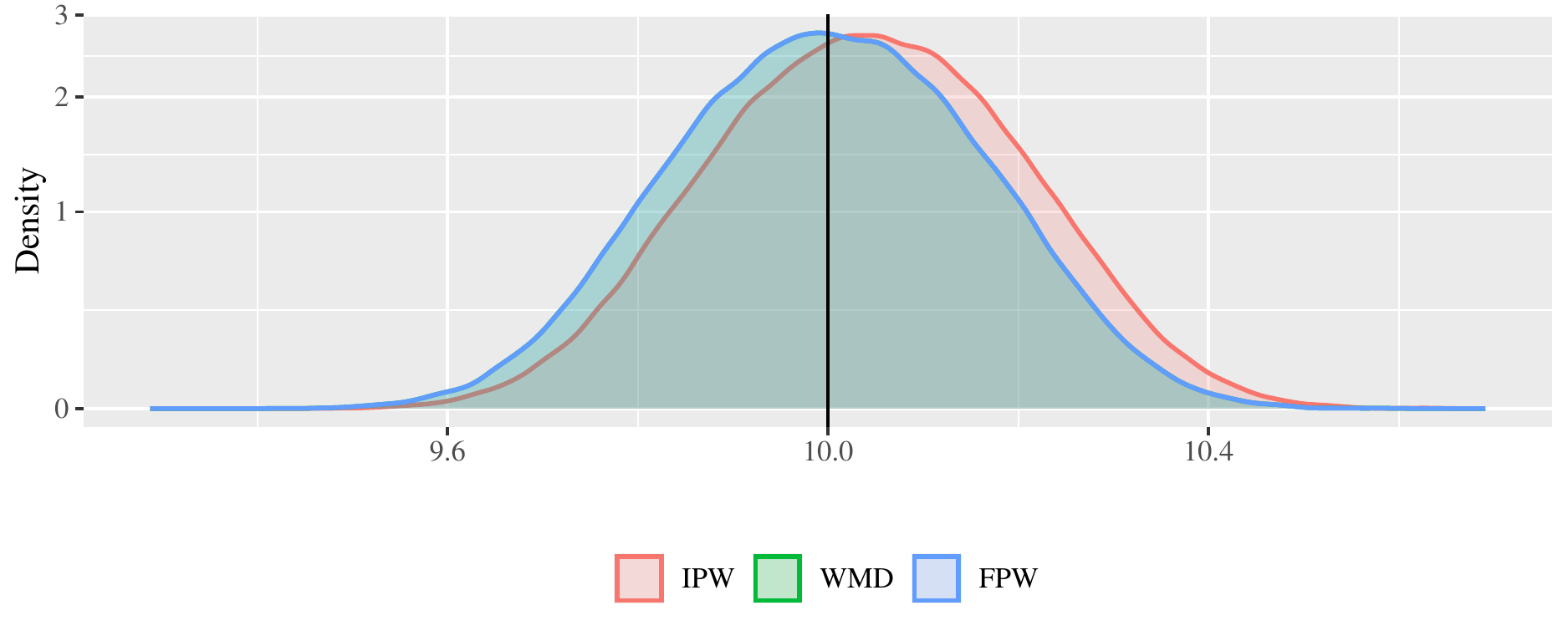}
\end{center}
\footnotesize Note: The above figures show the simulated densities of the sampling distributions of the IPW, WMD, and FPW estimators of the ATE $= 10$ for the specified nuisance parameters $(\lambda_{0,k},\lambda_{1,k})_{k \in \{0, 1\}}$ when $\mathcal{D}_n = (Y_i, X_i, W_i)_{i = 1}^n$ with $n = 500$, $W_i \in \mathbb{W} = \{0, 1\}$, and $X_i = 1\{i > 0.8\,n\}$ such that $Y_i = 10 + 2(1+X_i)\,\upsilon_{1,i} + W_i[10 + (1+2X_i)\,\upsilon_{2,i}]$, where $(\upsilon_{1,i},\upsilon_{2,i}) \sim \mathrm{Uniform}[(-1,1)^2]$, and $W_i \sim \mathrm{Bernoulli}(\lambda_{1,X_i})$, where $\lambda_{1,X_i} = 1 - \lambda_{0, X_i}$, for all $i \in \{1, \dots, n\}$. Although the FPW estimator is a set-estimator in general, it usually reduces to a point-estimator for the above choices of the nuisance parameters $\lambda_{1,0} = 1 - \lambda_{0, 0} = 1 - \lambda_{1,1} = \lambda_{0, 1}$. The black vertical lines mark the true ATE value $= 10$.
\end{figure}

\end{document}